\documentclass[11pt, letterpaper]{article}
\usepackage[pdftex]{graphicx}
\usepackage[utf8]{inputenc}
\usepackage[textfont=scriptsize]{caption,subfig}
\usepackage{latexsym}
\usepackage{textcomp}
\usepackage{tikz}
\usepackage[letterpaper]{geometry}
\usepackage{verbatim}
\usepackage[section]{placeins}
\usepackage{multirow}
\usepackage{authblk}

\usepackage{amssymb,amsfonts,amsmath}
\usepackage[figure,boxed,vlined,linesnumbered]{algorithm2e}

\graphicspath{{figures/}}
\newenvironment{proof}[1][Proof]
{
{\noindent \bf #1.}
}
{
\hfill \ensuremath{\Box} }

\newcommand{\tp}[1][\alpha]{\ensuremath{top\left(#1\right)}}

\newcommand{\old}[1]{}
\newcommand{\B}{\ensuremath{B}}

\newcommand{\ser}[1]{\ensuremath{\vec{#1}}}

\newcommand{\bfb}{\ensuremath{\stackrel{\scriptscriptstyle{\text{BFB}}}{\longrightarrow}}}

\renewcommand{\times}{}
\newcommand{\rev}[1]{\ensuremath{\bar{#1}}}
\newcommand{\rsigma}{\ensuremath{\rev{\sigma}}}
\newcommand{\revt}[1]{\ensuremath{\bar{\text{#1}}}}

\newcommand{\modTwo}[1]{\ensuremath{\textnormal{mod2}\left(#1\right)}}

\newcommand{\lesseq}[2]{\ensuremath{#1 \leq^{t} #2}}
\newcommand{\less}[2]{\ensuremath{#1 <^{t} #2}}

\newcommand{\np}{\ensuremath{\vec{n}\hspace{1pt}'}}
\newcounter{example}

\newtheorem{definition}{Def\,inition}
\newtheorem{claim}{Claim}

\title{Detecting Breakage Fusion Bridge cycles in tumor genomes---an algorithmic approach}
\author[1,*]{Shay Zakov}
\author[2,*]{Marcus Kinsella}
\author[1,2]{Vineet Bafna}
\affil[1]{Department of Computer Science and Engineering, University of California, San Diego}
\affil[2]{Bioinformatics and Systems Biology Program, University of California, San Diego}
\affil[*]{Joint first authorship}

\begin{document}

\maketitle

\begin{abstract}
Breakage-Fusion-Bridge (BFB) is a mechanism of genomic instability
characterized by the joining and subsequent tearing apart of sister chromatids.
When this process is repeated during multiple rounds of cell division, it leads
to patterns of copy number increases of chromosomal segments as well as fold-back
inversions where duplicated segments are arranged head-to-head. These
structural variations can then drive tumorigenesis.

BFB can be observed in progress using cytogenetic techniques, but generally
BFB must be inferred from data like microarrays or sequencing
collected after BFB has ceased. Making correct inferences
from this data is not straightforward, particularly given the complexity of
some cancer genomes and BFB's ability to generate a wide
range of rearrangement patterns.

Here we present algorithms to aid the interpretation of
evidence for BFB. We first
pose the BFB count vector problem: given a chromosome segmentation and segment copy numbers, decide whether BFB can yield a
chromosome with the given segment counts. We present the first
linear-time algorithm for the problem, improving a previous exponential-time
algorithm. We then combine this algorithm with fold-back inversions to develop
tests for BFB.
We show that, contingent on assumptions about
cancer genome evolution, count vectors and fold-back inversions are
sufficient evidence for detecting BFB. We apply the presented techniques to paired-end sequencing data from
pancreatic tumors and confirm a previous finding of BFB as well as identify
a new chromosomal region likely rearranged by BFB cycles, demonstrating
the practicality of our approach.
\end{abstract}

\section{Introduction}
\label{sec:introduction}
Genomic instability allows cells to acquire the functional
capabilities needed to become cancerous~\cite{HallmarksOfCancer}, so
understanding the origin and operation of genomic instability is
crucial to finding effective treatments for cancer.  Numerous mechanisms of
genomic instability have been proposed~\cite{MechChangeCopyNumber},
including the faulty repair of double-stranded DNA breaks by
recombination or end-joining and polymerase hopping caused by
replication fork collapse~\cite{PolyHopping}.  These mechanisms are
generally not directly observable, so their elucidation requires the
deciphering of often subtle clues after genomic instability has
ceased.

In contrast, the breakage-fusion-bridge (BFB) mechanism creates gross
chromosomal abnormalities that can be seen
in progress using methods that have been available for decades~\cite{BM_BFB_1941}.
BFB begins when a chromosome loses a telomere (Figs.~\ref{fig:overview}a, \ref{fig:overview}b). Then during
replication, the two sister chromatids of the telomere-lacking
chromosome fuse together (Figs.~\ref{fig:overview}c, \ref{fig:overview}d). During anaphase, as the
centromeres of the chromosome migrate to opposite ends of the cell (Fig.~\ref{fig:overview}e),
the fused chromatids are torn apart (Fig.~\ref{fig:overview}f). Each daughter cell receives a
chromosome missing a telomere, and the cycle can begin again. As this process
repeats, it can lead to the rapid accumulation of amplifications and
rearrangements that facilitates the transition to
malignancy~\cite{NatGen_review}.

This process produces several plainly identifiable cytogenetic signatures such
as anaphase bridges and dicentric chromosomes.
However, as cancer genomics has shifted to high-throughput techniques, the
signatures of BFB have become less clear. Methods like microarrays and
sequencing do not allow for direct observation of BFB; instead BFB is now
similar to other mechanisms of instability in that it must be inferred by
finding its footprint in complex data.

Multiple groups have begun to address the problem of finding evidence for BFB
in high-throughput data. For example, Bignell \textit{et al.} found a pattern
of inversions and exponentially increasing copy numbers ``[bearing] all the
architectural hallmarks at the sequence level'' of BFB~\cite{Bignell}. Kitada and Yamasaki
found a pattern of copy counts and segment organization consistent with a
particular set of BFB cycles~\cite{Lung_BFB}. Hillmer \textit{et al.} used paired-end
sequencing to find patterns of inversions and amplification explainable by
BFB~\cite{Hillmer}.

The procedures of these investigators, among others~\cite{Lim, Hicks,Selvarajah},
share an element in common: they determine whether a particular observation is
consistent with or could be explained by BFB. While this is helpful, it does
not on its own allow one to infer whether or not BFB occurred. Indeed, in a previous
work~\cite{first_bfb_paper} we examined short patterns of copy number increases consisting of five or
six chromosome segments. We found that most such patterns, whether produced by
BFB or not, were consistent or nearly consistent with BFB. Thus, finding that
such a pattern was consistent with BFB would only be weak evidence that it
had been produced by BFB. This finding highlights the need for a rigorous and
systematic approach to the interpretation of modern data for BFB in order to
avoid being misled by the complexity of cancer genomes and the BFB mechanism
itself.

Here we present a framework for interpreting high-throughput data for
signatures of BFB. We incorporate observations of breakpoints as well as
copy numbers to create a scoring scheme for chromosomes. Through simulations,
we find appropriate threshold scores for labeling a chromosome as having
undergone BFB based on varying models of cancer genome evolution and tolerances
for error. This framework complements the work of previous groups by not only
finding breakpoint and copy number patterns consistent with BFB but also
showing under what assumptions they are more likely to be observed if BFB
occurred than if it did not.

The key technical contribution that underlies our scoring scheme is a new, fast
algorithm for determining if a given pattern of copy counts is consistent with
BFB. This algorithm is related to a previously described
algorithm~\cite{first_bfb_paper} in that it takes advantage
of a distinctive feature of BFB: when fused chromatids are torn apart, they may
not tear at the site of fusion. This yields chromosomes with either a terminal
deletion or a terminal inverted duplication. When a chromosome undergoes this
process repeatedly, it results in particular patterns of copy number increases.
The running time of the earlier algorithm grew exponentially with the amount
of amplification and the number of segments in a copy number pattern. This
greatly narrowed the scope of copy number patterns that could be investigated.
This was particularly limiting because it appeared that copy number patterns
with more segments would be more useful for identifying BFB, but these
patterns could not be evaluated in a reasonable amount of time with the
previous method. The new algorithm presented here is linear time and therefore
allows complex copy number patterns to be checked in a trivial amount of time.

We begin by describing the kinds of high-throughput data that can provide
evidence for BFB. We then proceed to lay out some formalizations needed to
precisely describe scoring methods of samples based on BFB evidence implied from such data. Next, we define related computational problems, followed by
an outline of algorithms for these problems. In the results section,
we detail the simulations we used to measure the performance of our
scoring system for BFB. Based on simulation parameters, we find false and
true positive rates for different BFB signatures. We apply our methods to two
datasets. The first is copy number data from 746 cancer cell lines~\cite{CellLines}. We find
three chromosomes that have long copy number patterns consistent with BFB, but
the false positive rates from our simulations suggest that these may be false
discoveries. We also examine paired-end sequencing data from pancreatic cancers~\cite{pancreas}.
We find two chromosomes that likely have undergone BFB, one that was identified
by the original publishers of the data and one novel finding.

\section{High-throughput evidence for BFB}
We consider two experimental sources for evidence for BFB: microarrays and
sequencing. Microarrays allow for the estimation of the copy number of
segments of a chromosome by measuring probe intensities~\cite{microarray_methods}.
Sequencing also yields
copy number estimates by measuring depth of sequence coverage~\cite{sequencing_CNVs}.
In addition, if the sequencing uses paired-end reads and is performed on the
whole genome rather than, say, the exome, it can reveal genomic
breakpoints where different portions of the genome are unexpectedly adjacent.
This is generally the extent of evidence available from either technique.
Sequencing does not allow for a full reconstruction of a rearranged chromosome,
as the repetitive nature of the genome leads to multiple alternative assemblies.
Neither method can resolve segment copy numbers by orientation, so copy
numbers from both forward and reversed chromosome segments are summed.
Nevertheless, BFB should leave its signature in both breakpoints and copy
counts, and we examine each in turn.

\subsection{Breakpoints}
During BFB, the telomere-lacking sister chromatids are fused together. This
causes the ends of the sister chromatids to become adjacent but in opposite
orientations (see Fig.~\ref{fig:overview}d). This adjacency is unlikely to be
disrupted by subsequent BFB cycles and will remain in the final sequence as
two duplicated segments arranged head-to-head. If
the chromosome is paired-end sequenced, the rearrangement will appear as two
ends that map very near each other but in opposite orientations. This type of
rearrangement has been termed a ``fold-back inversion''~\cite{pancreas}, and
regions of a chromosome rearranged by BFB should have an enrichment of these
fold-back inversions.
Reliable indications for fold-back inversions may or may not be available, depending on the type of experiment and its intensity.

\subsection{Copy counts}
Each BFB cycle duplicates some telomeric portion of the
chromosome undergoing BFB. These repeated duplications should lead to certain
characteristic copy number patterns, which are the signature of BFB in copy
number data. We would like to evaluate copy numbers observed from microarrays
or sequencing and determine if the copy numbers contain the footprint of BFB.
Previous groups have searched for such a footprint by manually inspecting
copy number data and searching for a set of BFB cycles that could produce the
observed copy numbers~\cite{Bignell,Lung_BFB}. This approach is challenging and
labor intensive, but developing a more general approach turns out to be rather
difficult. A key technical contribution of this paper is the development of
efficient algorithms to evaluate copy counts for consistency with BFB.

\subsection{Formalizing BFB}
Creating an efficient method for evaluating copy numbers requires some
formalization, so we begin with some definitions and basic results.

We represent a chromosome as a string ABC\ldots, where each letter
corresponds to a contiguous segment of the chromosome. For example, the
string ABCD would symbolize a chromosome arm composed of four segments, where
A is the segment nearest the centromere. More generally, we use
$\sigma_l$ for the $l$-th segment in a chromosome. So, ABCD could be
written $\sigma_1\sigma_2\sigma_3\sigma_4$. A bar notation, $\rsigma$, is used to
signify that a segment is reversed. Greek letters $\alpha, \beta, \gamma, \rho$
denote concatenations of chromosomal segments, and a bar will again mean that the
concatenation is reversed. For example if
$\alpha = \sigma_1 \sigma_3 \rsigma_2$,
$\rev{\alpha} = \sigma_2 \rsigma_3 \rsigma_1$.
An empty string is denoted by $\varepsilon$.

Consider the following BFB cycle on a chromosome $\text{X} \bullet \text{ABCD}$,
where $\bullet$ represents the centromere, X is one chromosomal arm,
and ABCD is the 4-segmented other chromosomal arm which has lost a telomere.
The cycle starts with the duplication of the chromosome into two sister
chromatids and their fusion at the ends of the `D' segments, generating the
dicentric chromosome $\text{X} \bullet \text{ABCD}\revt{D}\revt{C}\revt{B}\revt{A} \bullet \revt{X}$.
During anaphase, the two centromeres migrate to opposite poles of the cell and
a breakage of the dicentric chromosome occurs between the centromeres, say
between $\revt{D}$ and $\revt{C}$, providing one daughter cell with a chromosome
with an inverted suffix, $\text{X} \bullet \text{ABCD}\textbf{\revt{D}}$, and
another daughter cell with the trimmed chromosome $\text{X} \bullet \text{ABC}$
(chromosomes $\revt{C}\revt{B}\revt{A} \bullet \revt{X}$ and
$\text{X} \bullet \text{ABC}$ are equivalent). The now
amplified segment D in the first daughter cell may confer some proliferative
advantage, causing its descendants to increase
in frequency. The daughter cells also lack
a telomere on one chromosome arm and therefore may undergo additional BFB cycles.
One possible subsequent cycle could, for example, cause an inverted
duplication of the suffix $\text{CD}\revt{D}$, yielding the chromosome
$\text{X} \bullet \text{ABCD}\revt{D}\textbf{D\revt{D}\revt{C}}$. As these
BFB cycles continue, the count of segments on the modified chromosome arm
can increase significantly.

The notation $\alpha \bfb \beta$ will be used for indicating that the string
$\beta$ can be obtained by applying 0 or more BFB cycles over the string
$\alpha$, as formally described in Definition~\ref{def:bfb}.
\begin{definition}
\label{def:bfb}
For two strings $\alpha, \beta$, say that $\alpha \bfb \beta$ if
$\beta = \alpha$, or $\alpha = \rho \gamma$ for some strings
$\rho, \gamma$ such that $\gamma \neq \varepsilon$, and
$\rho \gamma \rev{\gamma} \bfb \beta$.
\end{definition}

We say that $\beta$ is an \emph{$l$-BFB string} if for some consecutive
chromosomal region $\alpha = \sigma_l \sigma_{l+1} \ldots $ starting at the
$l$-th segment $\sigma_l$, $\alpha
\bfb \beta$. Say that $\beta$ is a \emph{BFB
  string} if it is an $l$-BFB string for some $l$.
As examples, $\text{CDE} = \sigma_3 \sigma_4 \sigma_5$ is a 3-BFB string, and so are
CDE\textbf{\revt{E}} and CDE\revt{E}\textbf{E\revt{E}\revt{D}}.
The empty string $\varepsilon$ is considered an $l$-BFB string for every integer $l > 0$.

Denote by $\ser{n}(\alpha) = [n_1, n_2, \ldots, n_k]$ the \emph{count vector}
of $\alpha$, where $\alpha$ represents a modified chromosomal arm
$\sigma_1 \sigma_2 \ldots \sigma_k$ with $k$ segments, and $n_l$ is the count
of occurrences (or \emph{copy number}) of $\sigma_l$ and $\rsigma_l$ in $\alpha$.
For example,
for $\alpha = \text{BCD\revt{D}\revt{C}C}$, $\ser{n}(\alpha) = [0, 1, 3, 2]$.
Say that a vector $\vec{n}$
is a \emph{BFB count vector} if there exists some 1-BFB string $\alpha$ such
that $\ser{n} = \ser{n}(\alpha)$.

\subsection{Handling experimental imprecision}
Experimental methods do not provide the precise and accurate
copy number of a given chromosome segment. Instead, some measurement error is
expected. Moreover, in a cancer genome, it is plausible that a region
undergoing BFB may also be rearranged by other mechanisms. So when we evaluate
a count vector for consistency with BFB, we must also consider whether the
count vector is ``nearly'' consistent with BFB.

For this, we define a distance measure $\delta$ between count vectors, where
$\delta\left(\ser{n}, \np\right)$ reflects a penalty for assuming that the
real copy counts are $\np$ while the measured counts are $\vec{n}$.
We have implemented such a distance measure based on the
\emph{Poisson likelihood} of the observation, as follows:
Let $\Pr(n | n') = \frac{{n'}^n e^{-n'}}{n!}$ be the Poisson probability of
measuring a copy number $n$, given that the segment's true copy number is $n'$.
Assuming measurement errors are independent, the probability for measuring a
count vector $\vec{n} = [n_1, n_2, \ldots, n_k]$, where the true counts are
$\np = [n'_1, n'_2, \ldots, n'_k]$ is given by
$
\displaystyle{\Pr(\vec{n} | \np) = \prod_{1 \leq i \leq k} \Pr(n_k | n'_k).}
$
Define the \emph{distance} of $\vec{n}$ from $\np$ by
$$
\delta(\vec{n}, \np) = 1- \frac{\Pr(\vec{n} | \np)}{\Pr(\np | \np)}
$$

For every pair of count vectors $\vec{n}$ and $\np$ of the same length,
$0 \leq \delta(\vec{n}, \np) < 1$, being closer to 0 the greater is the
similarity between $\vec{n}$ and $\np$.

\subsection{The BFB Count Vector Problem}
With these definitions, we can now precisely pose a set of problems that need to
be solved to evaluate copy number patterns for consistency with BFB:

\vspace{5pt}
\noindent
\textbf{BFB count vector problem variants}

\noindent
\textbf{Input:} \emph{a count vector $\ser{n} = [n_1, n_2, \ldots, n_k]$.}%
\begin{itemize}
	\item \textbf{The decision variant:} \emph{decide if $\ser{n}$ is a BFB count vector.}
	\item \textbf{The search variant:} \emph{if $\ser{n}$ is a BFB count vector, find a BFB string $\alpha$ such that $\ser{n} = \ser{n}(\alpha)$.}
	\item \textbf{The distance variant:} \emph{Identify a BFB
            count vector $\np$ such that $\delta\left(\ser{n},
              \np\right)$ is minimized. Output $\delta$.}
\end{itemize}

\section{Outline of the BFB Count Vector Algorithms}

We defer the full details of the algorithms we have developed to the
accompanying Supporting Information (SI) document, presenting here only
essential properties of BFB strings and some intuition of how to incorporate
these properties in algorithms for BFB count vector problems. We focus on the
search variant of the problem, where the goal of the algorithm is to output a
BFB string $\alpha$ consistent with the input counts, if such a string exists.

\subsection{Properties of BFB palindromes}

Call an $l$-BFB string $\beta$ of the form $\beta = \alpha \rev{\alpha}$ an
\emph{$l$-BFB palindrome}\footnote{We assume that genomic segments $\sigma$
satisfy $\sigma \neq \rev{\sigma}$, therefore strings of the form
$\alpha \sigma \rev{\alpha}$ will not be considered palindromes.}.
For an $l$-BFB string $\alpha$, the string $\beta = \alpha \rev{\alpha}$ is an
$l$-BFB palindrome by definition (choosing $\rho = \varepsilon$ and
$\gamma = \alpha$ in Definition~\ref{def:bfb}).
In~\cite{first_bfb_paper}, it was shown that every prefix of a BFB string is
itself a BFB string, thus, for an $l$-BFB palindrome
$\beta = \alpha \rev{\alpha}$, the prefix $\alpha$ of $\beta$ is also an
$l$-BFB string. Hence, it follows that $\alpha$ is an $l$-BFB string if and
only if $\beta = \alpha \rev{\alpha}$ is an $l$-BFB palindrome.
For a BFB string $\alpha$ with $\vec{n}(\alpha) = [n_1, n_2, \ldots, n_k]$ and
a corresponding BFB palindrome $\beta = \alpha \rev{\alpha}$, we have that
$\vec{n}(\beta) = 2\vec{n}(\alpha) = [2n_1, 2n_2, \ldots, 2n_k]$. Thus, a
count vector $\vec{n}$ is a BFB count vector if and only if there is a 1-BFB
palindrome $\beta$ such that $\vec{n}(\beta) = 2\vec{n}$. Considering BFB
palindromes instead of BFB strings will facilitate the algorithm description.

Define an \emph{$l$-block} as a palindrome of the form
$\beta = \sigma_l \beta' \rsigma_l$, where $\beta'$ is an $(l+1)$-BFB
palindrome. For example, from the 4-BFB palindromes $\beta'_1 =
\text{DE\revt{E}\revt{D}DE\revt{E}\revt{D}}$ and $\beta'_2 = \varepsilon$, we
can produce the 3-blocks $\beta_1 = \sigma_3 \beta'_1 \rsigma_3 =
\text{CDE\revt{E}\revt{D}DE\revt{E}\revt{D}\revt{C}}$ and
$\beta_2 = \sigma_3 \beta'_2 \rsigma_3 = \text{C\revt{C}}$. It may be asserted
that an $l$-block is a special case of an $l$-BFB palindrome. Next, we show how
$l$-BFB palindromes may be decomposed into $l$-block substrings.

For a string $\alpha \neq \varepsilon$, denote by $\tp$ the maximum integer $t$
such that $\sigma_t$ or $\rsigma_t$ occur in $\alpha$, and define
$\tp[\varepsilon] = 0$.  For two strings $\alpha$
and $\beta$, say that $\lesseq{\alpha}{\beta}$ if $\tp[\alpha] \leq
\tp[\beta]$, and that $\less{\alpha}{\beta}$ if $\tp[\alpha] <
\tp[\beta]$. For example, for $\alpha = \text{AB}$ and $\beta =
\text{ABCD}\revt{D}\revt{C}$, $\tp = 2$ and $\tp[\beta] = 4$, therefore
$\less{\alpha}{\beta}$.

\begin{definition}
\label{def:convexed}
A string $\alpha$ is a \emph{convexed $l$-palindrome} if
$\alpha = \varepsilon$, or $\alpha = \gamma \beta \gamma$ such that $\gamma$ is
a convexed $l$-palindrome, $\beta$ is an $l$-BFB palindrome, and
$\gamma <^t \beta$.
\end{definition}

While every $l$-BFB palindrome $\alpha$ is also a convexed $l$-palindromes
(since $\alpha = \varepsilon \alpha \varepsilon$), not every convexed
$l$-palindromes is a valid BFB string. For example, $\alpha =
\text{A\revt{A}AB\revt{B}\revt{A}A\revt{A}}$ is a convexed 1-palindromes
(choosing $\gamma = \text{A\revt{A}}$, $\beta = \text{AB\revt{B}\revt{A}}$),
yet it is not a 1-BFB string.
Instead, we have the following claim, proven in the SI document:

\begin{claim}
  A string $\alpha$ is an $l$-BFB palindrome if and only if $\alpha =
  \varepsilon$, $\alpha$ is an $l$-block, or $\alpha = \beta \gamma
  \beta$, such that $\beta$ is an $l$-BFB palindrome, $\gamma$ is a convexed
  $l$-palindrome, and $\lesseq{\gamma}{\beta}$.
\label{clm:lPalindrome}
\end{claim}

 From Definition~\ref{def:convexed} and Claim~\ref{clm:lPalindrome}, it follows
that an $l$-BFB palindrome $\alpha$ is a palindromic concatenation of
$l$-blocks. In addition, for the total count $2 n_l$ of $\sigma_l$ and
$\rsigma_l$ in $\alpha$, $\alpha$ contains exactly $n_l$ $l$-blocks, where each
block contains one occurrence of $\sigma_l$ and one occurrence of $\rsigma_l$.
When $n_l$ is even, $\alpha$ is of the form $\alpha =
\beta_1 \beta_2 \ldots \beta_{\frac{n_l}{2}-1} \beta_{\frac{n_l}{2}}
\beta_{\frac{n_l}{2}} \beta_{\frac{n_l}{2}-1} \ldots \beta_2 \beta_1$, each
$\beta_i$ is an $l$-block.
When $n_l$ is odd, $\alpha$ is of the form $\alpha =
\beta_1 \beta_2 \ldots \beta_{\left\lfloor \frac{n_l}{2}\right\rfloor}
\beta_{\left\lfloor \frac{n_l}{2}\right\rfloor + 1}
\beta_{\left\lfloor \frac{n_l}{2}\right\rfloor} \ldots \beta_2 \beta_1$. In the
latter case, say that $\beta_{\left\lfloor \frac{n_l}{2}\right\rfloor+1}$ is
the \emph{center} of $\alpha$, where in the former case say that the
\emph{center} of $\alpha$ is $\varepsilon$. Note that every $l$-block $\beta$
appearing in $\alpha$ and different from its center occurs an even number of
times in $\alpha$. If the center of $\alpha$ is an $l$-block, this particular
block is the only block which appears an odd number of times in $\alpha$, while
if it is an empty string then no block appears an odd number of times in
$\alpha$.

Now, let $\beta$ be a 1-BFB palindrome with a count vector
$\vec{n}(\beta) = 2\vec{n} = [2n_1, 2n_2, \ldots, 2n_k]$.
It is helpful to depict $\beta$ so that each character $\sigma_l$ is
at its own layer $l$, increasing with increasing $l$, as shown in
Fig.~\ref{fig:layers}a. As $\beta$ is a concatenation of $1$-blocks, we can
consider the collection $\B^1 = \{m_1 \beta_1, m_2 \beta_2, \ldots, m_q \beta_q\}$
of these blocks, where each count $m_i$ is the number of distinct repeats of
$\beta_i$ in $\beta$. For example, for the string in Fig.~\ref{fig:layers}a,
$\B^1 = \{2 \beta_1, \beta_2, 2 \beta_3, 4 \beta_4\}$, where
$\left|\B^1\right| = n_1 = 9$, and $\beta_2$ is the center of $\beta$. Masking
from strings in $\B^1$ all occurrences of A and \revt{A}, each 1-block
$\beta_i = \text{A} \beta'_i \revt{A}$ in $\B^1$ becomes a 2-BFB palindrome
$\beta'_i$. Such 2-BFB palindromes may be further decomposed into 2-blocks,
yielding a 2-block collection $B^2$ (in Fig~\ref{fig:layers}b,
$\B^2 = \{2 \beta_5, \beta_6, 2 \beta_7\}$, where $\left|\B^2\right| = n_2 = 5$).
In general, for each $1 \leq l \leq k$, masking in $\beta$ all letters
$\sigma_r$ and $\rsigma_r$ such that $r < l$ defines a corresponding collection
of $l$-block substrings of $\beta$.
Each collection $\B^l$ contains exactly $n_l$ elements, as each $l$-block in
the collection contains exactly two out of the $2 n_l$ occurrences of
$\sigma_l$ in the string (where one occurrence is reversed).
The collection $\B^{l+1}$ is obtained from $\B^{l}$ by masking occurrences of
$\sigma_{l}$ and $\rsigma_l$ from the elements in $\B^{l}$, and decomposing the
obtained $(l+1)$-BFB palindromes into $(l+1)$-blocks. We may define
$\B^{k+1} = \emptyset$ (where $\emptyset$ denotes an empty collection), since
after masking in $\beta$ all segments $\sigma_1, \ldots, \sigma_k$ we are left
with an empty collection of $(k+1)$-blocks.

The algorithm we describe for the search variant of the BFB count vector problem exploits the above described property of BFB palindromes. Given a count vector $\vec{n} = [n_1, n_2, \ldots, n_k]$, the algorithm processes iteratively the counts in the vector one by one, from $n_k$ down to $n_1$, producing a series of collections $\B^k, \B^{k-1}, \ldots, \B^1$. Starting with $\B^{k+1} = \emptyset$, each collection $\B^l$ in the series is obtained from the preceding collection $\B^{l+1}$ in a
two-step procedure: First, $(l+1)$-blocks from $\B^{l+1}$ are
concatenated in a manner that produces an $(l+1)$-BFB palindrome
collection $\B'$ of size $n_l$ ($\B'$ may contain empty strings, which
can be thought of as concatenations of zero elements from
$\B^{l+1}$). Then, $\B^l$ is obtained by ``wrapping''
each element $\beta' \in \B'$ with a
pair of $\sigma_l$ characters to become an $l$-block $\beta = \sigma_l
\beta' \rsigma_l$. We will refer to the first step in this procedure as
collection \emph{folding}, and to the second step as collection
\emph{wrapping}.
For example, in Fig~\ref{fig:layers}d, the elements in $\B^4 = \{4 \beta_{10}\}$ are folded to form a 4-palindrome collection $\B' = \{2 \beta_{10}\beta_{10}, \varepsilon\}$ of size $n_3 = 3$. After wrapping each elements of $\B'$ by C to the left and \revt{C} to the right, we get the $3$-block collection $\B^3 = \{2 \text{C} \beta_{10}\beta_{10} \revt{C}, \text{C} \revt{C}\} = \{2 \beta_8, \beta_9\}$.
Algorithm SEARCH-BFB$(\vec{n})$ in Fig.~\ref{alg:BFB} gives the pseudo-code for the described procedure, excluding the implementation of the folding phase which is kept abstract here.
We next discuss some restrictions over the folding procedure, and point out that greedy folding is nontrivial. Nevertheless, in the SI document we show an explicit implementation of a folding procedure, which guarantees that the search algorithm finds a BFB string provided that the input is a valid BFB count vector.

\subsection{Required conditions for folding}

Recall that the input of the folding procedure is an $l$-block collection $\B$ and an integer $n$, and the procedure should concatenate all strings in $\B$ in some manner to produce an $l$-BFB palindrome collection $\B'$ of size $n$.
Since both $l$-blocks and empty strings are special cases of $l$-BFB palindromes, when $n \geq \left|\B\right|$ it is always possible to obtain $\B'$ by simply adding $n - \left|\B\right|$ empty strings to $\B$. Nevertheless, when $n < \left|\B\right|$, there are instances for which no valid folding exists, as shown next.

For a pair of collections $\B$ and $\B'$, $\B + \B'$ is the collection containing all elements in $\B$ and $\B'$.
 When $\B'' = \B + \B'$, we say that $\B = \B'' - \B'$ (note that $\B'' - \B'$ is well defined only when $\B''$ contains $\B'$).
 For some (possibly rational) number $x \geq 0$, denote by $x \times \B$ the collection $\{\left\lfloor xm_1\right\rfloor \times \beta_1, \left\lfloor xm_2\right\rfloor \times \beta_2, \ldots, \left\lfloor xm_q\right\rfloor \times \beta_q\}$.
The operation $\modTwo{\B}$ yields the sub-collection of $\B$ containing a single copy
of each distinct element $\beta$ with an odd count in $\B$. For example, for $\B = \{2 \times \beta_1, \beta_2, 5 \times \beta_3, 6 \times \beta_4\}$, $\modTwo{\B}=\{\beta_{2}, \beta_3\}$.
Observe that $\B = \modTwo{\B} + 2 \times \left(\frac{1}{2} \times \B\right)$.

\begin{claim}
\label{clm:evenConcatenation}
Let $\B$ be an $l$-BFB palindrome collection such that $\modTwo{\B} = \emptyset$. Then, it is possible to concatenate all elements in $\B$ to obtain a single $l$-BFB palindrome.
\end{claim}

\begin{proof}
By induction on the size of $\B$. By definition, $\modTwo{\B} = \emptyset$ implies that the counts of all distinct elements in $\B$ are even. When $\B = \emptyset$, the concatenation of all elements in $\B$ yields an empty string $\varepsilon$, which is an $l$-BFB palindrome as required. Otherwise, assume the claim holds for all collections $\B'$ smaller than $\B$. Let $\beta \in \B$ be an element such that for every $\beta' \in \B$, $\tp[\beta'] \leq \tp[\beta]$, and let $\B' = \B - \{2 \beta\}$. Note that $\modTwo{\B'} = \emptyset$ (since the count parity is identical for every element in both $\B$ and $\B'$), and from the inductive assumption it is possible to concatenate all elements in $\B'$ into a single $l$-BFB palindrome $\alpha'$. From Claim~\ref{clm:lPalindrome}, the string $\alpha = \beta \alpha' \beta$ is an $l$-BFB palindrome, obtained by concatenating all elements in $\B$.
\end{proof}

\begin{claim}
\label{clm:foldingUpperBound}
Let $\B$ be an $l$-block collection. There is a folding $\B'$ of $\B$ such that $\left|\B'\right| = \left|\modTwo{\B}\right| + 1$.
\end{claim}

\begin{proof}
Recall that $\B = \modTwo{\B} + 2 \left(\frac{1}{2} \B\right)$. Since all element counts in the collection $2 \left(\frac{1}{2} \B\right)$ are even, $\modTwo{2 \left(\frac{1}{2} \B\right)} = \emptyset$, and from Claim~\ref{clm:evenConcatenation} it is possible to concatenate all elements in $2 \left(\frac{1}{2} \B\right)$ into a single $l$-BFB palindrome $\alpha$. Thus, the collection $\B' = \modTwo{\B} + \alpha$ is a folding of $\B$ of size $\left|\modTwo{\B}\right| + 1$.
\end{proof}

\begin{claim}
\label{clm:foldingModTwo}
For every folding $\B'$ of an $l$-block collection $\B$, $\left|\modTwo{\B'}\right| \geq \left|\modTwo{\B}\right|$.
\end{claim}

\begin{proof}
Let $\beta \in \modTwo{\B}$ be an $l$-block repeating an odd number of times $m$ in $\B$. Therefore, $\beta$ appears as a center of at least one element $\beta'$ that occurs an odd number of times in $\B'$ (otherwise, $\beta$ has an even number of distinct repeats as a substring of elements in $\B'$, in contradiction to the fact that $m$ is odd). Hence, for each $\beta \in \modTwo{\B}$ there is a corresponding unique element $\beta' \in \modTwo{\B'}$, and so $\left|\modTwo{\B'}\right| \geq \left|\modTwo{\B}\right|$.
\end{proof}

The SEARCH-BFB$(\vec{n})$ algorithm described in Fig.~\ref{alg:BFB} tries in each iteration $l$ to fold the block collection $\B^{l+1}$ obtained in the previous iteration into an $(l+1)$-BFB palindrome collection of size $n_l$. When $n_l \geq \left|\modTwo{\B^{l+1}}\right| + 1$, there always exists a folding as required: $\B^{l+1}$ maybe folded into a collection of size $\left|\modTwo{\B^{l+1}}\right| + 1$ due to Claim~\ref{clm:foldingUpperBound}, and additional $n_l - \left|\modTwo{\B^{l+1}}\right| - 1$ empty strings may be added in order to get a folding of size $n_l$. On the other hand, when $n_l < \left|\modTwo{\B^{l+1}}\right|$, no folding as required exists, due to Claim~\ref{clm:foldingModTwo}. In the remaining case of $n_l = \left|\modTwo{\B^{l+1}}\right|$, the existence of an $n_l$-size folding of $\B^{l+1}$ depends on the element composition of $\B^{l+1}$, as exemplified next.

Consider the run of Algorithm SEARCH-BFB$(\vec{n})$ over the input count vector $\vec{n} = [1, 3, 2]$. Here, $k = 3$, and the algorithm starts by initializing the collection $\B^4 = \emptyset$. In the first loop iteration $l = 3$, and the algorithm first tries to fold the empty collection $\B^4$ into a $4$-BFB palindrome collection containing $n_3 = 2$ elements. Since there are no elements in $\B^4$ to concatenate, the only way to perform this folding is by adding to $\B^4$ two empty strings, yielding the collection $\B' = \{2 \varepsilon\}$, which after wrapping becomes $\B^3 = \{2 \text{C}\revt{C}\} = \{2 \beta_1\}$. In the next iteration $l = 2$, and $\B^3$ should be folded into a collection $\B'$ of size $n_2 = 3$. Among the possibilities to perform this folding are the following: $\B'^a = \{2 \beta_1, \varepsilon\}$, and $\B'^b = \{\beta_1 \beta_1, 2 \varepsilon\}$, which after wrapping become $\B^{2a} = \{2\text{B}\beta_1\revt{B}, \text{B}\revt{B}\} = \{2 \beta_2, \beta_3\}$, and $\B^{2b} = \{\text{B}\beta_1 \beta_1\revt{B}, 2 \text{B}\revt{B}\} = \{\beta_4, 2\beta_3\}$, respectively. Note that $\left|\modTwo{\B^{2a}}\right| = \left|\modTwo{\B^{2b}}\right| = 1$. Nevertheless, it is possible to fold $\B^{2a}$ in the next iteration into the collection $\{\beta_2 \beta_3 \beta_2\}$ of size $n_1 = 1$, while $\B^{2b}$ cannot be folded into such a collection. The reason is that the only concatenation of all elements in $\B^{2b}$ into a single palindrome is the concatenation $\beta_3 \beta_4 \beta_3$, but since $\tp[\beta_4] = \tp[\text{BC}\revt{C}\text{C}\revt{C}\revt{B}] = 3 > 2 = \tp[\text{B}\revt{B}] = \tp[\beta_3]$, Claim~\ref{clm:lPalindrome} implies that this concatenation is not a valid BFB palindrome.

In the SI document, we define a property called the \emph{signature} of a collection, and show how the exact minimum folding size depends on this signature. We also show how to fold a collection in a manner that optimizes this signature, and guarantees for valid BFB count vector inputs that the search algorithm finds an admitting BFB string.

\section{Running time}

For a count vector $\vec{n} = [n_1, \ldots, n_k]$, let
$N = \displaystyle{\sum_{1 \leq i \leq k}{n_i}}$ be the number of segments
in a string corresponding to $\vec{n}$. Let $\tilde{N} =
\displaystyle{\sum_{1 \leq i \leq k}{\log n_i}}$ denote a number proportional
to the number of bits in the representation of $\vec{n}$, assuming each count
$n_i$ is represented by $O(\log n_i)$ bits. In the SI document, we complete
the implementation details of algorithms for the decision, search, and distance
variants of the BFB count vector problem, and show these algorithms have the
asymptotic running times of $O(\tilde{N})$ (bit operations), $O(N)$, and
$O(N^{\log N})$ (under some realistic assumptions), respectively. For the
decision and search variants, these running times are optimal, being linear
in the input (for the decision variant) or output (for the search variant)
lengths.

In practical terms, this has a significant effect on our ability to evaluate
copy number signatures of BFB when compared to the previous exponential-time
algorithm~\cite{first_bfb_paper}. To determine if a count vector consistent with
BFB is in fact strong evidence for BFB, we have to check many count vectors.
Analyzing the simulations we explain below required testing tens of millions
of different count vectors, so even a small improvement in running time can
have a large impact of the scope of analysis we can perform.

But, the running time improvement with the new algorithm is not small.
For example, a count
vector that took 9 seconds with the previous algorithm can be processed by the
new algorithm in 1.2x$10^{-5}$ seconds. A count vector that needed 148 seconds with the
old algorithm now completes in 1.9x$10^{-5}$ seconds. A count vector that was
abandoned after 30 hours with the old algorithm now takes only 8.1x$10^{-6}$
seconds. Thus, the improvement in running time is not of merely theoretical interest.
The earlier algorithm did not allow a thorough study of longer count vectors,
while with the new algorithm such a study is possible.

\section{Detecting Signatures of BFB}
We can now describe the two features we will use to determine if a chromosome
has undergone BFB. The first feature is based on the fold-back inversions that BFB produces.
For a given region, we can find all the breakpoints identified by sequencing
and determine what proportion are fold-back inversions. We call this the
\emph{fold-back fraction}. The second feature relies on our algorithm that
solves the BFB count vector problems we have posed. For a given contiguous
pattern of copy counts, that is, a count vector, we can find the distance to
the nearest count vector that could be produced by BFB using the distance
metric we defined above. We call this the \emph{count vector distance}. For a
particular count vector, we define a score $s$ that combines these two features:
\begin{equation}
  s=\lambda \delta+(1-\lambda)(1-f)
\label{eqn:score}
\end{equation}
Here, $f$ refers to the fold-back fraction, $\delta$ refers to the count vector
distance, and $\lambda$ refers to the weight we give to the count vector
distance versus the fold-back fraction when calculating the score. When
$\lambda=1$, we are only looking at count vector distance, whereas when
$\lambda=0$, we are only using fold-back fraction and ignoring the count
vectors.

\section{Results}
\label{sec:results}
To determine whether our two proposed features could identify BFB against the
complex backdrop of a cancer genome, we simulated rearranged
chromosomes. Our overall goal was to simulate cancer
chromosomes that were highly rearranged yet had not undergone BFB to see if
evidence for BFB appeared in
them, suggesting that using such evidence would lead to false positives.
Conversely, we also wanted to simulate chromosomes whose rearrangements
included BFB to determine if a proposed BFB signature was sensitive enough to
identify BFB when it occurred. Since it is not clear how to faithfully
simulate cancer genome rearrangements, we used a wide range of simulation
parameters so we could understand how different assumptions affect the
features' ability to identify BFB.

We began with a pair of unrearranged chromosomes and then
introduced 50 rearrangements to each. Each rearrangement was an
inversion, a deletion, or a duplication. Duplications were either direct or
inverted and could be tandem or interspersed. The type of each rearrangement
was chosen from a
distribution. In some chromosome pairs, we imitated BFB by
successively duplicating and inverting segments of one end of one
chromosome for each round of BFB. The number of BFB rounds varied from
two to ten. Then, we calculated the copy counts and breakpoints for
the chromosome pair and introduced error to the copy counts according
to a random model and also randomly deleted or inserted breakpoint
observations. For each combination of rearrangement type distribution
and number of BFB rounds, we simulated 5,000 chromosome pairs with BFB
and 15,000 without BFB. Complete details are in the SI.

We first examined the usefulness of count vector distance alone in identifying
BFB by setting $\lambda=1$ in our score function (Eqn.~\ref{eqn:score}).
For each chromosome pair, we found all
contiguous count vectors of a given length and calculated their scores, as
described above and in the SI. We used the minimum score $s$ over all of these
sub-vectors in the chromosome as a score for the whole chromosome. Then, for
varying thresholds, we
classified all chromosomes with a score lower than the threshold as
having been rearranged by BFB. The performance of this classification
varied with the parameters used to simulate the chromosomes, but
typical results can be seen in Figure \ref{fig:ROC_chr12}a. The solid lines
show ROC curves for different count vector lengths for
the simulation with eight rounds of BFB and a distribution that yields
roughly equal probabilities of the other rearrangement types. Consistent with
previous observations,
short count vectors that are perfectly consistent with BFB can be
found in many chromosomes, even if BFB did not occur. So, even with a score
threshold of zero, they would still be classified as consistent with BFB.
For example,
63\% of chromosomes without any true BFB rearrangements in Figure
\ref{fig:ROC_chr12}a had a count vector of length six perfectly consistent with BFB.

In contrast, examining longer count vectors produced a better
classification.  For instance, setting the score
threshold to .10, count vectors of length twelve could achieve a true
positive rate (TPR) of 67\% and a false positive rate (FPR) of only
10\%. However, this performance must be considered in the context of
an experiment seeking evidence for BFB. Chromosomes that have
undergone BFB are probably rare. If only one in a hundred chromosomes
tested underwent BFB, then a test with an FPR of even 1\% will produce
mostly false discoveries. Achieving this FPR with count vectors of
length twelve with the chromosomes in Figure \ref{fig:ROC_chr12}a would
result in a TPR of only 16\%. A more appropriate target FPR for
screening many samples, say .1\%, could not be achieved with count
vectors alone.

Next, we incorporated fold-back inversions into the scoring function. We set
$\lambda=.5$, giving equal weight to fold-back fraction and count vector
distance.
ROC curves using this approach are shown by dashed lines in
Figure \ref{fig:ROC_chr12}a. Incorporating fold-back fractions into the
scoring leads to better discrimination of chromosomes with and without
BFB rearrangements; the test in Figure \ref{fig:ROC_chr12}a that combines
count vectors of length 12 and fold-back inversions can achieve a TPR
of 48\% with an FPR of .1\% by setting the score threshold to .27. This
suggests that it could detect BFB in a large dataset without being overwhelmed by
false discoveries.

Of course, these conclusions depend on our simulation resembling actual cancer
rearrangements and BFB cycles. A true specification of cancer genome evolution
is unknown and in any case varies from cancer to cancer. Recognizing this
complication, we repeated the analysis in Figure \ref{fig:ROC_chr12}a for the
different rearrangement distributions, number of BFB rounds, and count vector
lengths. For each combination, we recorded the score threshold needed to
achieve FPRs of .1\%, 1\%, and 5\%, and the respective expected TPRs. The full
results are shown in Dataset S1 and ROC curves are shown in Figures S1-5.
Generally, different simulations showed the same trends. Fold-back inversions
alone were better at identifying BFB than count vectors alone, but the
combination of both features provided the best classification.
By examining a wide range of simulation
parameters, we illustrate how changes in assumptions about cancer genome
evolution and BFB influence the appropriateness and expected outcomes of tests
for BFB.

We applied our method to a publicly available dataset of copy number profiles
from 746 cancer cell lines~\cite{CellLines}. We found three chromosomes with count vectors
of length 12
nearly consistent with BFB: chromosome 8 from cell line AU565, chromosome 10
from cell line PC-3, and chromosome 8 from cell line MG-63 (see SI). While
the patterns of copy counts on these chromosomes do bear the hallmarks of BFB,
our simulations suggest that labeling chromosomes as having undergone BFB
based on these count vectors would lead to an FPR between 1\% and 10\%. Given
that thousands of chromosomes were examined, many of which were highly
rearranged, the consistency of these copy counts with BFB may be spurious.

We also applied our method to paired-end sequencing data from seven
previously
published pancreatic cancer samples~\cite{pancreas}. We estimated copy
numbers from the reads and used breakpoints as reported by the original
investigators. We examined count vectors of length 8 and chose a threshold
score of .18, which would give an FPR of .1\% based on simulations where the
non-BFB rearrangement types are roughly equally likely. We identified two chromosomes
that showed evidence for BFB, both from the same sample, PD3641. The first was the long
arm of chromosome 8. This chromosome was identified by the original
investigators as likely being rearranged by BFB. Our analysis suggests that,
barring rearrangements that differ significantly from any of our simulations,
this chromosome did indeed undergo BFB cycles. We also found evidence for
BFB rearrangements from a count vector spanning ten megabases on the short arm
of chromosome 12 (Figure \ref{fig:ROC_chr12}b). Thus, we were able to recover evidence for
BFB previously
identified by hand curation. And by combining count vector and fold-back
analysis, we found an additional strong BFB candidate that would not be
apparent without modeling and simulation.

\section{Discussion}
Some $80$ years after Barbara McClintock's discovery of the
Breakage Fusion Bridge mechanism, it is seeing renewed interest in the
context of tumor genome evolution. Recent publications have
claimed, based on empirical observations of segmentation counts and
other features, that their data counts are ``consistent with
BFB''. The main technical contribution of the paper is an efficient
algorithm for detecting if given segmentation counts can indeed be
created by Breakage Fusion Bridge cycles. That algorithm turns out to
be non-trivial, requiring a deep foray into the combinatorics
of BFB count vectors, even though its final implementation is
straightforward and fast. Experimenting with the implementation
reveals that in fact, (a) there is a big diversity of count-vectors
created by true BFB cycles not all of which are easily recognizable as
BFB; and, (b) at least for short count-vectors, it is often possible to
create BFB-like vectors by non-BFB operations. Thus, being
``consistent with BFB'', and ``caused by BFB'' are not
equivalent. Fortunately, our results also suggest that using longer
count vectors, and additional information of fold-backs gives stronger
prediction of BFB, even in the presence of noise, and diploidy. While
assembly of these highly rearranged genomes continues to be difficult,
recent advances in long single-molecule sequencing will provide
additional spatial information that will improve the resolving power
of our algorithm. As more cancer genomes are sequenced, including
single-cell sequencing, the method presented here will be helpful in
determining the extent and scope of BFB cycles in the evolution of the
tumor genome.

\section{Materials}
Details on the algorithm, and on the simulation methods are available
in the accompanying supplemental information (SI).
\subsection{Code availability}
Java and Python code used to analyze chromosomes is available at
www.bitbucket.org/mckinsel/bfb

\subsection{Estimating pancreas tumor copy number}
The pancreas tumor data was downloaded from the European Genome-Phenome
Archive, accession number EGAS00000000064. The data was paired-end reads; each
end was 37 bases long. We aligned the reads with Bowtie~\cite{bowtie}. Then
we used readDepth~\cite{readDepth} for segmentation and integer copy number
estimation.

\section{Acknowledgements}
  This research was supported by grants from the National
  Institute of Health (5RO1-HG004962, U54 HL108460) and the National
  Science Foundation (NSF-CCF-1115206).

\bibliographystyle{unsrt}




\begin{figure}
\centering
\includegraphics[width=.6\textwidth]{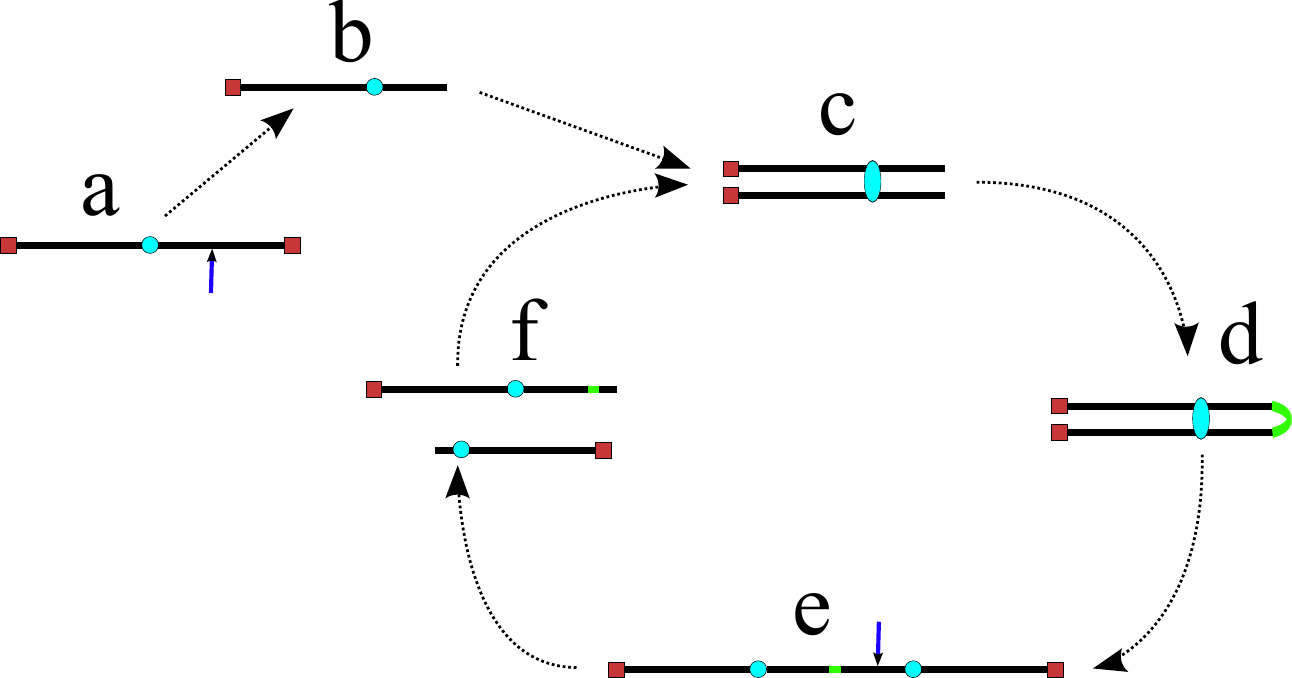}
\caption{%
{\normalsize
A schematic BFB process.
}
}
\label{fig:overview}
\end{figure}

\begin{figure}
\centering
\includegraphics[width=.6\textwidth]{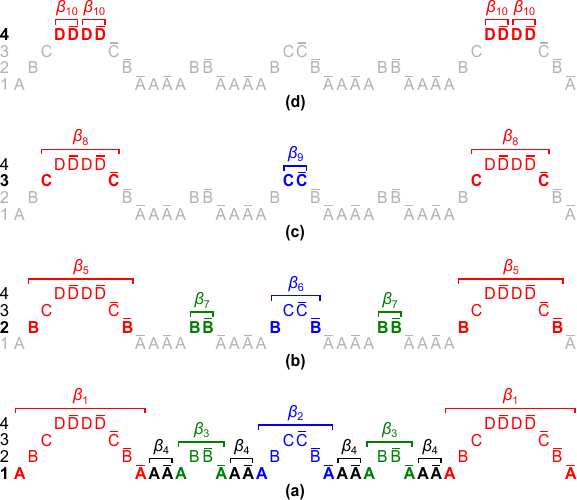}
\caption{%
{\normalsize
Layer visualization of a BFB palindrome $\beta = \alpha \rev{\alpha}$, where $\alpha = \text{ABCD\revt{D}D\revt{D}\revt{C}\revt{B}\revt{A}A\revt{A}AB\revt{B}\revt{A}A\revt{A}ABC}$. A possible BFB sequence that produces $\alpha$ is
$\text{ABCD} \rightarrow \text{ABCD\textbf{\revt{D}}} \rightarrow \text{ABCD\revt{D}\textbf{D\revt{D}\revt{C}\revt{B}\revt{A}}} \rightarrow \text{ABCD\revt{D}D\revt{D}\revt{C}\revt{B}\revt{A}\textbf{A}} \rightarrow \text{ABCD\revt{D}D\revt{D}\revt{C}\revt{B}\revt{A}A\textbf{\revt{A}AB}} \rightarrow \text{ABCD\revt{D}D\revt{D}\revt{C}\revt{B}\revt{A}A\revt{A}AB\textbf{\revt{B}\revt{A}A\revt{A}ABC}}$.
$\vec{n}(\alpha) = [9, 5, 3, 4]$, and $\vec{n}(\beta) = 2\vec{n}(\alpha)$. Figures (a) to (d) depict layers 1 to 4 of $\beta$, respectively. In each layer $l$, the $l$-blocks composing the collection $\B^l$ are annotated as substrings of the form $\beta_i$. These collections are: $\B^1 = \{2 \times \beta_1, \beta_2, 2 \times \beta_3, 4 \times \beta_{4}\}$, $\B^2 = \{2 \times \beta_5, \beta_6, 2 \times \beta_7\}$, $\B^3 = \{2 \times \beta_8, \beta_9\}$, $\B^4 = \{4 \times \beta_{10}\}$.
}
}
\label{fig:layers}
\end{figure}

\begin{figure}
    \centering
    \includegraphics{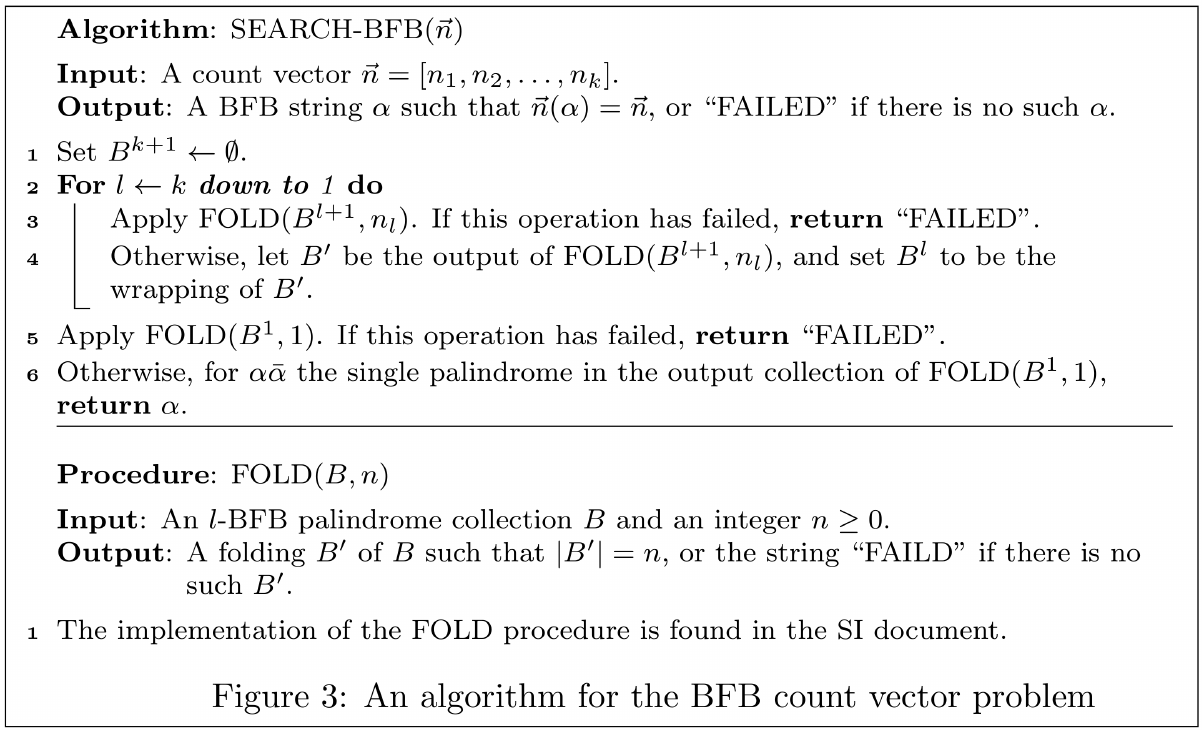}
    \caption{{\normalsize An algorithm for the BFB count vector problem}}
    \label{alg:BFB}
\end{figure}

\begin{figure}
    \centering
    \includegraphics[width=.9\textwidth]{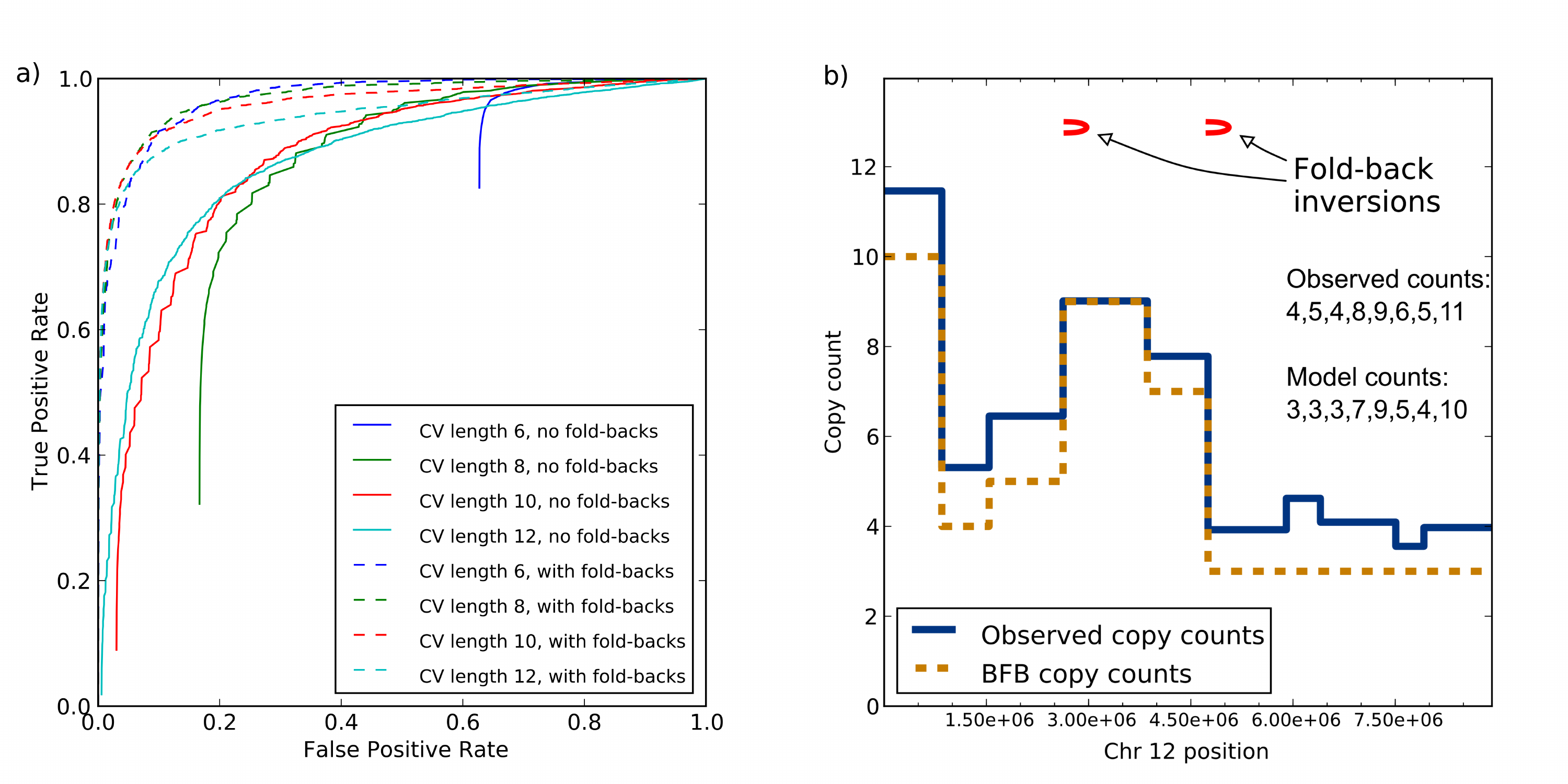}
    \caption{{\normalsize Simulation and pancreatic cancer results. a) ROC curves for
    different count vector lengths with and without fold-back fractions for the
    simulation of eight BFB rounds and equally likely other rearrangements. b)
    Observed copy counts and copy counts compatible with BFB on the short arm
    of chromosome 12 in pancreatic cancer sample PD3641. The presence of
    fold-back inversions and the count vector's consistency with BFB suggests
    that this portion of chromosome 12 underwent BFB cycles.}}
    \label{fig:ROC_chr12}
\end{figure}

\end{document}


\title{Detecting Breakage Fusion Bridge cycles in tumor genomes---an algorithmic approach}
\subtitle{Supporting Information}
\maketitle

\setcounter{definition}{6}
\setcounter{claim}{2}
\renewcommand{\thefigure}{S\arabic{figure}}
\renewcommand{\thetable}{S\arabic{table}}

\section{Properties of BFB Strings}

\newcommand{\bfb}{\ensuremath{\stackrel{\scriptscriptstyle{\text{BFB}}}{\longrightarrow}}}
\newcommand{\bfbm}[1][m]{\ensuremath{\begin{array}{c}
	\scriptscriptstyle \text{BFB} \\[-6pt]
	\longrightarrow \\[-8pt]
	\scriptscriptstyle #1
\end{array}}}

In this section, we prove Claim 1 from the main manuscript. To do so, we first formulate several auxiliary claims.

\begin{observation}
\label{obs:bfb}
If $\alpha \bfb \beta$, $\beta = \beta' \beta''$, and $\beta'' \bfb \beta''\gamma$, then $\alpha \bfb \beta \gamma$.
\end{observation}

Call a string $\alpha$ an \emph{$l$-$t$-string} if for the count vector $\vec{n}(\alpha) = [n_1, n_2, \ldots, n_k]$, $n_r > 0$ if and only if $l \leq r \leq t$.
Thus, an $l$-$t$-BFB string is an $l$-BFB string $\alpha$ such that $\tp = t$. Denote by $\alpha_{l,t}$ the consecutive genomic region $\alpha_{l,t} = \sigma_l \sigma_{l+1} \ldots \sigma_t$ (when $t < l$, $\alpha_{l, t} = \varepsilon$), and observe that $l$-$t$-BFB strings always start with the prefix $\alpha_{l,t}$.

\begin{claim}
\label{clm:substring}
Let $l', l, t$ be integers and $\alpha \beta$ an $l'$-BFB string such that $\beta$ is an $l$-$t$-string. Then,
\begin{enumerate}
	\item If $\beta$ starts with the prefix $\alpha_{l,t}$, then $\alpha_{l,t} \bfb \beta$ (i.e. $\beta$ is an $l$-$t$-BFB string).
	\item If $\beta$ ends with the suffix $\alpha_{l,t}$, then $\rev{\alpha}_{l,t} \bfb \rev{\beta}$.
	\item If $\beta$ starts with the prefix $\rev{\alpha}_{l,t}$, then $\rev{\alpha}_{l,t} \bfb \beta$.
	\item If $\beta$ ends with the suffix $\rev{\alpha}_{l,t}$, then $\alpha_{l,t} \bfb \rev{\beta}$.
\end{enumerate}
\end{claim}

\begin{proof}
When $t < l$, $\alpha_{l, t} = \beta = \varepsilon$, and all four items in the claim are sustained in a straightforward manner. Similarly, when $\alpha \beta = \alpha_{l', t}$, then $\beta = \alpha_{l,t}$ and again all four items in the claim are sustained. Otherwise, $t \geq l$ and there are some $\rho, \gamma$ such that $\gamma \neq \varepsilon$, $\rho \gamma$ is an $l'$-BFB string, and $\alpha \beta = \rho \gamma \rev{\gamma}$. In particular, $\alpha_{l', t}$ is a proper prefix of $\alpha \beta$. Assume by induction that the claim is sustained with respect to all proper prefixes of $\alpha \beta$ (from Lemma~2 in~\cite{first_bfb_paper}, all such prefixes are $l'$-BFB strings). Note that $\beta$, $\rev{\gamma}$, and $\gamma \rev{\gamma}$ are all suffixes of $\alpha \beta = \rho \gamma \rev{\gamma}$.
Consider three cases: 1. $\beta$ is a proper suffix of $\rev{\gamma}$,  2. $\beta$ is a proper suffix of $\gamma \rev{\gamma}$ and $\rev{\gamma}$ is a suffix of $\beta$, and 3. $\gamma \rev{\gamma}$ is a suffix of $\beta$.

\vspace{10pt}
\noindent
\textbf{1.} $\beta$ is a proper suffix of $\rev{\gamma}$.
In this case, $\rev{\gamma} = \gamma' \beta$ for some string $\gamma' \neq \varepsilon$, therefore $\rho \gamma \rev{\gamma} = \rho \rev{\beta} \rev{\gamma'} \gamma' \beta$. From the inductive assumption and the fact that $\rho \rev{\beta}$ is a proper prefix of $\rho \rev{\beta} \rev{\gamma'} = \rho \gamma$ (which is in turn a proper prefix of $\alpha \beta$), $\rho \rev{\beta}$ sustains the claim. Therefore,
\begin{enumerate}
	\item If $\beta$ starts with the prefix $\alpha_{l,t}$, then $\rev{\beta}$ ends with the suffix $\rev{\alpha}_{l,t}$, therefore $\alpha_{l,t} \bfb \beta$.
	\item If $\beta$ ends with the suffix $\alpha_{l,t}$, then $\rev{\beta}$ starts with the prefix $\rev{\alpha}_{l,t}$, therefore $\rev{\alpha}_{l,t} \bfb \rev{\beta}$.
	\item If $\beta$ starts with the prefix $\rev{\alpha}_{l,t}$, then $\rev{\beta}$ ends with the suffix $\alpha_{l,t}$, therefore	$\rev{\alpha}_{l,t} \bfb \beta$.
	\item If $\beta$ ends with the suffix $\rev{\alpha}_{l,t}$, then $\rev{\beta}$ starts with the prefix $\alpha_{l,t}$, therefore $\alpha_{l,t} \bfb \rev{\beta}$.
\end{enumerate}

\vspace{10pt}
\noindent
\textbf{2.} $\beta$ is a proper suffix of $\gamma \rev{\gamma}$ and $\rev{\gamma}$ is a suffix of $\beta$.
In this case, there are some $\gamma_1$ and $\gamma_2$ such that $\gamma_1 \neq \varepsilon$, $\gamma = \gamma_1 \gamma_2$, $\gamma \rev{\gamma} = \gamma_1 \gamma_2 \rev{\gamma}_2 \rev{\gamma}_1$  and $\beta = \gamma_2 \rev{\gamma}_2 \rev{\gamma}_1$. Thus, $\alpha \beta = \rho \gamma \rev{\gamma} = \rho \gamma_1 \gamma_2 \rev{\gamma}_2 \rev{\gamma}_1 = \rho \rev{\beta} \rev{\gamma}_1$. Here also, we get that $\rho \rev{\beta}$ is a proper prefix of $\alpha \beta$, and similarly as in the previous case the inductive assumption implies the correctness of the claim.

\vspace{10pt}
\noindent
\textbf{3.} $\gamma \rev{\gamma}$ is a suffix of $\beta$.
In this case, there is some $\gamma'$ such that $\beta = \gamma' \gamma \rev{\gamma}$, and therefore $\alpha \beta = \alpha \gamma' \gamma \rev{\gamma}$. To show items (1) and (3) in the claim, assume that $\beta$ starts with the prefix $\phi$ such that either $\phi = \alpha_{l,t}$ or $\phi = \rev{\alpha}_{l,t}$, respectively. It must be that $\phi$ is a prefix of $\gamma' \gamma$, since the first character of $\rev{\gamma}$ is the reverse of the last character of $\gamma$, and thus cannot be included in $\phi$. Therefore, from the inductive assumption and the fact that $\alpha \gamma' \gamma$ is a proper prefix of $\alpha \gamma' \gamma \rev{\gamma} = \alpha \beta$ (recall that $\gamma \neq \varepsilon$ and therefore $\rev{\gamma} \neq \varepsilon$), $\phi \bfb \gamma' \gamma$. By definition, $\phi \bfb \gamma' \gamma \rev{\gamma} = \beta$, proving items (1) and (3) in the claim.

To show items (2) and (4) in the claim, assume that $\beta$ ends with the suffix $\phi$ such that either $\phi = \alpha_{l,t}$ or $\phi = \rev{\alpha}_{l,t}$, respectively. Similarly as above, it must be that $\phi$ is a suffix of $\rev{\gamma}$. Note that case (2) of this proof implies that $\rev{\phi} \bfb \gamma$, and by definition $\rev{\phi} \bfb \gamma \rev{\gamma}$. In particular, $\rev{\phi}$ is a prefix of $\gamma$, and therefore the string $\alpha \gamma' \rev{\phi}$ is a proper prefix of $\alpha \beta = \alpha \gamma' \gamma \rev{\gamma}$, and $\rev{\phi}$ is the suffix of the suffix $\gamma' \rev{\phi}$ of $\alpha \gamma' \rev{\phi}$. From the inductive assumption, $\phi \bfb \phi \rev{\gamma'}$. Thus, from Observation~\ref{obs:bfb}, and the fact that $\phi$ is a suffix of $\rev{\gamma}$, we get that $\rev{\phi} \bfb \gamma \rev{\gamma} \bfb \gamma \rev{\gamma} \rev{\gamma'} = \rev{\beta}$, and items (2) and (4) in the clam follow.
\end{proof}

\begin{claim}
\label{clm:tPalindrome}
Let $\alpha$ be a BFB string, and let $\sigma \beta \rev{\sigma}$ be a substring of $\alpha$ such that $\beta$ contains no occurrences of $\sigma$ or $\rev{\sigma}$. Then, $\beta$ is a palindrome.
\end{claim}

\begin{proof}
From Lemma~2 in~\cite{first_bfb_paper}, every prefix of $\alpha$ is a BFB string, and thus we may assume without loss of generality that $\sigma \beta \rev{\sigma}$ is a suffix of $\alpha$. We prove the claim by induction over the length of $\alpha$. Note that for getting a substring of the form $\sigma \beta \rev{\sigma}$, $\alpha$ must be of the form $\alpha = \rho \gamma \rev{\gamma}$, where $\gamma \neq \varepsilon$ (since strings of the form $\alpha_{l,t}$ cannot contain both characters $\sigma$ and $\rev{\sigma}$). If $\rev{\gamma}$ is a suffix of $\sigma \beta \rev{\sigma}$, then $\rev{\gamma}$ ends with $\rev{\sigma}$, and does not contain any additional occurences of $\sigma$ or $\rev{\sigma}$. Therefore, $\gamma$ starts with $\sigma$, and it must be that $\sigma \beta \rev{\sigma} = \gamma \rev{\gamma}$, and in particular $\beta$ is a palindrome. Else, $\sigma \beta \rev{\sigma}$ is a suffix of $\rev{\gamma}$, therefore $\sigma \rev{\beta} \rev{\sigma}$ is a prefix of $\gamma$. In particular, the prefix $\rho \sigma \rev{\beta} \rev{\sigma}$ of $\rho \gamma$ is a proper prefix of $\alpha$ (since $\rev{\gamma} \neq \varepsilon$). Since $\rho$ is a BFB string (Lemma~2 in~\cite{first_bfb_paper}), the inductive assumption implies that $\rev{\beta}$, and therefore $\beta$, is a palindrome.
\end{proof}

\begin{claim}
\label{clm:convexed}
Let $\alpha$ be a BFB string and $\gamma$ a palindromic concatenation of $l$-blocks, such that $\alpha$ contains $\rev{\alpha}_{l,t} \gamma \alpha_{l,t}$ as a substring and $\tp[\gamma] = t' \leq t$. Then, $\gamma$ is a convexed $l$-palindrome.
\end{claim}

\begin{proof}
By induction on the number of $l$-blocks composing $\gamma$. If $\gamma$ is composed of zero $l$-blocks, then $\gamma = \varepsilon$, which is a convexed $l$-palindrome by definition. Otherwise, $\gamma$ is of the form $\gamma = \beta_1 \beta_2 \ldots \beta_{q} \beta_{q+1} \beta_{q} \ldots \beta_2 \beta_1$, where $\beta_i$ is an $l$-block for every $1 \leq i \leq q$, and $\beta_{q+1}$ is an $l$-block in case $\gamma$ is composed of an odd number $2q + 1$ of blocks and $\beta_{q+1} = \varepsilon$ in case $\gamma$ is composed of an even number $2q$ of blocks. Let $i$ be the minimum index such that $\tp[\beta_i] = t'$. Observe that $\gamma = \gamma' \gamma'' \rev{\gamma'}$, where $\gamma' = \beta_1 \beta_2 \ldots \beta_{i-1}$ is a concatenation of $l$-blocks such that $\tp[\gamma'] < t'$  (from the selection of $i$), and $\gamma'' = \beta_i \ldots \beta_{q} \beta_{q+1} \beta_{q} \ldots \beta_i$ is a palindromic concatenation of $l$-blocks with $\tp[\gamma''] = t'$. Since $\gamma''$ is a substring of $\alpha$, it is the suffix of some prefix $\alpha'$ of $\alpha$. From Lemma~2 in~\cite{first_bfb_paper}, $\alpha'$ is a BFB string. From the fact that $\gamma''$ starts with $\alpha_{l,t'}$ (as $\alpha_{l,t'}$ is a prefix of the $l$-$t'$-block $\beta_i$), we get from Claim~\ref{clm:substring} that $\gamma''$ is an $l$-BFB string, and in particular it is an $l$-BFB palindrome. In addition, observe that $\alpha$ contains $\rev{\alpha}_{l,t'} \gamma' \alpha_{l,t'} = \rev{\sigma}_{t'} \rev{\alpha}_{l,t'-1} \gamma' \alpha_{l,t'-1} \sigma_{t'}$ as a substring. Since $\rev{\alpha}_{l,t'-1} \gamma' \alpha_{l,t'-1}$ does not contain occurrences of $\sigma_{t'}$ or $\rev{\sigma}_{t'}$, from Claim~\ref{clm:tPalindrome}, $\rev{\alpha}_{l,t'-1} \gamma' \alpha_{l,t'-1}$, and in particular $\gamma'$, is a palindrome.
Thus, from the inductive assumption, $\gamma'$ is a convexed $l$-palindrome, and by definition $\gamma = \gamma' \gamma'' \rev{\gamma'} = \gamma' \gamma'' \gamma'$ is a convexed $l$-palindrome.
\end{proof}

\begin{claim}
\label{clm:convexedGeneration}
Let $l, t', t$ be integers such that $l, t' \leq t$. For every convexed $l$-$t'$-palindrome $\gamma$,  $\rev{\alpha}_{l,t} \bfb \rev{\alpha}_{l,t} \gamma \alpha_{l,t}$.
\end{claim}

\begin{proof}
We prove the claim by induction on $t'$. When $t' < l$, $\gamma = \varepsilon$ is the only convexed $l$-$t'$-palindrome, and by definition $\rev{\alpha}_{l,t} \bfb \rev{\alpha}_{l,t} \alpha_{l,t}$. Otherwise, $t' \geq l$, and assume by induction the claim holds for every $l, t'', t$ such that $t'' < t' \leq t$. By definition, $\gamma$ is of the form $\gamma' \beta \gamma'$, where $\gamma'$ is a convexed $l$-$t''$-palindrome such that $t'' < t'$, and $\beta$ is an $l$-$t'$-BFB palindrome. From the inductive assumption, $\rev{\alpha}_{l,t} \bfb \rev{\alpha}_{l,t} \gamma' \alpha_{l,t}$, and therefore $\rev{\alpha}_{l,t} \bfb \rev{\alpha}_{l,t} \gamma' \alpha_{l,t'}$. As $\beta$ is an $l$-$t'$-BFB palindrome, $\beta$ is of the form $\beta = \alpha \rev{\alpha}$, where $\alpha$ is an $l$-$t'$-BFB string. In particular, $\alpha_{l,t'} \bfb \alpha$, and from Observation~\ref{obs:bfb} and the fact that $\rev{\alpha}_{l,t} \bfb \rev{\alpha}_{l,t} \gamma' \alpha_{l,t'}$, we get that $\rev{\alpha}_{l,t} \bfb \rev{\alpha}_{l,t} \gamma' \alpha \bfb \rev{\alpha}_{l,t} \gamma' \alpha \rev{\alpha} \rev{\gamma'} \alpha_{l, t} = \rev{\alpha}_{l,t} \gamma' \beta \gamma' \alpha_{l, t} = \rev{\alpha}_{l,t} \gamma \alpha_{l, t}$.
\end{proof}

Finally, we turn to prove the correctness of Claim 1 from the main manuscript.

\vspace{6pt}
\noindent
\textbf{Claim 1.}
\emph{
  A string $\alpha$ is an $l$-BFB palindrome if and only if $\alpha =
  \varepsilon$, $\alpha$ is an $l$-block, or $\alpha = \beta \gamma
  \beta$, such that $\beta$ is an $l$-BFB palindrome, $\gamma$ is a convexed
  $l$-palindrome, and $\lesseq{\gamma}{\beta}$.
}

\begin{proof}

By definition, if $\alpha = \varepsilon$ or $\alpha$ is an $l$-block, then $\alpha$ is an $l$-BFB palindrome. Thus, it remains to show that when $\alpha$ is neither $\varepsilon$ nor an $l$-block, $\alpha$ is an $l$-BFB palindrome if and only if $\alpha = \beta \gamma \beta$, such that $\beta$ is an $l$-BFB palindrome, $\gamma$ is a convexed $l$-palindrome, and $\lesseq{\gamma}{\beta}$. Let $t = \tp[\alpha]$.

Assume that $\alpha$ is an $l$-BFB palindrome which is neither $\varepsilon$ nor an $l$-block. Therefore, $\alpha$ is a concatenation of at least two $l$-blocks, and so $\alpha$ is of the form $\alpha = \beta \gamma \beta$, such that $\beta$ is an $l$-block and $\gamma$ is some palindromic concatenation of $l$-blocks. Thus, $\beta$ must start with the prefix $\alpha_{l,t}$ and end with the suffix $\rev{\alpha}_{l,t}$, and $\tp[\gamma] \leq t = \tp[\beta]$. In addition, observe that $\rev{\alpha}_{l,t} \gamma \alpha_{l,t}$ is a substring of $\alpha$, and from Claim~\ref{clm:convexed}, $\gamma$ is a convexed $l$-palindrome, proving this direction of the claim.

For the other direction, assume that $\alpha = \beta \gamma \beta$, such that $\beta$ is an $l$-BFB palindrome, $\gamma$ is a convexed $l$-palindrome, and $\lesseq{\gamma}{\beta}$. Therefore, $\tp[\beta] = t$ , and $\tp[\gamma] = t' \leq t$. Since $\beta$ is an $l$-$t$-BFB string, it starts with the prefix $\alpha_{l,t}$, and being a palindrome it ends with the suffix $\rev{\alpha}_{l,t}$. From Claim~\ref{clm:convexedGeneration} and Observation~\ref{obs:bfb}, $\beta \gamma \alpha_{l,t}$ is an $l$-BFB string, and applying again Observation~\ref{obs:bfb}, $\beta \gamma \beta = \alpha$ is an $l$-BFB string. Being a palindrome, $\alpha$ is an $l$-BFB palindrome.
\end{proof}

\section{Algorithm SEARCH-BFB}

This section completes the missing details in the description of Algorithm SEARCH-BFB in the main manuscript. We describe the FOLD procedure, prove the correctness of the algorithm, and analyze its running time.

\subsection{Additional Notation and Collection Arithmetics}

In order to give an implementation of the FOLD procedure, we first add notation and definitions of some new entities, and observe related properties. For short, from now on we simply say a ``collection'' when referring to an $l$-BFB palindrome collection (in some cases we will explicitly indicate that the collection is an $l$-block collection). A collection containing a single element $\beta$ will be simply denoted by $\beta$, instead of $\{\beta\}$.

For two numbers $t, t'$ and a collection $\B$, $\rng{t}{t'}$ denotes the sub-collection containing all elements $\beta$ in $\B$ such that $t \leq \tp[\beta] < t'$. Denote $\geqt{t} = \rng{t}{\infty}$ and $\lt{t} = \B - \geqt{t} = \rng{0}{t}$.
For a nonempty collection $\B$, denote $\mint(\B) = \displaystyle{\min_{\beta \in \B} \{\tp[\beta]\}}$, where $\mint(\emptyset)$ is defined to be $\infty$. Say that an element $\beta \in \B$ is \emph{minimal} in $\B$ if $\tp[\beta] = \mint(\B)$.
The collection $\B = \B' \cap \B''$ contains all elements appearing in both $\B'$ and $\B''$, where the count of each element $\beta \in \B$ equals to the minimum among the counts of $\beta$ in $\B'$ and $\B''$.
Say that $\B' \subseteq \B$ if $\B' = \B \cap \B'$.
Notations of the form $\ser{a}$ will denote series $\ser{a} = a_0, a_1, a_2, \ldots$, and $\ser{a}_d$ denotes the prefix $a_0, a_1, \ldots, a_d$ of $\ser{a}$.
For an integer $m \neq 0$, denote by $\dig{m}$ the maximum integer $d \geq 0$ such that $m$ is divided by $2^d$. For example, $\dig{8} = \dig{-24} = 3$, and $\dig{7} = 0$. Observe that $\dig{m} = 0$ when $m$ is odd, and otherwise $\dig{m} = 1 + \dig{\frac{m}{2}}$. $\dig{m}$ can also be understood as the index of the least significant bit different from 0 in the binary representation of $m$, and in particular $\dig{m} \leq \log_2 m$.

\begin{observation}
\label{obs:modTwo}
For two collections $\B, \B'$,
\begin{itemize}
	\item $\begin{array}{rl}
		\modTwo{\B + \B'} = & \modTwo{\B + \modTwo{\B'}} = \modTwo{\modTwo{\B} + \B'} \\
			= & \modTwo{\modTwo{\B} + \modTwo{\B'}} \\
			= & \modTwo{\B} + \modTwo{\B'} - 2 \times (\modTwo{\B} \cap \modTwo{\B'})
	\end{array}$.
	\item For an integer $i \geq 0$, $\modTwo{\B + i \times \B'} = \modTwo{\B + \B'}$ when $i$ is odd, and $\modTwo{\B + i \times \B'} = \modTwo{\B}$ when $i$ is even.
	 In particular, $\modTwo{\B - \B'} = \modTwo{\B - \B' + 2\B'} = \modTwo{\B + \B'}$.
	\item For two integers $t$ and $t'$, $\modTwo{\rng{t}{t'}} = \rng[\left(\modTwo{\B}\right)]{t}{t'}$.
\end{itemize}
\end{observation}

\begin{definition}
\label{def:convexedCollection}
A \emph{convexed $l$-collection of order $q$} is an $l$-BFB palindrome
collection $\A$ of the form $\A = \{\alpha_1,
2\times\alpha_1, 4\times\alpha_2, \ldots, 2^{q-1}\times\alpha_{q}\}$, where
$\alpha_{q} <^t \alpha_{q-1} <^t \ldots <^t \alpha_{1}$.
\end{definition}

A convexed $l$-collection  of order $q$ $\A = \{\alpha_1, \ldots, 2^{q-1}\times\alpha_{q}\}$ satisfies $\left|\A\right| =
2^q -1$. In addition, $\A = \emptyset$ when $q = 0$, and when $\A \neq \emptyset$, $\modTwo{\A} = \alpha_1$ and $\frac{\A}{2} \equiv \frac{1}{2} \times \A = \{\alpha_2,
2 \times \alpha_{3}, \ldots, 2^{q-2} \times \alpha_{q}\}$ is a convexed
$l$-collection of order $q-1$. It is possible to concatenate all elements in $\A$ to produce a convexed $l$-palindrome $\gamma_{\A}$, where $\gamma_{\A} = \varepsilon$ if $\A = \emptyset$, and otherwise $\gamma_{\A} = \gamma_{\frac{\A}{2}} \alpha_1 \gamma_{\frac{\A}{2}}$.
In Fig. 2a, all 1-blocks besides the two repeats of $\beta_{1}$ form a convexed $1$-collection $\A=\{\beta_{2}, 2\times \beta_{3}, 4\times \beta_{4}\}$ of order
$3$, where $\gamma_{\A} =
\beta_{4}\beta_{3}\beta_{4}\beta_{2}\beta_{4}\beta_{3}\beta_{4}$.

\begin{observation}
\label{obs:modTwoA}
For a convexed $l$-collection $\A$ and an integer $m$, $|\modTwo{m \times \A}| = 0$ if either $m$ is even or $\A = \emptyset$, and otherwise $|\modTwo{m \times \A}| = 1$.
\end{observation}

\begin{claim}
\label{clm:subConvex}
Let $\A = \{\alpha_1, \ldots, 2^{j-1} \times \alpha_j, \ldots, 2^{r-1} \times \alpha_r\}$ and $\A' = \{\alpha_1, \ldots, 2^{j-1} \times \alpha_j\}$ be two convexed $l$-collections (where $\A' \subseteq \A$, and it is possible that $\A' = \emptyset$).
For every number $t$, there is an integer $x \geq 0$ and a convexed $l$-collection $\hat{\A}$
such that $\lt[(\A - \A')]{t} = 2^x \times \hat{\A}$. In addition, if $\A' \neq \emptyset$ then $x > 0$ and $|\hat{\A}| < |\A|$
\end{claim}

\begin{proof}
First, note that $\A - \A' = \{2^j \times \alpha_{j+1}, \ldots, 2^{r-1} \times \alpha_r\}$. Now, let $x = j$ if $\tp[\alpha_{j+1}] < t$, and otherwise let $x$ be the maximum integer in the range $j < x \leq r$ such that $\tp[\alpha_{x}] \geq t$. Then, $\lt[(\A-\A')]{t} = \{2^{x} \times \alpha_{x+1}, \ldots, 2^{r-1} \times \alpha_r\} = 2^{x} \times \{\alpha_{x}, \ldots, 2^{r-x-1} \times \alpha_r\}$. Choosing $\hat{\A} = \{\alpha_{x}, \ldots, 2^{r-x-1} \times \alpha_r\}$, the claim follows.
\end{proof}

\begin{definition}
\label{def:decomposition}
Let $\B = \{\muli{1} \times \beta_1, \muli{2} \times \beta_2, \ldots,
\muli{q} \times \beta_q\}$ be an $l$-BFB palindrome collection. The
\emph{decomposition} of $\B$ is a series triplet $\left\langle \Bs, \Bls,
  \Bhs\right\rangle$, whose elements are recursively
defined as follows:
\begin{itemize}
	\item $\B_0 = \B$, and $\B_d = \frac{1}{2} \times \left(\B_{d-1} - \Bl_{d-1} - \Bh_{d-1}\right)$ for $d>0$.
	\item $\Bl_d = \modTwo{\B_d}$.
	\item $\Bh_d = \geqt[(\B_d - \Bl_d)]{\mint(\Bl_d)}$.
\end{itemize}

Denote by $r(\B)$ the minimum integer $r$ such that
$\B_r = \emptyset$.
\end{definition}

\ifthenelse{\boolean{showFigs}}{
\begin{table}
  \caption{{\normalsize
  The decomposition and signature of the collection $\B = \{2 \times \beta_1, \beta_2, 2 \times \beta_3, 4 \times \beta_{4}\}$ appearing in Fig. 2a. Here, $r(\B) = 3$.
}}
\begin{minipage}[H]{.40\textwidth}
\begin{tabular*}{\hsize}{@{\extracolsep{\fill}}c|cccc}
$d$	& $\B_d$  & $\Bl_d$ & $\Bh_d$  & $s_d$ \\ \hline
$0$ 	&
$\{2 \times \beta_{1}, \beta_{2}, 2 \times \beta_{3}, 4 \times \beta_{4}\}$ &
$\beta_{2}$ & $2 \times \beta_{1}$ & $1$\\
$1$		&
$\{\beta_{3}, 2 \times \beta_{4}\}$ &
$\beta_{3}$ & $\emptyset$ & $0$\\
$2$		&
$\beta_{4}$ & $\beta_{4}$ & $\emptyset$ & $0$\\
$3$ & $\emptyset$ & $\emptyset$ & $\emptyset$ &  $-1$ \\
$4$ & $\emptyset$ & $\emptyset$ & $\emptyset$ &  $0$ \\
$\vdots$ & $\vdots$ & $\vdots$ & $\vdots$ & $\vdots$
\end{tabular*}
\label{tab:decomposition}
 \end{minipage}
\end{table}

}{}
Table~\ref{tab:decomposition} gives the decomposition of the collection $\B^1$ corresponding to Fig. 2a in the main manuscript.
In what follows, let $\B$ be a collection, $\left\langle \Bs, \Bls, \Bhs\right\rangle$
its decomposition, and $r = r(\B)$.

\begin{definition}
\label{def:t}
For a collection $\B$, $\ser{t}(\B) = \ser{t}$ is the non-decreasing series of numbers whose elements are given by
$t_0 = \infty$, and $t_d = \min(\mint(\Bl_d), t_{d-1}))$ for $d > 0$.
\end{definition}

The following observation may be easily asserted, in an inductive manner.

\begin{observation}
\label{obs:t}
For a collection $\B$ and every integer $d\geq0$, $\Bh_d = \geqt[(\B_d - \Bl_d)]{t_{d+1}}$, $\lt{t_{d}} = 2^d \times \B_d$,
 and $\rng{t_{d+1}}{t_{d}} = 2^d \times (\Bl_d + \Bh_d)$.
\end{observation}

Finally, we define the \emph{signature} of a collection, which is derived from its decomposition and will serve as an optimality measure implying the folding restrictions over the collection.
\begin{definition}
\label{def:signature}
The \emph{signature} of $\B$ is a series $\sigs = \sigs(\B)$, where $s_0 = |\Bl_0|$, and $s_d = |\Bl_d| - |\Bl_{d-1}| - \frac{|\Bh_{d-1}|}{2} + \max (s_{d-1}, 0)$ for $d > 0$.
\end{definition}

The last column of Table~\ref{tab:decomposition} shows the signature of the exemplified collection.
For two signatures $\sigs = s_0, s_1, \ldots$ and $\sert{s} = s'_0, s'_1, \ldots$, denote $\sigs < \sert{s}$ if $\sigs$ precedes $\sert{s}$ lexicographically, i.e. there is some integer $d \geq 0$ such that $\sig_{i} = \sig'_{i}$ for every $0 \leq i < d$, and $\sig_d < \sig'_d$. Denote $\sigs \leq \sert{s}$ if $\sigs < \sert{s}$ or $\sigs = \sert{s}$.
We will show that signatures can serve as an optimality measure for collections, where lower signature collections are always less restricted than higher signature collection with respect to folding possibilities.

From now on, when discussing derived entities such a decompositions $\left\langle \Bs, \Bls, \Bhs\right\rangle$, signatures $\sigs$, etc., we assume these entities correspond to the collection $\B$ discussed in the same context without stating so explicitly. When several collections are considered, these collections are annotated with superscripts (e.g. $\B', \B^*, \B^3$, etc.), which also annotate their correspondingly derived entities (e.g. $\Bl'_d, \sera{s}{3}$, etc.).

\begin{claim}
\label{clm:sL}
For every $d \geq 0$, $|\Bl_d| \geq \max(s_d, 0)$ and $|\B_d| + |\Bl_d| - s_d \geq \max(s_d, 0)$.
\end{claim}

\begin{proof}
We first show the first inequality in the Claim.
For $d = 0$, $s_0 = |\Bl_0| \geq 0$ by definition. Assume by induction $|\Bl_{d'}| \geq \max(s_{d'}, 0)$ for every $0 \leq d'< d$. Then, $|\Bl_{d}| = s_{d} + |\Bl_{d-1}| + \frac{|\Bh_{d-1}|}{2} - \max(s_{d-1}, 0) \geq s_{d} + \frac{|\Bh_{d-1}|}{2} \geq s_{d}$. In addition, $|\Bl_d| \geq 0$, and so $|\Bl_d| \geq \max(s_d, 0)$.
The second inequality follows immediately from the first one, as $|\B_d| + |\Bl_d| - s_d \geq |\B_d| \geq |\Bl_d| \geq \max(s_d, 0)$.
\end{proof}

\begin{claim}
\label{clm:r}
For $r = r(\B)$, $s_r  = -\frac{|\B_{r-1}| + |\Bl_{r-1}|}{2} + \max(s_{r-1}, 0) \leq 0$, and $s_d = 0$ for every $d > r$.
\end{claim}

\begin{proof}
The inequality $s_r \leq 0$ follows immediately from Claim~\ref{clm:sL} and the fact that $|\Bl_r| = 0$. In addition, since $\B_r = \frac{1}{2} \times (\B_{r-1} - \Bl_{r-1} - \Bh_{r-1}) = \emptyset$, we have that $\Bh_{r-1} = \B_{r-1} - \Bl_{r-1}$, and therefore $s_r = |\Bl_r| - |\Bl_{r-1}| - \frac{|\Bh_{r-1}|}{2} + \max(s_{r-1}, 0) = - \frac{|\B_{r-1}| + |\Bl_{r-1}|}{2} + \max(s_{r-1}, 0)$.

To show the second part of the claim, note that $|\B_d| = |\Bl_d| = |\Bh_d| = 0$ for every $d \geq r$. This implies that for every $d > r$, $s_d = |\Bl_d| - |\Bl_{d-1}| - \frac{|\Bh_{d-1}|}{2} + \max(s_{d-1}, 0) = \max(s_{d-1}, 0)$. Since we showed that $s_r \leq 0$, we have that $s_{r+1} = \max(s_r, 0) = 0$, and inductively it follows that that $s_d = 0$ for every $d > r$.
\end{proof}

Define the series $\ser{\Delta} = \ser{\Delta}(\B)$, where $\Delta_0 = 0$, and $\Delta_d = \Delta_{d-1} + 2^{d-1} \abs(s_{d-1})$ for $d > 0$ (where $\abs(s_{d-1})$ is the absolute value of $s_{d-1}$).

\begin{claim}
\label{clm:main}
For every integer $d \geq 0$,

$$
|\B| = 2^d\left(|\B_{d}| + |\Bl_d| - s_d\right) + \Delta_d \geq 2^d\max(s_d, 0) + \Delta_d.
$$
\end{claim}

\begin{proof}
The fact that $2^d\left(|\B_{d}| + |\Bl_d| - s_d\right) + \Delta_d \geq 2^d\max(s_d, 0) + \Delta_d$ follows from Claim~\ref{clm:sL}. The equality $|\B| = 2^d\left(|\B_{d}| + |\Bl_d| - s_d\right) + \Delta_d$ is proven by induction on $d$. For $d = 0$, since $\B_0 = \B$, $s_0 = |\Bl_0|$, and $\Delta_0 = 0$ by definition, we get that $2^0\left(|\B_{0}| + |\Bl_0| - s_0\right) + \Delta_d = |\B|$. Now, assuming the claim holds for some $d \geq 0$, we show it also holds for $d' = d+1$:

$$
\begin{array}{rl}
|\B| = & \displaystyle{2^{d}\left(|\B_{d}| + |\Bl_{d}| - s_d\right) + \Delta_d} \\
= & \displaystyle{2^{d}\left(|\B_{d}| + |\Bl_{d}| - s_d - \abs(s_d)\right) + \Delta_{d+1}} \\
= & \displaystyle{2^{d}\left(|\B_{d}| + |\Bl_{d}| - 2\max(s_d, 0)\right) + \Delta_{d+1}} \\
= & \displaystyle{2^{d}\left((2|\B_{d+1}| + |\Bl_d| + |\Bh_d|) + |\Bl_{d}| - 2\max(s_d, 0)\right) + \Delta_{d+1}} \\
= & \displaystyle{2^{d}\left(2|\B_{d+1}| + 2|\Bl_{d+1}| - 2s_{d+1}\right) + \Delta_{d+1}} \\
 = & \displaystyle{2^{d+1}\left(|\B_{d+1}| + |\Bl_{d+1}| - s_{d+1}\right) + \Delta_{d+1}}.
\end{array}
$$

\end{proof}

\begin{conclusion}
\label{conc:r}
For every $d \geq r = r(\B)$ we have that $|\B_d| = |\Bl_d| = 0$, where from Claim~\ref{clm:r} $s_{d} = 0$ for $d > r$. Thus, Claim~\ref{clm:main} implies that $|\B| = -2^rs_r + \Delta_{r}  = \Delta_d$ for every $d > r$.
\end{conclusion}

\begin{claim}
\label{clm:eqDeltar}
Let $\B$ and $\B'$ be two collections such that $|\B| = |\B'| = n$ and $\sigs_{r-1} = \sert{s}_{r-1}$ for $r = r(\B)$. Then, $\sigs \leq \sert{s}$.
\end{claim}

\begin{proof}
Since $\sigs_{r-1} = \sert{s}_{r-1}$, it follows that $\Delta_{r} = \Delta'_{r}$. From Conclusion~\ref{conc:r}, $s_r = \frac{\Delta_r - n}{2^r}$. From Claim~\ref{clm:main}, $s'_r = \frac{\Delta'_r - n}{2^r} + |\B'_r| + |\Bl'_r| \geq \frac{\Delta_r - n}{2^r} = s_r$. If $s'_r > s_r$, then $\sigs < \sert{s}$ and the claim holds. If $s'_r = s_r$, then in particular $|\B'_r| = 0$ and so $\B'_r = \emptyset$. Thus, $r(\B') \leq r$, and so for every $d > r$ we have that $s'_d = s_d = 0$, and $\sert{s} = \sigs$.
\end{proof}

\subsection{Folding Increases Signature}
\label{sec:FoldingIncreasesSignature}

This section is dedicated for proving the following claim:

\begin{claim}
\label{clm:foldIncreaseSignature}
Let $\B'$ be a folding of an $l$-block collection $\B$. If $\B \neq \B'$, then $\sigs < \sert{s}$.
\end{claim}

The proof, given at the end of this section, is based on an observation that shows how to present a general folding as a series of a special kind of elementary foldings, and showing that such elementary foldings always increase the signature of the collection.

\begin{definition}
\label{def:elementary}
Let $\B'$ be a folding of $\B$.
\begin{itemize}
	\item Say that $\B'$ is a \emph{type I elementary folding} if $\B'$ is of the form $\B' = \B - m \times (2 \times \beta + \A) + m \times \alpha$, where $\beta$ is an $l$-block, $\A \neq \emptyset$ is a convexed $l$-collection, $m > 0$ is an integer such that $m \times (2 \times \beta + \A) \subseteq \B$, and $\alpha = \beta \gamma_{\A} \beta$ is an $l$-BFB palindrome such that $\alpha \notin \B$.
	\item Say that $\B'$ is a \emph{type II elementary folding} if $\varepsilon \notin \B$ and $\B'$ is of the form $\B' = \B + m \times \varepsilon$.
\end{itemize}
\end{definition}

\begin{claim}
\label{clm:elementarySeq}
Let $\B$ be a collection of $l$-blocks.
For every folding $\B'$ of $\B$
there is a sequence of collections $\B^0, \B^{1}, \ldots, \B^j$, where $\B^0 = \B'$, $\B^j = \B$, and for every $0 \leq i < j$, $\B^{i}$ is a (type I or II) elementary folding of $\B^{i+1}$.
\end{claim}

\begin{proof}
By definition, each element in a folding $\B'$ of $\B$ is either an $l$-block from $\B$, a concatenation of several $l$-blocks from $\B$,
or $\varepsilon$.
The sequence $\B^0, \B^{1}, \ldots, \B^j$ is built iteratively as follows.

Initiate $\B^0 = \B'$, and $i = 0$. As long as $\B^i \neq \B$, we show how to compute $\B^{i+1}$ given the collection $\B^{i}$.
The construction maintains the property that each computed collection $\B^i$ is a folding of $\B$.
In the case where $\B^{i}$ contains some composite $l$-BFB palindrome of the form $\alpha = \beta \gamma_{\A} \beta$, let $m$ be the count of $\alpha$ in $\B^{i}$, and set $\B^{i+1} = \B^{i} - m \times \alpha + m \times \left(2 \times \beta + \A\right)$. We may assume $\A \neq \emptyset$, since when $\gamma_{\A} = \varepsilon$ we can choose $\A = \varepsilon$. Observe that $\B^{i+1}$ is a folding of $\B$ (where the same sub-collection of $l$-blocks from $\B$ which composes the $m$ copies of $\alpha$ in $\B^{i}$, composes the elements in the $m$ repeats of $2 \times \beta + \A$ in $\B^{i+1}$), and that
$\B^{i}$ is a type I elementary folding of $\B^{i+1}$. In addition, since the number of $l$-blocks composing each element in $\A$ is less than the number of $l$-blocks composing $\alpha$, after a finite number of such modification there will be no more composite palindromes in the collection.

In the case where $\B^{i}$ contains no composite palindrome, $\B^{i}$ is a folding of $\B$ containing only $l$-blocks and 0 or more $\varepsilon$ elements. If $\B^{i}$ contains no $\varepsilon$ elements, then $\B^{i} = \B$, and the process is completed choosing $j = i$. Else, for $m$ the count of $\varepsilon$ in $\B^{i}$, set $\B^{i+1} = \B^{i} - m \times \varepsilon$, and therefore $\B^{i}$ is a type II elementary folding of $\B_{i+1}$. Note that $\B^{i+1} = \B$, completing the process for $j = i+1$.
\end{proof}

Observe that the signature of a collection depends only in its decomposition, and is independent in the manner the collection was obtained. Therefore, from the above claim, in order to show that foldings necessarily increase signatures, it is enough to show that each elementary folding increases the signature.
In what follows, we give several technical claims that will prove this property.

\begin{claim}
\label{clm:addition}
Let $\B, \B'$, and $\A$ be $l$-BFB palindrome collections and $m > 0$ an integer such that $\B' = \B + m \times \A$.
Then,
\begin{enumerate}
	\item For every $0 \leq i \leq \dig{m}$, $\B'_i = \B_i + \frac{m}{2^{i}} \times \lt[\A]{t_{i}}$.
	\item For every $0 \leq i < \dig{m}$, $\Bl'_{i} = \Bl_{i}$ (i.e. $\Bls'_{\dig{m}-1} = \Bls_{\dig{m}-1}$).
	\item For every $0 \leq i < \dig{m}$, $\Bh'_i = \Bh_i + \frac{m}{2^{i}} \times \rng[\A]{t_{i+1}}{t_{i}}$.
\end{enumerate}
\end{claim}

\begin{proof}
$\B_0 = \B$ and $t_{0} = \infty$, therefore $\B'_0 = \B' = \B + m \times \A = \B_0 + \frac{m}{2^{0}} \times \lt[\A]{t_{0}}$. Thus, the first item in the claim holds for $i = 0$, and the two other items hold trivially for every $0 \leq i < 0$.

Assuming by induction that for some $i < \dig{m}$ the first item holds for every $0 \leq i' \leq i$ and the two other items hold for every $0 \leq i' < i$, we show that (1) $\B'_{i+1} = \B_{i+1} + \frac{m}{2^{i+1}} \times \lt[\A]{t_{i+1}}$, (2) $\Bl'_{i} = \Bl_{i}$, and (3) $\Bh'_{i} = \Bh_{i} + \frac{m}{2^{i}} \times \rng[\A]{t_{i+1}}{t_{i}}$.

We start by showing (2). Since $i < \dig{m}$, $\frac{m}{2^{i}}$ is even, and so $\Bl'_{i} = \modTwo{\B'_{i}} = \modTwo{\B_i + \frac{m}{2^{i}} \times \lt[\A]{t_{i}}} \stackrel{\text{Obs.\ref{obs:modTwo}}}{=} \modTwo{\B_{i}} = \Bl_i$.
In order to show (3), note that $\sert{\Bl}_i = \Bls_i$ implies that $t'_{i+1} = t_{i+1}$, therefore
$\Bh'_{i} \stackrel{\text{Obs.\ref{obs:t}}}{=} \geqt[(\B'_{i} - \Bl'_{i})]{t'_{i+1}} = \geqt[{(\B_i + \frac{m}{2^{i}} \times \lt[\A]{t_{i}} - \Bl_{i})}]{t_{i+1}} = \geqt[(\B_{i}  - \Bl_{i})]{t_{i+1}} + \frac{m}{2^{i}} \times \rng[\A]{t_{i+1}}{t_i} \stackrel{\text{Obs.\ref{obs:t}}}{=} \Bh_{i} + \frac{m}{2^{i}} \times \rng[\A]{t_{i+1}}{t_i}$.
Finally, (1) is true since $\B'_{i+1} = \frac{1}{2} \times (\B'_{i} - \Bl'_{i} - \Bh'_{i}) = \frac{1}{2} \times (\B_i + \frac{m}{2^{i}} \times \lt[\A]{t_{i}} - \Bl_{i} - \Bh_{i} - \frac{m}{2^{i}} \times \rng[\A]{t_{i+1}}{t_i}) = \frac{1}{2} \times (\B_i - \Bl_{i} - \Bh_{i}) + \frac{m}{2^{i+1}} \times (\lt[\A]{t_{i}} - \rng[\A]{t_{i+1}}{t_i})  = \B_{i+1} + \frac{m}{2^{i+1}} \times \lt[\A]{t_{i+1}}$.
\end{proof}

\begin{claim}
\label{clm:elementaryDecomposition}
Let $\B' = \B - m \times (2 \times \beta + \A) + m \times \alpha$ be a type I elementary folding of $\B$. Then,

\begin{equation}
		\forall 0 \leq i \leq \dig{m}, \ \B'_i = 
		 \B_i - \frac{m}{2^i} \times \lt[{(2 \times \beta + \A)}]{t_{i}}
		 + \frac{m}{2^i} \times \lt[\alpha]{t_{i}},
	\label{eq:elem}
\end{equation}

\begin{equation}
		\Bls'_{\dig{m}-1} = \Bls_{\dig{m}-1}
	\label{eq:elemL}
\end{equation}

\begin{equation}
		\forall 0 \leq i < \dig{m}, \ \Bh'_i = 
		 \Bh_i - \frac{m}{2^i} \times \rng[{(2 \times \beta + \A)}]{t_{i+1}}{t_{i}} +
		 \frac{m}{2^i} \times \rng[\alpha]{t_{i+1}}{t_{i}}.
	\label{eq:elemH}
\end{equation}

\end{claim}

\begin{proof}
Let $\B'' = \B - m \times(2 \times \beta + \A)$. Therefore, $\B = \B'' + m \times(2 \times \beta + \A)$, and $\B' = \B'' + m \times \alpha$.
From Claim~\ref{clm:addition},
\begin{enumerate}
	\item For every $0 \leq i \leq \dig{m}$, $\B_i = \B''_i + \frac{m}{2^{i}} \times \lt[{(2 \times \beta + \A)}]{t_{i}}$, and $\B'_i = \B''_i + \frac{m}{2^{i}} \times \lt[\alpha]{t_{i}}$. Therefore, $\B'_i = \B_i - \frac{m}{2^{i}} \times \lt[{(2 \times \beta + \A)}]{t_{i}} + \frac{m}{2^{i}} \times \lt[\alpha]{t_{i}}$.
	\item $\Bls'_{\dig{m}-1} = \Bls_{\dig{m}-1} = \Bls''_{\dig{m}-1}$.
	\item For every $0 \leq i < \dig{m}$, $\Bh_i = \Bh''_i + \frac{m}{2^i} \times \rng[{(2 \times \beta + \A)}]{t_{i+1}}{t_{i}}$, and $\Bh'_i = \Bh''_i +  \frac{m}{2^i} \times \rng[\alpha]{t_{i+1}}{t_{i}}$. Therefore, $\Bh'_i = \Bh_i - \frac{m}{2^i} \times \rng[{(2 \times \beta + \A)}]{t_{i+1}}{t_{i}} +  \frac{m}{2^i} \times \rng[\alpha]{t_{i+1}}{t_{i}}$.
\end{enumerate}
\end{proof}

\begin{claim}
\label{clm:elementaryAux}
Let $\B, \B'$, and $C$ be $l$-BFB palindrome collections, $\A$ a convexed $l$-collection, and
$m$ and $i$ two nonnegative integers, such that
\emph{(a)}	$\sert{s}_{i} = \sigs_{i}$, \emph{(b)} $t'_{i+1} \geq t_{i+1}$, \emph{(c)} $|\Bl'_i| + \frac{|\Bh'_i|}{2} < |\Bl_i| + \frac{|\Bh_i|}{2} - |C|$, \emph{(d)} $\B_{i+1} \cap C = \emptyset$, and \emph{(e)} $\B'_{i+1} = \B_{i+1} + C - m \times \A$.
Then, $\B <^s \B'$.
\end{claim}

\begin{proof}
We prove the claim by induction on the size of $\A$. Let $\A = \{\alpha_1, \ldots, 2^{r-1} \times \alpha_r\}$, and denote $\tilde{A} = \modTwo{m \times \A}$. Observe that $\tilde{A} = \alpha_1$ when $\A \neq \emptyset$ and $m$ is odd, and otherwise $\tilde{A} = \emptyset$. In addition, observe that $\modTwo{\B_{i+1} + C} \stackrel{\text{Obs.\ref{obs:modTwo}}}{=} \modTwo{\B_{i+1}} + \modTwo{C} - 2 \times \left(\modTwo{\B_{i+1}} \cap \modTwo{C}\right) \stackrel{\text{(d)}}{=} \modTwo{\B_{i+1}} + \modTwo{C} = \Bl_{i+1} + \modTwo{C}$. Therefore,

$$
\begin{array}{rl}
	\Bl'_{i+1} = & \modTwo{\B'_{i+1}} \stackrel{\text{(e)}}{=} \modTwo{\B_{i+1} + C - m \times \A} \\
	 \stackrel{\text{Obs.\ref{obs:modTwo}}}{=} & \modTwo{\B_{i+1} + C} + \modTwo{m \times \A} - 2 \times \left(\modTwo{\B_{i+1} + C} \cap \modTwo{m \times \A}\right)\\
	= & \Bl_{i+1} + \modTwo{C} + \tilde{A} - 2 \times \left((\Bl_{i+1} + \modTwo{C}) \cap \tilde{A}\right) \\
\end{array}
$$

Since $(\Bl_{i+1} + \modTwo{C}) \cap \tilde{A} \subseteq \tilde{A}$, it follows that $|\Bl'_{i+1}| \geq |\Bl_{i+1}| + |\modTwo{C}| + |\tilde{A}| - 2|\tilde{A}| = |\Bl_{i+1}| + |\modTwo{C}| - |\tilde{A}|$.
Therefore, $s'_{i+1} = |\Bl'_{i+1}| - |\Bl'_{i}| - \frac{|\Bh'_i|}{2} + \max(s'_{i}, 0) \stackrel{\text{(a),(c)}}{>} |\Bl_{i+1}| + |\modTwo{C}| - |\tilde{A}| - |\Bl_{i}| - \frac{|\Bh_i|}{2} + |C| + \max(s_{i}, 0) = s_{i+1} + |\modTwo{C}| + |C| - |\tilde{A}|$. As both $s'_{i+1}$ and $s_{i+1}$ are integers, $s'_{i+1} \geq s_{i+1} + |\modTwo{C}| + |C| - |\tilde{A}| + 1 \geq s_{i+1}$ (recall that $|\tilde{A}| \leq 1$).
 When $s'_{i+1} > s_{i+1}$, $\sigs < \sert{s}$ and the claim follows. Otherwise, $s'_{i+1} = s_{i+1}$, and there is a need to continue and examine positions grater than $i+1$ in the signatures of $\B$ and $\B'$.

Note that for obtaining $s'_{i+1} = s_{i+1}$ we must have that $C = \emptyset$ and $\tilde{\A} = \Bl_{i+1} \cap \tilde{\A} = \alpha_1$, which implies that $\A \neq \emptyset$, $m$ is odd, and $\Bl'_{i+1} = \Bl_{i+1} + \alpha_1 - 2 \times \alpha_1 =\Bl_{i+1} - \alpha_1$ (thus $\alpha_1 \in \Bl_{i+1}$, and in particular $\tp[\alpha_1] \geq \mint(\Bl_{i+1}) \geq t_{i+2}$).
Assuming by induction that the claim holds for every $\B'', \B''', C', \A', i'$, and $m'$ sustaining requirements (a) to (e) and $|\A'| < |\A|$, we show the claim also holds for $\B, \B', C, \A, i$, and $m$.

Now, since $\sert{s}_{i} \stackrel{\text{(a)}}{=} \sigs_{i}$ and $s'_{i+1} = s_{i+1}$, requirement (a) in the claim holds with respect to $i' = i+1$. In addition,
$$
\begin{array}{rl}
	t'_{i+2} = & \min(\mint(\Bl'_{i+1}), t'_{i+1}) \geq \min(\mint(\Bl_{i+1} - \alpha_1), t_{i+1}) \\
	\geq & \min(\mint(\Bl_{i+1}), t_{i+1}) = t_{i+2},
\end{array}
$$

\noindent
thus requirement (b) also holds with respect to $i' = i+1$. Furthermore,

$$
\begin{array}{rl}
	\Bh'_{i+1} \stackrel{\text{Obs.\ref{obs:t}}}{=} & \geqt[(\B'_{i+1} - \Bl'_{i+1})]{t'_{i+2}}
	=  \geqt[(\B_{i+1} - m \times \A - \Bl_{i+1} + \alpha_1)]{t'_{i+2}} \\
	= & \geqt[(\B_{i+1} - \Bl_{i+1})]{t'_{i+2}} - m \times \geqt[\A]{t'_{i+2}} + \geqt[\alpha_1]{t'_{i+2}} \\
	= & \left(\geqt[(\B_{i+1} - \Bl_{i+1})]{t_{i+2}} - \rng[(\B_{i+1} - \Bl_{i+1})]{t_{i+2}}{t'_{i+2}}\right) - m \times \geqt[\A]{t'_{i+2}} + \geqt[\alpha_1]{t'_{i+2}} \\
	= & \Bh_{i+1} - \rng[(\B_{i+1} - \Bl_{i+1})]{t_{i+2}}{t'_{i+2}} - m \times \geqt[\A]{t'_{i+2}} + \geqt[\alpha_1]{t'_{i+2}}.
\end{array}
$$

Since $\alpha_1 \in \A$, the operation $\A - \alpha_1$ yields a valid collection. In addition, note that the count of each element in $\B_{i+1} - \Bl_{i+1}$ is even, and
since $m$ is odd, $\B_{i+1}$ contains at least $m$ copies of $\alpha_1$ (as $m \times \A \subseteq \B_{i+1}$) and $\Bl_{i+1}$ contains exactly one copy of $\alpha_1$, $\B_{i+1} - \Bl_{i+1} - (m-1) \times \alpha_1$ is a valid collection, in which the count of each element is even. Denote $\hat{C} = \frac{1}{2} \times \left(\rng[(\B_{i+1} - \Bl_{i+1} - (m-1) \times \alpha_1)]{t_{i+2}}{t'_{i+2}}\right) =  \frac{1}{2}\left(\rng[(\B_{i+1} - \Bl_{i+1})]{t_{i+2}}{t'_{i+2}} - (m-1) \times \lt[\alpha_1]{t'_{i+2}}\right)$. Now we can write

$$
\begin{array}{rl}
	\Bh'_{i+1} = & \Bh_{i+1} - \rng[(\B_{i+1} - \Bl_{i+1})]{t_{i+2}}{t'_{i+2}} - m \times \geqt[\A]{t'_{i+2}} + \geqt[\alpha_1]{t'_{i+2}} \\
	= & \Bh_{i+1} - 2 \times \hat{C} - (m-1) \times \lt[\alpha_1]{t'_{i+2}} - m \times \geqt[\A]{t'_{i+2}} + \geqt[\alpha_1]{t'_{i+2}} \\
	= & \Bh_{i+1} - 2 \times \hat{C} - (m-1) \times \left(\alpha_1 - \geqt[\alpha_1]{t'_{i+2}}\right) - m \times \geqt[\A]{t'_{i+2}} + \geqt[\alpha_1]{t'_{i+2}} \\
	= & \Bh_{i+1} - 2 \times \hat{C} - (m-1) \times \alpha_1 - m \times \geqt[(\A - \alpha_1)]{t'_{i+2}} \\
\end{array}
$$

From the above $|\Bh'_{i+1}| \leq |\Bh_{i+1}| - 2 |\hat{C}|$, and since $|\Bl'_{i+1}| = |\Bl_{i+1}| - 1$ we get that $|\Bl'_{i+1}| + \frac{|\Bh'_{i+1}|}{2} \leq |\Bl'_{i+1}| - 1 + \frac{|\Bh'_{i+1}|}{2} + |\hat{C}| < |\Bl'_{i+1}| + \frac{|\Bh'_{i+1}|}{2} + |\hat{C}|$, and in particular requirement (c) holds with respect to $C' = \hat{C}$, and $i' = i+1$. Moreover, by definition the top values of all elements in $\hat{C}$ are at least $t_{i+2}$, and from Observation~\ref{obs:t}, $\B_{i+2} = \frac{1}{2^{i+2}} \times \B^{< t_{i+2}}$, hence the top values of all elements in $\B_{i+2}$ are lower than $t_{i+2}$. Thus, $\B_{i+2} \cap \hat{C} = \emptyset$, and requirement (d) holds with respect to $C' = \hat{C}$, and $i' = i+1$.

From Claim~\ref{clm:subConvex}, there is an integer $x > 0$ and a convex $l$-collection $\hat{\A}$ such that $|\hat{\A}| < |\A|$ and $\lt[(\A - \alpha_1)]{t'_{i+2}} = 2^x \times \hat{\A}$. Therefore,
$$
\begin{array}{rl}
	\Bh'_{i+1} 	= & \Bh_{i+1} - 2 \times \hat{C} - (m-1) \times \alpha_1 - m \times \geqt[(\A - \alpha_1)]{t'_{i+2}} \\
	= & \Bh_{i+1} - 2 \times \hat{C} - (m-1) \times \alpha_1 - m \times \left((\A - \alpha_1) - \lt[(\A - \alpha_1)]{t'_{i+2}}\right) \\
	= & \Bh_{i+1} - 2 \times \hat{C} - m \times \A + \alpha_1 + m \times \lt[(\A - \alpha_1)]{t'_{i+2}} \\
	= & \Bh_{i+1} - 2 \times \hat{C} - m \times \A + \alpha_1 + m 2^x \times \hat{\A}, \\
\end{array}
$$

\noindent
and

$$
\begin{array}{rl}
	\B'_{i+2} = & \frac{1}{2} (\B'_{i+1} - \Bl'_{i+1} - \Bh'_{i+1})\\
	= & \frac{1}{2} ((\B_{i+1} - m \times \A) - (\Bl_{i+1} - \alpha_1)
	 - (\Bh_{i+1} - 2 \times \hat{C} - m \times \A + \alpha_1 + m 2^x \times \hat{\A})) \\
	= & \frac{1}{2} (\B_{i+1} - \Bl_{i+1} - \Bh_{i+1}) + \hat{C} - m 2^{x-1} \times \hat{\A} \\
	= & \B_{i+2} + \hat{C} - m 2^{x-1} \times \hat{\A}. \\
\end{array}
$$

Since $x > 0$, $m 2^{x-1}$ is an integer. Therefore, requirement (e) holds with respect to
$C' = \hat{C}, \A' = \hat{\A}, \  i' = i+1$, and $m' = m 2^{x-1}$. From the inductive assumption and the fact that $|\hat{\A}| < |\A|$, the claim follows.
\end{proof}

\begin{claim}
\label{clm:typeI}
Let $\B'$ be a type I elementary folding of $\B$. Then, $\B <^s \B'$
\end{claim}

\begin{proof}
By definition, $\B'$ is of the form $\B' = \B - m \times (2 \times \beta + \A) + m \times \alpha$, where $m > 0$ is an integer, $\alpha$ and $\beta$ are $l$-BFB palindromes, and $\A = \{\alpha_1, \ldots, 2^{r-1} \times \alpha_{r}\}$ is a convexed $l$-collection such that $\alpha = \beta \gamma_{\A} \beta$.
Let $q \geq 0$ be the index such that $\beta \in \Bl_q + \Bh_q$. From Observation~\ref{obs:t}, $t_{q+1} \leq \tp[\beta] < t_{q}$, and therefore for every $0 \leq i \leq q$ and every $\alpha_j \in \A$ we have that
$(\star) \ \tp[\alpha_j] \leq \tp[\alpha] = \tp[\beta] < t_q \leq t_{i}$. Let $d = \dig{m}$. We consider two cases: (1) $q < d$, and (2) $q \geq d$, and show for each case that $\B <^s \B'$.

\vspace{10pt}
\noindent
(1) $q < d$. In this case, condition $(\star)$ implies that for every $0 \leq i < q$, we have that $\rng[{(2 \times \beta + \A)}]{t_{i+1}}{t_{i}} = \rng[\alpha]{t_{i+1}}{t_{i}} = \emptyset$. In addition $\rng[{(2 \times \beta + \A)}]{t_{q+1}}{t_{q}} = 2 \times \beta + \geqt[\A]{t_{q+1}}, \ \rng[\alpha]{t_{q+1}}{t_{q}} = \alpha$, and $\lt[{(2 \times \beta + \A)}]{t_{q+1}} = \lt[\A]{t_{q+1}}, \ \lt[\alpha]{t_{q+1}} = \emptyset$. Thus, form Claim~\ref{clm:elementaryDecomposition}, we get that

$$
\begin{array}{ll}
	\B'_{q+1} & = \B_{q+1} - \frac{m}{2^{q+1}} \times \lt[\A]{t_{q+1}}
	 \stackrel{\text{Clm.\ref{clm:subConvex}}}{=} \B_{q+1} - \frac{m2^x}{2^{q+1}} \hat{\A}, \\
	 & \text{ where } \frac{m2^x}{2^{q+1}} \text{ is a positive integer and } \hat{\A} \text{  is a convexed } l\text{-collection,}\\
	\Bls'_{q} & = \Bls_{q}, \\
	\Bhs'_{q-1} & = \Bhs_{q-1}, \\
	\Bh'_q & = \Bh_q - \frac{m}{2^q} \times (2 \times \beta  + \geqt[\A]{t_{q+1}}) + \frac{m}{2^q} \times \alpha.
\end{array}
$$

Observe that $\Bls'_{q}  = \Bls_{q}$ and $ \Bhs'_{q-1} = \Bhs_{q-1}$ imply that $\sert{s}_q = \sigs_q$ and $\sert{t}_{q+1} = \ser{t}_{q+1}$. Also, observe that $|\Bh'_q| \leq |\Bh_q| - \frac{m}{2^q} < |\Bh_q|$, therefore $|\Bl'_q| + \frac{|\Bh'_q|}{2} < |\Bl_q| + \frac{|\Bh_q|}{2}$. Applying Claim~\ref{clm:elementaryAux} with respect to entities $\B, \B', C = \emptyset, \hat{\A}, i = q$, and $m' = \frac{m2^x}{2^{q+1}}$, we get that $\B <^s \B'$.

\vspace{10pt}
\noindent
(2) $q \geq d$. In this case, condition $(\star)$ implies that for every $0 \leq i < d$, we have that $\rng[{(2 \times \beta + \A)}]{t_{i+1}}{t_{i}} = \rng[\alpha]{t_{i+1}}{t_{i}} = \emptyset$, and $\lt[{(2 \times \beta + \A)}]{t_d} = 2 \times \beta + \A, \ \lt[\alpha]{t_{d}} = \alpha$. Therefore, from Claim~\ref{clm:elementaryDecomposition},

$$
\begin{array}{ll}
\B'_{d} & = \B_{d} - \frac{m}{2^{d}} \times (2 \times \beta + \A) + \frac{m}{2^{d}} \times \alpha,\\
\Bls'_{d-1} & =  \Bls_{d-1}, \\
\Bhs'_{d-1} & =  \Bhs_{d-1}.
\end{array}
$$

Again, $\Bls'_{d-1}  = \Bls_{d-1}$ and $ \Bhs'_{d-1} = \Bhs_{d-1}$ imply that $\sert{t}_{d} = \ser{t}_{d}$ and $\sert{s}_{d-1} = \sigs_{d-1}$. Denote $m' = \frac{m}{2^{d}}$, and observe that $m'$ is an odd nonnegative integer. Thus,
$\modTwo{m' \times (2 \times \beta + \A)} \stackrel{\text{Obs.\ref{obs:modTwo}}}{=} \modTwo{\A} = \alpha_1$, and $\modTwo{m' \times \alpha} \stackrel{\text{Obs.\ref{obs:modTwo}}}{=} \alpha$.
Since $\alpha \notin \B_{d} - m' \times (2 \times \beta + \A) \subseteq \B$ (by definition of elementary folding), we get that $\modTwo{\B_{d} - m' \times (2 \times \beta + \A)} \cap \modTwo{m' \times \alpha} = \emptyset$.
Consequentially,

$$
\begin{array}{rl}
\Bl'_{d} = & \modTwo{\B'_{d}}
	= \modTwo{\B_{d} - m' \times (2 \times \beta + \A) + m' \times \alpha} \\
	\stackrel{\text{Obs.\ref{obs:modTwo}}}{=}  & \Bl_{d} + \alpha_1 - 2 \times (\Bl_{d} \cap \alpha_1) + \alpha.
\end{array}
$$

Note that $|\Bl'_d| = |\Bl_d| + 2 - 2 |\Bl_{d} \cap \alpha_1|$, and therefore $s'_{d} = |\Bl'_{d}| - |\Bl'_{d-1}| - \frac{|\Bh'_{d-1}|}{2} + \max(s'_{d-1}, 0) = |\Bl_{d}| + 2 - 2 |\Bl_{d} \cap \alpha_1| - |\Bl_{d-1}| - \frac{|\Bh_{d-1}|}{2} + \max(s_{d-1}, 0) = s_d + 2 - 2 |\Bl_{d} \cap \alpha_1|$.
When $\Bl_{d} \cap \alpha_1 = \emptyset$, $s'_{d} = s_d + 2$, and so $\sigs < \sert{s}$ and the claim follows. Else, $\Bl_{d} \cap \alpha_1 = \alpha_1$, $\Bl'_d = \Bl'_d -\alpha_1 + \alpha$, therefore $s'_{d} = s_d$,
and there is a need to continue and examine positions grater than $d$ in the signatures of $\B'$ and $\B$.

 In this remaining case $\sert{s}_d = \sigs_d$, thus requirement (a) in Claim~\ref{clm:elementaryAux} holds with respect to $\B, \B'$ and $i = d$. In addition, $\alpha_1 \in \Bl_{d}$, implying that $t_{d+1} \leq \mint(\Bl_{d}) \leq \tp[\alpha_1] \leq \tp[\beta]$, and so $q = d$. Moreover, $t'_{d+1} \leq \mint(\Bl'_d) \leq \tp = \tp[\beta]$, and $t'_{d+1} = \min(\mint(\Bl'_d), t'_d) = \min(\mint(\Bl_d - \alpha_1 + \alpha), t_d) \geq \min(\mint(\Bl_d), t_d) = t_{d+1}$, hence requirement (b) in Claim~\ref{clm:elementaryAux} holds with respect to $\B, \B'$ and $i = d$. Now,

$$
\begin{array}{rl}
 \Bh'_d = & \geqt[(\B'_d - \Bl'_d)]{t'_{d+1}} \\
 = & ((\B_{d} - m' \times (2 \times \beta + \A) + m' \times \alpha) - (\Bl_{d} - \alpha_1 + \alpha))^{\geq t'_{d+1}} \\
 = & \geqt[(\B_d - \Bl_d)]{t'_{d+1}} - 2m' \times \beta + (m'-1) \times \alpha
  - m' \times \geqt[\A]{t'_{d+1}} + \geqt[\alpha_1]{t'_{d+1}} \\
 = & \left(\geqt[(\B_d - \Bl_d)]{t_{d+1}} - \rng[(\B_d - \Bl_d)]{t_{d+1}}{t'_{d+1}}\right) - 2m' \times \beta + (m'-1) \times \alpha \\
 & - m' \times \geqt[\A]{t'_{d+1}} + \geqt[\alpha_1]{t'_{d+1}} \\
 = & \Bh_d - \rng[(\B_d - \Bl_d)]{t_{d+1}}{t'_{d+1}} - 2m' \times \beta + (m'-1) \times \alpha - m' \times \geqt[\A]{t'_{d+1}} + \geqt[\alpha_1]{t'_{d+1}}.
\end{array}
$$

Observe that each element in $\B_d - \Bl_d$ has an even count, $\B_d$ contains at least $m'$ copies of $\alpha_1$ (since $m' \times \A \subseteq \B_d$), and $\Bl_d$ contains exactly one copy of $\alpha_1$, therefore $\B_d - \Bl_d - (m' - 1) \times \alpha_1$ is a valid collection in which each element has an even count (recall that $m'$ is odd).
Denote $C = \frac{1}{2}\left(\rng[(\B_d - \Bl_d - (m' - 1) \times \alpha_1)]{t_{d+1}}{t'_{d+1}}\right) = \frac{1}{2} \left(\rng[(\B_d - \Bl_d)]{t_{d+1}}{t'_{d+1}} - (m' - 1) \times \alpha_1^{<t'_{d+1}}\right)$. Next, we can write

$$
\begin{array}{rl}
 \Bh'_d   = & \Bh_d - \rng[(\B_d - \Bl_d)]{t_{d+1}}{t'_{d+1}} - 2m' \times \beta + (m'-1) \times \alpha - m' \times \geqt[\A]{t'_{d+1}} + \geqt[\alpha_1]{t'_{d+1}} \\
	= & \Bh_d - 2 \times C - (m' - 1) \times \alpha_1^{<t'_{d+1}} - 2m' \times \beta + (m' - 1) \times \alpha - m' \times \geqt[\A]{t'_{d+1}} + \geqt[\alpha_1]{t'_{d+1}} \\
= & \Bh_d - 2 \times C - (m' - 1) \times \left(\alpha_1 - \alpha_1^{\geq t'_{d+1}}\right) - 2m' \times \beta + (m' - 1) \times \alpha - m' \times \geqt[\A]{t'_{d+1}} + \geqt[\alpha_1]{t'_{d+1}}\\
= & \Bh_d - 2 \times C - (m' - 1) \times \alpha_1 - 2m' \times \beta + (m' - 1) \times \alpha - m' \times (\A - \alpha_1)^{\geq t'_{d+1}}.
\end{array}
$$

Since $m' \geq 1$ (being an odd nonnegative integer), we get that $|\Bh'_d| = |\Bh_d| - 2 |C| - 2m' - m' \left|(\A - \alpha_1)^{\geq t'_{d+1}}\right| < |\Bh_d| - 2 |C|$.
Therefore, $|\Bl'_d| + \frac{|\Bh'_d|}{2} < |\Bl_d| + \frac{|\Bh_d|}{2} + |C|$, and condition (c) of Claim~\ref{clm:elementaryAux} holds with respect to $\B, \B', C$, and $i = d$.
In addition, $\B_{d+1} \stackrel{\text{Obs.\ref{obs:t}}}{=} \frac{1}{2^{d+1}} \times \B^{< t_{d+1}}$, and in particular the top values of all elements in $\B_{d+1}$ are lower than $t_{d+1}$.
Since the top values of all elements in $C$ are at least $t_{d+1}$, we have that $\B_{d+1} \cap C = \emptyset$, and condition (d) of Claim~\ref{clm:elementaryAux} holds with respect to $\B, \B', C$, and $i = d$.
From Claim~\ref{clm:subConvex}, there is an integer $x > 0$ and a convexed $l$-collection $\hat{\A}$ such that $|\hat{A}| < |\A|$ and $(\A - \alpha_1)^{< t'_{d+1}} =  2^x \times \hat{\A}$, and so

$$
\begin{array}{rl}
 \Bh'_d  = & \Bh_d - 2 \times C - (m' - 1) \times \alpha_1 - 2m' \times \beta + (m'-1) \times \alpha - m' \times (\A - \alpha_1)^{\geq t'_{d+1}}\\
 = & \Bh_d - 2 \times C - (m' - 1) \times \alpha_1 - 2m' \times \beta + (m'-1) \times \alpha - m' \times \left((\A - \alpha_1) - (\A - \alpha_1)^{< t'_{d+1}}\right)\\
= & \Bh_d - 2 \times C - 2m' \times \beta + (m'-1) \times \alpha - (m' \times \A - \alpha_1) + m'2^x \times \hat{\A},
\end{array}
$$

\noindent
and

$$
\begin{array}{rl}
 \B'_{d+1} = & \frac{1}{2}(\B'_d - \Bl'_d - \Bh'_d) \\
 = & \frac{1}{2}((\B_{d} - m' \times (2 \times \beta + \A) + m' \times \alpha) - (\Bl_d - \alpha_1 + \alpha) \\
 & - (\Bh_d - 2 \times C - 2m' \times \beta + (m'-1) \times \alpha - (m' \times \A - \alpha_1) + m'2^x \times \hat{\A})) \\
 = & \frac{1}{2}(\B_{d} - \Bl_d - \Bh_d)
  + C - m'2^{x-1} \times \hat{\A} \\
 = & \B_{d+1} + C - m'2^{x-1} \times \hat{\A} \\
\end{array}
$$

Since $x > 0$, $m'2^{x-1}$ is an integer. Therefore, requirement (e) in Claim~\ref{clm:elementaryAux} holds with respect to $\B, \B', C, \hat{\A}, i = d$, and $m'' = m'2^{x-1}$, and the claim follows.

\end{proof}

\begin{claim}
\label{clm:addEpsilons}
Let $\B' = \B + m \times \varepsilon$ be a type II elementary folding. For $d = \dig{m}$, we have
\begin{enumerate}
	\item $\sert{s}_{d-1} = \sigs_{d-1}$,
	\item $\sig'_{d} = \sig_{d} + 1$,
	\item $r' = d+1$.
\end{enumerate}
\end{claim}

\begin{proof}
Since the folding is elementary, $\varepsilon \notin \B$, and for every $i \geq 0$ we have $t_i > 0$. Therefore, $\lt[\varepsilon]{t_i} = \varepsilon$ and $\rng[\varepsilon]{t_i+1}{t_i} = \emptyset$. From Claim~\ref{clm:addition}, we get that

$$
\begin{array}{rl}
	\B'_{d} = & \B_{d} + \frac{m}{2^{d}} \times \varepsilon, \\
	\Bls'_{d-1} = & \Bls_{d-1}, \\
	\Bhs'_{d-1} = & \Bhs_{d-1}.
\end{array}
$$.

From 	$\Bls'_{d-1} = \Bls_{d-1}$ and $\Bhs'_{d-1} =\Bhs_{d-1}$, it follows that $\sert{s}_{d-1} = \sigs_{d-1}$.
Since $\frac{m}{2^{d}}$ is an odd integer (by definition) and $\varepsilon \notin \Bl_d \subseteq \B$, $\Bl'_d = \modTwo{\B'_d} = \modTwo{\B_d + \frac{m}{2^{d}} \times \varepsilon} \stackrel{\text{Obs.\ref{obs:modTwo}}}{=} \modTwo{\B_d + \varepsilon} \stackrel{\text{Obs.\ref{obs:modTwo}}}{=} \Bl_d + \varepsilon - 2 \times (\Bl_d \cap \varepsilon) = \Bl_d + \varepsilon$. If $d = 0$ then $s'_d = |\Bl'_d| = |\Bl_d| + 1 = s_d + 1$, and otherwise $s'_{d} = |\Bl'_d| - |\Bl'_{d-1}| - \frac{|\Bh'_{d-1}|}{2} + \max(s'_{d-1}, 0) = |\Bl_d| + 1 - |\Bl_{d-1}| - \frac{|\Bh_{d-1}|}{2} + \max(s_{d-1}, 0) = s_d + 1$. Finally, as $\varepsilon \in \Bl'_d$ (and in particular $r' > d$), it follows that $t'_{d+1} = \mint(\Bl'_t) = \tp[\varepsilon] = 0$, $\B'_{d+1} \stackrel{\text{Obs.\ref{obs:t}}}{=} \frac{1}{2^{d+1}} \times \lt[\B]{0} = \emptyset$, and so $r' = d+1$.
\end{proof}

Finally, we prove the main claim in this section.

\begin{proof}[Proof of Claim~\ref{clm:foldIncreaseSignature}]
The correctness of the claim follows immediately from Claims~\ref{clm:elementarySeq},~\ref{clm:typeI}, and~\ref{clm:addEpsilons}.
\end{proof}

\subsection{The FOLD Procedure}
Using the notation and definitions given in the previous sections, we now give an explicit description of the FOLD procedure.

\begin{algorithm}[H]
\normalsize
	\DontPrintSemicolon
	\footnotesize
	\SetKwInput{Algorithm}{Algorithm}
	\SetKwInput{Procedure}{Procedure}
	\SetKwInput{Input}{Input}
	\SetKwInput{Output}{Output}
	\SetKwIF{If}{ElseIf}{Else}{If}{then}{Else if}{Else}{endif}
	\SetKwFor{For}{For}{do}{endfor}
	\SetKwFor{While}{While}{do}{endw}

	\Procedure{FOLD$(\B, n)$}
	\BlankLine
	\Input{An $l$-BFB palindrome collection $\B$ and an integer $n \geq 0$.}
	\Output{A minimum signature folding $\B'$ of $\B$ such that $|\B'| = n$, or the string ``FAILD'' if there is no such $\B'$.}
	\BlankLine
	\lIf{$|\B| \leq  n$}{\Return{$\B + (n - |\B|) \varepsilon$}.\;}
	\Else{
	Let $\sigs = \sigs(\B)$ and $\ser{\Delta} = \ser{\Delta}(\B)$.\;
	\If{there exists $0 \leq d \leq \dig{|\B| - n}$ such that $n \geq 2^d\max(s_d+1, 0) + \Delta_d$}{
		Let $d$ be the maximum integer sustaining the condition above.
		Set $\B' \leftarrow \B + 2^d \times \varepsilon$.\;
		\While{$|\B'| > n$}{
			Set $\B' \leftarrow \text{RIGHT-FOLD}(\B')$.\;
		}
		Set $\B' \leftarrow \B' + (n-|B'|) \times \varepsilon$.\;
		\textbf{Return} $\B'$\;
	}
	\lElse{\Return{``FAILED''}\;}
	}
	\label{alg:FOLD}
\end{algorithm}

\begin{algorithm}[H]
\normalsize
	\DontPrintSemicolon
	\footnotesize
	\SetKwInput{Algorithm}{Algorithm}
	\SetKwInput{Procedure}{Procedure}
	\SetKwInput{Input}{Input}
	\SetKwInput{Output}{Output}
	\SetKwInput{Precondition}{Precondition}
	\SetKwIF{If}{ElseIf}{Else}{If}{then}{Else if}{Else}{endif}
	\SetKwFor{For}{For}{do}{endfor}
	\SetKwFor{While}{While}{do}{endw}

	\Procedure{RIGHT-FOLD$(\B)$}
	\BlankLine
	\Input{An $l$-BFB palindrome collection $\B$.}
	\Precondition{Let $\left\langle \Bs, \Bls, \Bhs\right\rangle$ be the decomposition of $\B$, and $r = r(\B)$. There is an integer $0 \leq g < r$ such that $\Bh_g \neq \emptyset$, $\Bl_g \neq \emptyset$, and for every $g < i < r$, $\Bh_i = \emptyset$ and $\Bl_i \neq \emptyset$.}
	\Output{A folding $\B'$ of $\B$ of size $|\B| - 2^r$.}
	\BlankLine
	Let $g$ be an integer as implied from the precondition (note that $g$ is unique), $\beta$ a minimal element in $\Bh_g$, $\A = \{\alpha_1, 2\times \alpha_2, \ldots, 2^{r-g-1} \times \alpha_{r-g}\}$ a convexed $l$-collection such that $\alpha_i \in  \Bl_{g+i-1}$ for each $1 \leq i \leq r-g$ and $\alpha_1$ is a minimal element in $\Bl_g$, and $\alpha = \beta \gamma_{\A} \beta$.\;
	\textbf{Return} the collection $\B' = \B - 2^g \times (2 \times \beta + \A) + 2^g \times \alpha$.\;
	\label{alg:rightFold}
\end{algorithm}

In general, it is easy to assert that when the precondition holds, the returned collection $\B'$ from the RIGHT-FOLD procedure is a folding of the input collection $\B$, where each one of the $2^g$ copies of $\alpha$ in $\B'$ is obtained by concatenating all elements in $\A$ and two copies of $\beta$.
Since a right-folding adds to the collection $2^g$ copies of $\alpha$ while reducing $2^g$ repeats of the collection $2 \times \beta + \A$ of size $2 + 2^{r-g} - 1 = 1 + 2^{r-g}$, the size of the folded collection $\B'$ has decreased by $2^r$ with respect to the size of the original collection $\B$.

\subsubsection{Right-folding Properties}

In this section we show certain characteristics of right-foldings.

\begin{claim}
\label{clm:rightFoldingAllowed}
There is a right folding of a collection $\B$ if and only if $s_r < 0$.
\end{claim}
\begin{proof}
For the first direction of the proof, assume that there is a right folding $\B'$ of $\B$. From Claim~\ref{clm:sL}, $|\Bl_g| \geq \max(s_g, 0)$. Since $\Bh_g \neq \emptyset$, we get that $- |\Bl_{g}| - \frac{|\Bh_g|}{2} + \max(s_g, 0) < 0$. This, in turn, implies that $s_{g+1} = |\Bl_{g+1}| - |\Bl_{g}| - \frac{|\Bh_g|}{2} + \max(s_g, 0) < |\Bl_{g+1}|$. If $r = g+1$, then $s_r = s_{g+1} < |\Bl_{g+1}| = 0$, and the claim follows. Otherwise, $\Bl_{g+1} \neq \emptyset$, and in particular $- |\Bl_{g+1}| - \frac{|\Bh_{g+1}|}{2} + \max(s_{g+1}, 0) < 0$. Inductively, this shows that $s_r < 0$.

For the second direction, assume that $s_r < 0$. Assume that for some $i < r$ we have that $- |\Bl_{i}| - \frac{|\Bh_{i}|}{2} + \max(s_{i}, 0) < 0$, and that $\Bh_d = \emptyset$ and $\Bl_d \neq \emptyset$ for all $i < d < r$. Note that this requirement holds for $i = r-1$, since $- |\Bl_{r-1}| - \frac{|\Bh_{r-1}|}{2} + \max(s_{r-1},0) = |\Bl_r| - |\Bl_{r-1}| - \frac{|\Bh_{r-1}|}{2} + \max(s_{r-1}, 0) = s_r < 0$, and there are no integers $d$ such that $r-1 < d < r$.
If $\Bh_{i} \neq \emptyset$, then also $\Bl_{i} \neq \emptyset$ (since $\Bl_{i} = \emptyset$ implies that $\mint(\Bl_i) = \infty$, and by definition $\Bh_{i} = \geqt[(\B_i - \Bl_i)]{\infty} = \emptyset$), and therefore the requirements for the existence of a right-folding hold for $g = i$. Else, $\Bh_{i} = \emptyset$, and so $- |\Bl_{i}| + \max(s_{i}, 0) < 0$. This implies that $\Bl_i \neq \emptyset$ and that $i \neq 0$ (as $|\Bl_{0}| = s_{0}$), and so $|\Bl_i| > \max(s_i, 0) \geq s_i = |\Bl_i| - |\Bl_{i-1}| - \frac{\Bh_{i-1}}{2} + \max(s_{i-1}, 0)$, and we get that $-|\Bl_{i-1}| - \frac{\Bh_{i-1}}{2} + \max(s_{i-1}, 0) < 0$. Inductively, for some $i' < r$, it must hold that $\Bh_{i'} \neq \emptyset$, $\Bh_d = \emptyset$ for every $i' < d < r$, and $\Bl_d \neq \emptyset$ for every $i' \leq d < r$, meeting the requirements for the existence of a right folding for $g = i'$.
\end{proof}

Throughout the remaining of this section, assume that $\B$ is a collection satisfying the pre-condition in Procedure RIGHT-FOLD, and let $\B' = \B - 2^g \times(2 \times \beta + \A) + 2^g \times \alpha$ be the output of the procedure (where $g$, $r$, $\beta$, $\A = \{\alpha_1, \ldots, 2^{r-g-1} \times \alpha_{r-g}\}$, and $\alpha$ are as defined in the procedure). Note that when $\alpha \notin \B$, $\B'$ is also a type I elementary folding of $\B$. We later show that all right-foldings preformed in line~7 of the FOLD procedure are elementary.

\begin{claim}
\label{clm:rightFolding}
If $\B'$ is an elementary folding of $\B$, then
\begin{enumerate}
	\item $\Bls'_{g-1} = \Bls_{g-1}$, and $\Bl'_g = \Bl_g - \alpha_1 + \alpha$.
	\item $\Bhs'_{g-1} = \Bhs_{g-1}$, and $\Bh'_g = \Bh_g - 2 \times \beta$.
	\item For every $g < i < r$, $\Bl'_{i} = \Bl_{i} - \alpha_{i-g}$, and $\Bh'_{i} = \Bh_i = \emptyset$.
	\item $\B'_{r} = \B_{r} = \emptyset$ (thus $r' \leq r$).
	\item $|\B'| = |\B| - 2^{r}$.
\end{enumerate}
\end{claim}

\begin{proof}
We start by showing the first two items in the claim.
Since $\beta \in \Bh_g \stackrel{\text{Obs.\ref{obs:t}}}{\subseteq} \B^{< t_g}$, it follows that for every $\alpha_j \in \A$ and every $0 \leq i \leq g$, $\tp[\alpha_j] \leq \tp = \tp[\beta] < t_g \leq t_i$. Therefore, from Claim~\ref{clm:elementaryDecomposition} and the fact that $\dig{2^g} = g$, we get that $\Bls'_{g-1} = \Bls_{g-1}$ (and in particular $\sert{t}_{g} = \ser{t}_{g}$), $\Bhs'_{g-1} = \Bhs_{g-1}$, and $\B'_g = \B_g - (2 \times \beta + \A) + \alpha$.
Since $\modTwo{2 \times \beta + \A} = \alpha_1 \in \Bl_g$, we have $\modTwo{\B_g - (2 \times \beta + \A)} \stackrel{\text{Obs.\ref{obs:modTwo}}}{=} \modTwo{\B_g + \alpha_1} \stackrel{\text{Obs.\ref{obs:modTwo}}}{=} \modTwo{\B_g} + \alpha_1 - 2 \times (\modTwo{\B_g} \cap \alpha_1) = \Bl_g + \alpha_1 - 2 \times (\Bl_g \cap \alpha_1) = \Bl_g - \alpha_1$. As $\alpha \notin \Bl_g + \alpha_1$ (follows from $\alpha \notin \B$), we get that $\Bl'_g = \modTwo{\B'_g} = \modTwo{\B_g - (2 \times \beta + \A) + \alpha} \stackrel{\text{Obs.\ref{obs:modTwo}}}{=} \modTwo{(\Bl_g - \alpha_1) + \alpha} = \Bl_g - \alpha_1 + \alpha$.
By definition, $\alpha_1$ is a minimal element in $\Bl_g$, therefore $\tp[\alpha_1] = \mint(\Bl_g)$, and so $\tp[\alpha_1] \leq \mint(\Bl_g - \alpha_1)$. In addition, $\tp[\alpha_1] \leq \tp$, and therefore $\tp[\alpha_1] \leq \min(\mint(\Bl_g - \alpha_1), \tp) = \mint(\Bl_g - \alpha_1 + \alpha) = \mint(\Bl'_g)$. On the other hand, $\tp[\beta] = \tp$, and therefore $\mint(\Bl'_g) = \mint(\Bl_g - \alpha_1 + \alpha) \leq \tp[\beta]$.
From the minimality of $\beta$ in $\Bh_g$, $\Bh_g = \Bh_g^{\geq \tp[\beta]} = ((\B_g - \Bl_g)^{\geq \mint(\Bl_g)})^{\geq \tp[\beta]} = ((\B_g - \Bl_g)^{\geq \tp[\alpha_1]})^{\geq \tp[\beta]} = (\B_g - \Bl_g)^{\geq \max(\tp[\alpha_1], \tp[\beta])} = (\B_g - \Bl_g)^{\geq \tp[\beta]}$. Since $\tp[\alpha_1] \leq \mint(\Bl'_g) \leq \tp[\beta]$, we get that $\Bh_g = (\B_g - \Bl_g)^{\geq \tp[\beta]} \subseteq (\B_g - \Bl_g)^{\geq \mint(\Bl'_g)} \subseteq (\B_g - \Bl_g)^{\geq \tp[\alpha_1]} = \Bh_g$, and so $(\B_g - \Bl_g)^{\geq \mint(\Bl'_g)} = \Bh_g$. In addition, $\beta^{\geq \mint(\Bl'_g)} = \beta$, and $(\A  - \alpha_1)^{\geq \mint(\Bl'_g)} = \emptyset$ (since for all $i>0$, $\tp[\alpha_i] < \tp[\alpha_1] \leq \mint(\Bl'_g)$). Thus, we get that $\Bh'_g = (\B'_g - \Bl'_g)^{\geq \mint(\Bl'_g)} = ((\B_g - 2 \times \beta - \A + \alpha) - (\Bl_g - \alpha_1 + \alpha))^{\geq \mint(\Bl'_g)}
=
(\B_g - \Bl_g)^{\geq \mint(\Bl'_g)} - 2 \times \beta^{\geq \mint(\Bl'_g)} - (\A  - \alpha_1)^{\geq \mint(\Bl'_g)} = \Bh_g - 2 \times \beta$, as required.

Next, we turn to show item 3 in the claim, which is relevant only for the case where $g < r-1$. We prove this item inductively for all $g < i < r$. Note that $\B'_{g+1} = \frac{1}{2} \times (\B'_g - \Bl'_g - \Bh'_g) = \frac{1}{2}((\B_g - 2 \times \beta - \A + \alpha) - (\Bl_g - \alpha_1 + \alpha) - (\Bh_g - 2 \times \beta)) = \frac{1}{2} \times ((\B_g - \Bl_g - \Bh_g) - (\A - \alpha_1)) = \B_{g+1} - \frac{1}{2} \times \A$. Now, assume that for some $g < i < r$, $\B'_i = \B_{i} - \frac{1}{2^{i-g}} \times \A$. Note that $\alpha_{i-g} \in \Bl_i$, and in particular $\Bl_{i} \cap \alpha_{i-g} = \alpha_{i-g}$.
Therefore, $\Bl'_{i} = \modTwo{\B'_{i}} = \modTwo{\B_{i} - \frac{1}{2^{i-g}} \times \A} \stackrel{\text{Obs.\ref{obs:modTwo}}}{=} \modTwo{\B_{i}} + \modTwo{\frac{1}{2^{i-g}} \times \A} - 2 \times (\modTwo{\B_{i}} \cap \modTwo{\frac{1}{2^{i-g}} \times \A}) = \Bl_{i} + \alpha_{i-g} - 2 \times (\Bl_{i} \cap \alpha_{i-g}) = \Bl_{i} + \alpha_{i-g} - 2 \times  \alpha_{i-g} = \Bl_{i} - \alpha_{i-g}$, as required. In addition, since $\mint(\Bl'_i) = \mint(\Bl_i - \alpha_{i-g}) \geq \mint(\Bl_i) = \tp[\alpha_{i-g}]$, we get that  $\Bh'_i = (\B'_i - \Bl'_i)^{\geq \mint(\Bl'_i)} = ((\B_{i} - \frac{1}{2^{i-g}} \times \A) - (\Bl_i - \alpha_{i-g}))^{\geq \mint(\Bl'_i)} \subseteq
(\B_{i} - \Bl_i - (\frac{1}{2^{i-g}} \times \A - \alpha_{i-g}))^{\geq \mint(\Bl_i)} \subseteq \Bh_i = \emptyset$, and so $\Bh'_i = \emptyset$, as required. Finally, it follows that $\B'_{i+1} = \frac{1}{2} \times (\B'_i - \Bl'_i - \Bh'_i) = \frac{1}{2} \times((\B_{i} - \frac{1}{2^{i-g}} \times \A) - (\Bl_{i} - \alpha_{i-g}) - \Bh_i) = \B_{i+1} - \frac{1}{2} \times (\frac{1}{2^{i-g}} \times \A - \alpha_{i-g}) = \B_{i+1} - \frac{1}{2^{i+1-g}} \times \A$, as required for the next inductive step.

Item~4 in the claim follows from the fact that $\B'_r = \B_{r} - \frac{1}{2^{r-g}} \times \A = \emptyset - \emptyset = \emptyset$, and item~5 is obtained from the fact that $|\B'| = |\B| - 2^g(2  + |\A|) + 2^g = |\B| - 2^g(2 + 2^{r-g} - 1) + 2^g = |\B| - 2^{r}$.
\end{proof}

\begin{claim}
\label{clm:rightFoldingSig}
If $\B'$ is an elementary folding of $\B$, then
$\sert{s}_{r-1} = \sigs_{r-1}$ and $\sig'_r = \sig_r + 1$. In addition, $r' \leq r$, where if $\sig'_r < 0$ then $r' = r$.
\end{claim}

\begin{proof}
By definition~\ref{def:signature}, the values in the series $\sigs'$ depend only on sizes of collections in $\Bls'$ and $\Bhs'$. These sub-collection sizes may be inferred from Claim~\ref{clm:rightFolding}, and their assignments in definition~\ref{def:signature} imply the correctness of the claim in a straightforward manner.
\end{proof}

\begin{sloppypar}
Let $\beta$ be a palindrome obtained by concatenating zero or more $l$-blocks. If $\beta$ is obtained by concatenating an odd number of blocks, $\beta$ is of the form $\beta = \beta_1 \beta_2 \ldots \beta_{q-1} \beta_q \beta_{q-1} \ldots \beta_2 \beta_1$ (where each $\beta_i$ is an $l$-block), whereas if $\beta$ is obtained by concatenating an even number of blocks it is of the form $\beta = \beta_1 \beta_2 \ldots \beta_{q-1} \beta_q \varepsilon \beta_q \beta_{q-1} \ldots \beta_2 \beta_1$. Call $\beta_q$ or $\varepsilon$ respectively the \emph{center} of $\beta$, in these two cases. Note that a center of an $l$-block $\beta$ is $\beta$.
\end{sloppypar}

\begin{definition}
\label{def:uniqueCenters}
Say that an $l$-BFB palindrome collection $\B$ has \emph{unique centers} if all elements in collections of the form $\Bh_d$ are $l$-blocks, and for every $\beta \in \Bl_d$ and $\beta' \in \Bl'_d$ (for some possibly equal integers $d$ and $d'$) such that $\beta \neq \beta'$, the centers of $\beta$ and $\beta'$ differ.
\end{definition}

\begin{claim}
\label{clm:rightFoldUniqueCenters}
If $\B$ has unique centers then $\B'$ is an elementary folding of $\B$, and $\B'$ has unique centers.
\end{claim}

\begin{proof}
To prove the folding is elementary, we need to show that $\alpha \notin \B$.
Note that $\beta \in \Bh_g$, $\alpha_1 \in \Bl_g$, $\tp = \tp[\beta]$, and the center of $\alpha$ is the the center of $\alpha_1$.
 Assume by contradiction that $\alpha \in \B$. Since $\tp = \tp[\beta]$, Observation~\ref{obs:t} implies that $\alpha \in \Bl_g + \Bh_g$. Since $\alpha$ is not an $l$-block, $\alpha \notin \Bh_g$. Since $\alpha_1$ and $\alpha$ have the same center, and since $\alpha_1 \in \Bl_g$ and $\B$ has unique centers, it follows that $\alpha \notin \Bl_g$, leading to a contradiction.

 The fact that $\B'$ has unique centers follows from the contents of collections in the series $\Bls'$ and $\Bhs'$, as given in Claim~\ref{clm:rightFolding}.
\end{proof}

\begin{claim}
\label{clm:rightFoldingSeq}
Let $\B$ be an $l$-BFB palindrome collection with unique centers.
Then, for $i = -s_r - \min(s_{r-1}, 0)$, it is possible to produce a series of collections $\B^0 = \B, \B^1, \B^2, \ldots, \B^i$, each collection $\B^j$ obtained by applying an elementary right-foldings over the preceding collection $\B^{j-1}$. In addition, for every $0 \leq j \leq -s_r$,
\begin{itemize}
	\item $|\B^j| = |\B| - 2^{r}j$,
	\item $\sera{s}{j}_{r-1} = \sigs_{r-1}$,
	\item $\sig^j_{r} = s_r + j$,
\end{itemize}
\noindent
and for each $-s_r \leq j \leq -s_r - \min(\sig_{r-1}, 0)$,

\begin{itemize}
	\item $r^j \leq r-1$,
	\item $|\B^j| = |\B| + 2^{r-1}(s_r - j)$,
	\item $\sera{s}{j}_{r-2} = \sigs_{r-2}$,
	\item $\sig^j_{r-1} = \sig_{r-1} + j + s_r$.
\end{itemize}
\end{claim}

\begin{proof}

From Claim~\ref{clm:r}, $s_r \leq 0$. Note that $s_r = -|\Bl_{r-1}| - \frac{|\Bh_{r-1}|}{2} + \max(s_{r-1}, 0) \leq \max(s_{r-1}, 0)$.
When $s_r = 0$, $\max(s_{r-1}, 0) \geq 0$, and so $\min(s_{r-1}, 0) = 0$ and in particular $-s_r - \min(s_{r-1}, 0) = 0$ and the claim holds trivially.
Otherwise, $s_r < 0$, thus $-s_r - \min(s_{r-1}, 0) > 0$, and we continue to assert the correctness of the claim.

In the remaining case, $s_r < 0$, and from Claims~\ref{clm:rightFoldingAllowed},~\ref{clm:rightFoldUniqueCenters}, and~\ref{clm:rightFoldingSig} there is an elementary right-folding $\B^1$ of $\B^0 = \B$ with unique centers, such that
$\sera{s}{1}_{r-1} = \sigs_{r-1}$ and $\sig^1_{r} = \sig_r + 1$. Note that $r^1 \leq r$, where $s^1_r < 0$ implies that $r_1 = r$.
Similarly, it is possible to apply a series of a total amount of $x = -s_r$ right-foldings $\B = \B^0, \B^1, \ldots, \B^x$, where for every $j < x$ we have that $r_j = r$ and $r_x \leq r$, and for every $j \leq x$ we have that $\sera{s}{j}_{r-1} = \sigs_{r-1}$ and $\sig^j_{r} = \sig_r + j$. Since each such a right-folding decreases the size of the collection by $2^r$ elements, $|\B^j| = |\B| - 2^r j$, hence the first part of the claim.

After performing $x = -s_r$ right-foldings, we get the collection $\B^x$ for which $r_x \leq r$, $\sera{s}{j}_{r-1} = \sigs_{r-1}$, $\sig^j_{r} = \sig_r + x = 0$, and $|\B^x| = |\B| - 2^{r+1}x = |\B| + 2^{r+1}s_r$. If $- \min(\sig_{r-1}, 0) = 0$, then the second part of the claim follows immediately. Else, $- \min(\sig_{r-1}, 0) = - \min(\sig^x_{r-1}, 0) > 0$, thus $\sig^x_{r-1} = \sig_{r-1} < 0$, and $\max(s^x_{r-1}, 0) = 0$. Since $0 = s^x_r = -|\Bl^x_{r-1}| - \frac{|\Bh_{r^x-1}|}{2} + \max(s^x_{r-1}, 0) = -|\Bl^x_{r-1}| - \frac{|\Bh_{r^x-1}|}{2}$, it follows that $\Bl^x_{r-1} = \Bh^x_{r-1} = \emptyset$, and therefore $r^x \leq r-1$. On the other hand, from Claim~\ref{clm:r} and the fact that $s^x_{r-1} \neq 0$ we get that $r^x \geq r-1$, and thus $r^x = r-1$. As above, it is possible to apply additional consecutive $y = -s^x_{r-1} = -s_{r-1} = -\min(s_{r-1}, 0)$ right-foldings, where each such folding maintains the signature values at positions $0$ to $r-2$, increases by 1 the signature value at position $r-1$ with respect to the preceding collection in the series, and decreases the collection size by $2^{r-1}$. Hence, for $-s_r \leq j \leq -s_r - \min(\sig_{r-1}, 0)$, we have that $r^j \leq r-1$, $|\B^j| = |\B^x| - 2^{r-1}(j-x) = |\B| + 2^r s_r - 2^{r-1}(j + s_r) = |\B| + 2^{r-1}(s_r - j)$, $\sera{s}{j}_{r-2} = \sigs_{r-2}$, and $\sig^j_{r-1} = \sig^x_{r-1} + j - x = \sig_{r-1} + j + s_r$, as required.
\end{proof}

\subsubsection{Correctness of the FOLD Procedure}
\label{sec:foldProcedure}

Throughout this section, $\B$ and $n$ correspond to an $l$-block collection and an integer given as an input to the FOLD procedure. When $n \neq |\B|$, denote $d' = \dig{|\B| - n}$.

\begin{claim}
\label{clm:required}
If there is a folding $\B'$ of $\B$ of size $n \neq |\B|$, then there is an integer $0 \leq d \leq d'$ such that $\sert{s}_{d-1} = \sigs_{d-1}$, $s'_d \geq s_d+1$, and $n \geq 2^d\max(s_{d}+1, 0) + \Delta_d$. In addition, if $d < d'$, then $s'_d \geq s_d+2$.
\end{claim}

\begin{proof}
Assume there is a folding $\B'$ of $\B$ of size $n \neq |\B|$, and let $d \geq 0$ be the integer such that $\sigs_{d-1} = \sert{s}_{d-1}$ and $s'_d \geq s_d+1$ (whose existence is implied by Claim~\ref{clm:foldIncreaseSignature}). Since $\sigs_{d-1} = \sert{s}_{d-1}$, it follows that $\Delta_d = \Delta'_d$.
Thus, $n = |\B'| \stackrel{\text{Clm.\ref{clm:main}}}{=} 2^d (|\B'_d| + |\Bl'_d| -s'_d) + \Delta'_d \stackrel{\text{Clm.\ref{clm:main}}}{\geq} 2^d \max(s'_d, 0) + \Delta'_d \geq 2^d \max(s_d+1, 0) + \Delta_d$. In addition, $|\B| = 2^d\left(|\B_{d}| + |\Bl_d| - s_d\right) + \Delta_d$, therefore $|\B| -  n = 2^d\left(|\B_{d}| + |\Bl_d| - s_d - |\B'_{d}| - |\Bl'_d| + s'_d\right)$. Since all parameters in the right-hand side of the latter equation are integers, $|\B| - n$ divides by $2^d$, and in particular $d \leq d'$. Furthermore, if $d < d'$, then $\frac{|\B| - n}{2^d}$ is even, and therefore  $|\B_{d}| + |\Bl_d| - s_d - |\B'_{d}| - |\Bl'_d| + s'_d$ is also even. As $|\B_{d}| + |\Bl_d|$ is even, as well as $|\B'_{d}| + |\Bl'_d|$, it follows that $s'_d - s_d$ has to be even. Since $s'_d > s_d$, it follows that $s'_d \geq s_d + 2$.
\end{proof}

\begin{claim}
\label{clm:FoldAux}
If there is a folding of $\B$ of size $n$ then FOLD$(\B, n)$ returns such a folding $\B'$, and otherwise FOLD$(\B, n)$ returns ``FAILED''. In addition, if $n \neq |\B|$ and FOLD$(\B, n)$ has returned $\B'$, then for the maximum integer $0 \leq d \leq d'$ for which $n \geq 2^d\max(s_{d}+1, 0) + \Delta_d$ (whose existence is guaranteed by Claim~\ref{clm:required}), $\sert{s}_{d-1} = \sigs_{d-1}$ and $r' \leq d+1$, and if $r' = d+1$ then $s'_d = s_d + 1$ in case $d = d'$ and $s'_d = s_d + 2$ in case $d < d'$.
\end{claim}

\begin{proof}
When there is no folding of $\B$ of size $n$, then in particular $n \neq |\B|$, and the procedure does not halt at line~1. In addition, from Claim~\ref{clm:required}, the condition in line~4 does not met, and the procedure returns ``FAILED'' in line~10 as required.

Else, there is a folding of $\B$ of size $n$, and we show that the procedure finds such a folding sustaining the stated requirements. When $|\B| <= n$, the FOLD procedure halts by returning $\B + (n - |\B|) \varepsilon$ in line~1, which is in particular a folding of $\B$ of size $n$ as required. In addition, if $|\B| < n$, we have from Claim~\ref{clm:addEpsilons} that $\sert{s}_{d-1} = \sigs_{d-1}$, $s'_d = s_d + 1$, and $r' = d+1$, thus the remaining requirements in the claim hold.
Otherwise, $n < |\B|$, and from Claim~\ref{clm:required} the condition in line~4 holds, therefore in line~5 of the FOLD procedure, the value of the parameter $d$ is selected to be the maximum integer in the range $0 \leq d \leq d'$ such that $n \geq 2^d\max(s_{d}+1, 0) + \Delta_d$.

Let $\B^0 = \B+2^d \times \varepsilon$ be the value of the collection $\B'$ after executing line~5. Thus $|\B^0| = |\B| + 2^d$, and from Claim~\ref{clm:addEpsilons}, we have that
\begin{enumerate}
	\item $\sera{s}{0}_{d-1} = \sigs_{d-1}$,
	\item $\sig^0_{d} = \sig_{d} + 1$,
	\item $r^0 = d+1$.
\end{enumerate}

From the proof of Claim~\ref{clm:addEpsilons} and the fact that $\B$ is an $l$-block collection it can be seen that $\B^0$ has unique centers.
From Conclusion~\ref{conc:r}, $\sig^0_{d+1} = \frac{\Delta^0_{d+1} - |\B^0|}{2^{d+1}} = \frac{\Delta_{d} + 2^d\abs(s_d+1) - |\B| - 2^d}{2^{d+1}}$. From Claim~\ref{clm:rightFoldingSeq}, the collection $\B^0$ can undergo a series of $i$ right-foldings producing the sequence $\B^0, \B^1, \ldots, \B^i$, where $i = -s^0_{d+1} - \min(s^0_{d}, 0)$.
 The size of the collection $\B^i$ according to Claim~\ref{clm:rightFoldingSeq} is $|\B^i| = |\B^0| + 2^d(s^0_{d+1} - i) = (|\B| + 2^d) + 2^d(2s^0_{d+1} + \min(s^0_{d}, 0)) = |\B| + 2^d + 2^d\left(\frac{\Delta_{d} + 2^d\abs(s_d+1) - |\B| - 2^d}{2^{d}} + \min(s_{d}+1, 0)\right) = \Delta_d + 2^d(\abs(s_d+1) + \min(s_d+1,0)) = \Delta_d + 2^d\max(s_d+1, 0)$. From the condition in line~4, $n \geq 2^d\max(s_d+1, 0) + \Delta_d = |\B^i|$, and in particular there exists some $0 \leq j \leq i$ such that $|\B^j| \leq n$. The sequence of right-foldings computed by the loop lines~6-7 is a prefix of such a right-folding sequence (i.e. after $x$ iterations of the loop, $\B' = \B^x$), where the loop terminates after $j$ iterations for a minimal integer $j$ such that $|\B^j| \leq n$.
 After executing line~8, $\B'$ is a folding of $\B$ of size $|\B'| = |\B^j| + (n - |\B^j|) = n$, and so the output $\B'$ of the procedure is a folding of $\B$ of size $n$, as required.

To complete the proof, we need to show that when $n < |\B|$, $\sert{s}_{d-1} = \sigs_{d-1}$ and $r' \leq d+1$, and if $r' = d+1$ then $s'_d = s_d + 1$ in case $d = d'$ and $s'_d = s_d + 2$ in case $d < d'$. To do so, we consider two cases for the number of loop iterations $j$ conducted by the procedure. Note that $j > 0$, since in the first iteration we have that $|\B^0| = |\B| + 2^d > n$.

\noindent
\textbf{1. $0 < j \leq -s^0_{d+1}$.}
In this case, Claim~\ref{clm:rightFoldingSeq} and the loop termination condition imply that $n \geq |\B^j| = |\B^0| - 2^{d+1}j = |\B| + 2^d - 2^{d+1}j = |\B| + 2^d(1 - 2j)$, and that $n < |\B^{j-1}| = |\B| + 2^d(1 - 2(j-1))$, therefore, $2j - 3 < \frac{|\B|-n}{2^d} \leq 2j - 1$. Note that when $d = d'$, $\frac{|\B|-n}{2^d}$ is odd, hence $\frac{|\B|-n}{2^d} = 2j - 1$, whereas when $d < d'$, $\frac{|\B|-n}{2^d}$ is even, and  $\frac{|\B|-n}{2^d} = 2j - 2$.

In the case where $d = d'$, $|\B^j| = |\B| + 2^d(1 - 2j) = |\B| - 2^d(\frac{|\B|-n}{2^d}) = n$, thus no $\varepsilon$ elements are added to the collection in line~8 of the procedure and the returned collection is $\B' = \B^j$. From Claim~\ref{clm:rightFoldingSeq}, $r' = r^0 = d+1$, and $\sert{s}_{d} = \sera{s}{0}_{d}$, i.e. $\sert{s}_{d-1} = \sera{s}{0}_{d-1} = \sigs_{d-1}$ and $s'_d = s^0_d = s_d+1$, and the claim follows.

In the case where $d < d'$, $|\B^j| = |\B| + 2^d(1 - 2j) = |\B| - 2^d(2j - 2 + 1) = |\B| - 2^d(\frac{|\B|-n}{2^d} + 1) = n - 2^d$, thus after line~8 of the procedure the returned collection is $\B' = \B^j + 2^d \times \varepsilon$. It may be asserted that $\varepsilon$ is the unique minimal element in $\B^0_d$ (as all other elements are $l$-blocks with higher top values), and thus this element participates in the right-folding that transforms $\B^0$ to $\B^1$. Therefore, for each $1 \leq j' \leq j$, $\varepsilon \notin \B^{j'}$, and in particular $\B'$ is a type II elementary folding of $\B^j$. From Claim~\ref{clm:addEpsilons}, $r' = d+1$, $\sert{s}_{d-1} = \sera{s}{0}_{d-1} = \sigs_{d-1}$, and $s'_d = s^0_{d}+1 = s_d + 2$, hence the claim follows.

\vspace{10pt}
\noindent
\textbf{2. $-s^0_{d+1} < j \leq -s^0_{d+1} - \min(s_d^0, 0)$.}
In this case, Claim~\ref{clm:rightFoldingSeq} and the loop termination condition imply that $n \geq |\B^j| = |\B^0| + 2^{d}(s^0_{d+1} - j) = (|\B| + 2^d) + 2^d(s^0_{d+1} - j) = |\B| + 2^d(s^0_{d+1} - j + 1)$, and $n < |\B^{j-1}| = |\B| + 2^d(s^0_{d+1} - j + 2)$. Therefore, $-s^0_{d+1} + j - 2 < \frac{|\B| - n}{2^d} \leq -s^0_{d+1} + j - 1$. Since $\frac{|\B| - n}{2^d}$ is an integer, it follows that $\frac{|\B| - n}{2^d} = -s^0_{d+1} + j - 1$, therefore $|\B^j| = n$, and consequentially after line~8 of the procedure the returned collection is $\B' = \B^j$.
From Claim~\ref{clm:rightFoldingSeq}, $r' \leq d$, and $\sert{s}_{d-1} = \sera{s}{0}_{d-1} = \sigs_{d-1}$, and the claim follows.
\end{proof}

Finally, we now prove the correctness of the FOLD procedure, as formulated by Claim~\ref{clm:FOLDisOptimal}.

\begin{claim}
\label{clm:FOLDisOptimal}
Let $\B$ be an $l$-block collection and let $n \geq 0$ be an integer. FOLD$(\B, n)$ returns a folding $\B'$ of $\B$ of size $n$ if such a folding exists, and otherwise it returns ``FAILED''. In addition, for every $l$-block collection $\B^*$ such that $|\B| = |\B^*|$ and $\sigs(\B) \leq \sigs(\B^*)$, if there is a folding $\B'^*$ of $\B^*$ of size $n$, then
FOLD$(\B, n)$ returns a collection $\B'$ such that
$\sigs(\B') \leq \sigs(\B'^*)$.

\end{claim}

\begin{proof}
Claim~\ref{clm:FoldAux} proves the first statement in Claim~\ref{clm:FOLDisOptimal}, thus it remains to show that for every $l$-block collection $\B^*$ such that $|\B| = |\B^*|$ and $\sigs \leq \sera{s}{*}$, if there is a folding $\B'^*$ of $\B^*$ of size $n$, then
FOLD$(\B, n)$ returns a collection $\B'$ such that
$\sert{s} \leq \sert{s}^*$.

First, note that when $n = |\B| = |\B^*|$, then in particular $\B^*$ and $\B$ are minimum signature $n$-size foldings of $\B^*$ and $\B$, respectively (Claim~\ref{clm:foldIncreaseSignature}), and thus $\B \leq^s \B'^*$ for every $n$-size folding $\B'^*$ of $\B^*$. Since in this case FOLD$(\B, n)$ returns $\B$, the claim follows.
Otherwise, $n \neq |\B|$, and we first show that FOLD$(\B, n)$ returns a folding $\B'$ of $\B$ of size $n$ if such a folding exists, and otherwise it returns ``FAILED''.

In the reminder of this proof we assume that $n \neq |\B| = |\B^*|$, and note that $d' = \dig{|\B| - n} = \dig{|\B^*| - n}$.
Since $\sigs \leq \sera{s}{*}$, either $\sigs = \sera{s}{*}$, or $\sigs < \sera{s}{*}$ and there is an integer $i$ such that $\sigs_{i-1} = \sera{s}{*}_{i-1}$ and $s_{i} < s^*_{i}$.

We first show that if is a folding $\B'^*$ of $\B^*$ of size $n$, FOLD$(\B, n)$ returns a folding $\B'$ of $\B$ of size $n$ satisfying $\sert{s} \leq \sert{s}^*$.
In this case, Claim~\ref{clm:required} states that there is an integer $0 \leq d^* \leq d'$ such that $\sert{s}^*_{d^*-1} = \sera{s}{*}_{d^*-1}$, $s'^*_{d^*} \geq s^*_{d^*}+1$, and $n \geq 2^{d^*}\max(s^*_{d^*}+1, 0) + \Delta^*_{d^*}$.
Consider two cases: (1) $\sigs_{d^*-1} = \sera{s}{*}_{d^*-1}$, which occurs when $\sigs = \sera{s}{*}$ or when $\sigs < \sera{s}{*}$ and $i \geq d^*$, and (2) $\sigs_{d^*-1} < \sera{s}{*}_{d^*-1}$, which occurs when $\sigs < \sera{s}{*}$ and $i < d^*$.

\vspace{10pt}
\noindent
(1) $\sigs_{d^*-1} = \sera{s}{*}_{d^*-1}$. In this case, $n \geq 2^{d^*}\max(s^*_{d^*}+1, 0) + \Delta^*_{d^*} \geq 2^{d^*}\max(s_{d^*}+1, 0) + \Delta_{d^*}$. Thus,
 when executing FOLD$(\B, n)$, the condition in line~4 is met and the algorithm does not return ``FAILED''. From Claim~\ref{clm:FoldAux}, FOLD$(\B, n)$ returns an $n$-size folding $\B'$ of $\B$, such that for the maximum integer $0 \leq d \leq d'$ for which $n \geq 2^d\max(s_{d}+1, 0) + \Delta_d$ we have that $\sert{s}_{d-1} = \sigs_{d-1}$ and $r' \leq d+1$, and if $r' = d+1$ then $s'_d = s_d + 1$ in case $d = d'$ and $s'_d = s_d + 2$ in case $d < d'$. By selection, $d^* \leq d \leq d'$. If $d^* < d$, then $\sert{s}_{d^*} = \sigs_{d^*} = \sera{s}{*}_{d^*} < \sert{s}^{*}_{d^*}$, and in particular $\sert{s} \leq \sert{s}^*$ and the claim follows. If $d^* = d$, then $\sert{s}_{d-1} = \sigs_{d-1} = \sera{s}{*}_{d-1} = \sert{s}^{*}_{d-1}$. If $r' < d+1$ then $\sert{s} \leq \sert{s}^*$ from Claim~\ref{clm:eqDeltar}, and the claim follows. If $r' = d+1$, then $s'_d = s_d + 1$ in case $d = d'$ and $s'_d = s_d + 2$ in case $d < d'$. In addition, from Claim~\ref{clm:required}, $s'^*_d \geq s_d + 1$ in case $d = d'$ and $s'^*_d \geq s_d + 2$ in case $d < d'$, thus in both cases $s'_d \leq s'^*_d$. If $s'_d < s'^*_d$ then $\sert{s}_d < \sert{s}^*_d$, and in particular $\sert{s} < \sert{s}^*$ and the claim follows. If $s'_d = s'^*_d$ then $\sert{s}_d = \sert{s}^*_d$, and from Claim~\ref{clm:eqDeltar} $\sert{s} \leq \sert{s}^*$ and the claim follows.

\vspace{10pt}
\noindent
(2) $\sigs_{d^*-1} < \sera{s}{*}_{d^*-1}$. In this case, for $i < d^*$ we have that $\sigs_{i-1} = \sera{s}{*}_{i-1}$ and $s_i < s^*_i$. Now, $n \geq 2^{d^*}\max(s^*_{d^*}+1, 0) + \Delta^*_{d^*} \geq \Delta^*_{d^*} \geq \Delta^*_{i+1} = 2^i \abs(s^*_i) + \Delta^*_{i} \geq 2^i \max(s^*_i, 0) + \Delta^*_{i}$. Similarly as before, Claims~\ref{clm:eqDeltar} and~\ref{clm:FoldAux} can be applied to show that $\sert{s} \leq \sert{s}^*$ .

\end{proof}

\subsection{Correctness of Algorithm SEARCH-BFB}
Assuming there is a BFB string $\alpha^*$ admitting the algorithm's input count vector $\vec{n}$, the BFB palindrome $\beta^* = \alpha^* \rev{\alpha}^*$ admits the count vector $2\vec{n}$. Let $\B^{*k+1} = \emptyset, \B^{*k}, \B^{*k-1}, \ldots, \B^{*1}$ be the block collection series corresponding to the layers of $\beta^*$ as described in the main manuscript. Since $\B^{k+1} = \B^{*k+1} = \emptyset$ ($\B^{k+1}$ is initialized in line~1 of Algorithm SEARCH-BFB), we have that $\ser{s}(B^{k+1}) = \sigs(\B^{*k+1})$. Assume that for some $0 \leq l \leq k$ we have that $\ser{s}(B^{l+1}) \leq \sigs(\B^{*l+1})$. Recall that the collection $\B^{*l}$ is obtained by the wrapping of some folding $\B'^{*}$ of size $n_l$ of $\B^{*l+1}$. Since the wrapping operation does not change element multiplicities and top values, it follows that $\sigs(\B^{*l}) = \sigs(\B'^{*})$.
From Claim~\ref{clm:FOLDisOptimal}, the application of the FOLD procedure in the $l$-th iteration of the loop in lines~2-4 of the algorithm returns a folding $\B'$ of $\B^{l+1}$ of size $n_l$, where $\sigs(\B^l) = \sigs(\B') \leq \sigs(\B'^*) = \sigs(\B^{*l})$. Inductively, the algorithm does not return ``FAILED'' in each one of the loop iterations, and after the last iteration $\ser{s}(\B^1) \leq \sigs(\B^{*1})$. From the same arguments as above and since $\B^{*1}$ can be folded into the single palindrome $\beta^*$, it follows that the application of FOLD in line~4 of the algorithm does not fail, and returns a collection containing a single palindrome $\beta  = \alpha\rev{\alpha}$, where $\alpha$ is a BFB string admitting $\vec{n}(\alpha) = \vec{n}$.

For the other direction of the proof, assume that the BFB algorithm returned the string $\alpha$. In this case, the series of collections $\B^{k+1}, \B^{k}, \ldots, \B^{1}$ satisfies that each collection $\B^l$ is an $l$-block collection of size $n_l$ and is obtained by folding and wrapping of the preceding collection in the series $\B^{l+1}$. The final collection $\B^1$ is folded into a single BFB palindrome $\beta = \alpha \rev{\alpha}$ admitting the count vector $2\vec{n}$, and therefore $\alpha$ is a BFB string admitting $\vec{n}$.

\subsection{Time Complexity of Algorithm SEARCH-BFB}

\subsubsection{Object Representation}
The algorithm handles two types of objects: BFB palindromes, and BFB palindrome collections.
BFB palindromes are further divided into three subtypes, who are implemented separately: empty palindromes, $l$-blocks, and composite $l$-BFB palindromes of the form $\beta \gamma \beta$ (see Claim 1 in the main paper). Each BFB palindrome object contains a filed maintaining the top value of the represented palindrome, allowing $O(1)$ time queries of this value. An empty palindrome is represented by an object containing only the top value field (which always holds the value 0), and generating new such objects take $O(1)$ time. An $l$-block is implemented as an object containing, in addition to the top-value field, a pointer to its internal $(l+1)$-BFB palindrome. Given a pointer to the internal $(l+1)$-BFB palindrome, generating new $l$-block objects take $O(1)$ time by copying the pointer, and setting the top value field to the top value of the pointed $(l+1)$-BFB palindrome. A composite $l$-BFB palindrome $\beta \gamma \beta$ is implemented by specifying a pointer to the $l$-BFB palindrome $\beta$, and a list of $l$-BFB palindromes $\alpha_1, \alpha_2, \ldots, \alpha_p$ representing the convexed $l$-collection $\A$ such that $\gamma = \gamma_{\A}$. Composite palindromes can be generated in a time proportional to the order of their internal convexed $l$-collection (where the top value field is set to be the top value of $\beta$).

A collection $\B = \{m_1 \beta_1, m_2 \times \beta_2, \ldots, m_q \times \beta_q\}$ is implemented by an object containing a field which maintains the size $|\B|$ of the collection, and two doubly linked lists maintaining the prefixes $\Bls_{r-1}$ and $\Bhs_{r-1}$ of the series $\Bls$ and $\Bhs$ in the decomposition of $\B$, where $r = r(\B)$. Note that for $i \geq r$, $\Bl_i = \Bh_i = \emptyset$. Each element $\Bl_i$ or $\Bh_i$ is implemented as
a linked list of $l$-BFB palindromes ordered with nondecreasing top values (it is possible that an $\Bh_i$ list contains multiple repeats of identical elements). Thus, computing $\mint(\Bl_i)$ or $\mint(\Bh_i)$ and extracting minimal elements from such lists is done in $O(1)$ time. Generating an empty collection is done in $O(1)$ time (where the two lists $\Bl_{r-1}$ and $\Bh_{r-1}$ are empty), and duplicating or wrapping a collection $\B$ take at most $O(|\B|)$ time (note that $r-1 \leq \log |\B|$, since an element $\beta \in \B_{r-1}$ corresponds to $2^{r-1}$ repeats of $\beta$ in $\B$, and that the total number of elements in all lists $\Bl_i$ and $\Bh_i$ is at most $|\B|$).

\subsubsection{Type II Elementary Folding}
\label{sec:typeIITime}
Using the object representation described above, for a collection $\B$ such that $\varepsilon \notin \B$ and an integer $m > 0$, it is possible to compute a type II elementary folding $\B' = \B + m \varepsilon$ in $O(|\B| + m)$ time as follows.
First, the number $d = \dig{m}$ is computed.
Note that $d \leq \log m$ ($d$ can be defined as the index of the least significant bit different from 0 in the binary representation of $m$), and may be computed in $O(\log m)$ time.
$\B'$ is initialized by copying $\B$, i.e. generating the list $\Bls'_{r-1}$ and $\Bhs'_{r-1}$ (in $O(|\B|)$ time). Then, if $d \geq r$, empty collections $\Bl'_i$ and $\Bh'_i$ are added to the prefixes of $\Bls'$ and $\Bhs'$ for $r \leq i \leq d$, and a single $\varepsilon$ element is added to $\Bl'_d$. Else, $d < r$, and a single $\varepsilon$ element is added as the first element in $\Bl'_d$ (being of minimum top value among all elements in the list), and elements from collections $\Bl'_i$ and $\Bh'_{i}$ for $i > d$ are moved into $\Bh'_d$. This latter modification is performed by first merging each $\Bl'_i$ and $\Bh'_i$ lists for $i > d$ to a single list ordered with nondecreasing top values (in a linear time with respect to the number of elements in the two lists), and then, with increasing index $i$, each merged list is added to the beginning of $\Bh'_d$, where $2^{i-d}$ repeats of each element in the merged list of $\Bl'_i$ and $\Bh'_i$ are added to $\Bh'_d$. In both cases where $d \geq r$ or $d < r$, it is possible to assert the modification updates properly the representation of $\B'$ to represent the collection $\B + 2^d \varepsilon$, that $r(\B') = d+1$, and that total time required for the modification is at most $O(|\B| + d) = O(|\B| + \log m)$. Finally, additional $\frac{m}{2^d}-1$ repeats of $\varepsilon$ are added to $\Bh'_d$ in $O(m)$ time, where now it is possible to assert that $\B'$ properly represents the collection $\B + m \varepsilon$, and that the total computation time is $O(|\B| + m)$.

\subsubsection{Right-folding}
\label{sec:rightFoldTime}
In order to right-fold a collection $\B$, the algorithm first gets pointers to the elements $\Bl_{r-1}$ and $\Bh_{r-1}$, in $O(r)$ time for $r = r(\B)$. Then, it starts traversing these lists backward for $i = r-1$ down to $g$, inclusive, where $g$ is the first encountered index such that $\Bh_g \neq \emptyset$. For each such $i$, the algorithm extracts the first (minimal) element in the list $\Bl_i$, and accumulates these elements in a list $\A$. Finally, the algorithm extracts two copies of the minimal element $\beta$ in $\Bh_g$, and uses $\beta$ and $\A$ to construct the BFB palindrome $\alpha = \beta \gamma_{\A} \beta$. Then, $\alpha$ is inserted into $\Bl_g$. As this procedure takes $O(r)$ time and decreases the size of the collection by $2^r$, any valid consecutive application of right-foldings over $\B$ takes at most $O(|\B|)$ time.

\subsubsection{The FOLD Procedure}
\label{sec:FOLDTime}
Consider the application of the FOLD procedure on a collection $\B = \{m_1 \beta_1, m_2 \times \beta_2, \ldots, m_q \times \beta_q\}$ and an integer $n \geq 0$. If $n \geq |B|$, the procedure applies in line~1 a type II elementary folding in $O(|\B| + n)$ time (Section~\ref{sec:typeIITime}) and hults. Otherwise, given the series $\Bls_{r-1}$ and $\Bhs_{r-1}$, it is possible to compute $\sigs_{r}$ and $\ser{\Delta}_{r+1}$ in $O(r) = O(\log(|\B|))$ time. Note that $s_i = 0$ for $i > r$, and $\Delta_{i} = \Delta_{r+1}$ for $i > r+1$. The number $\dig{|\B| - n}$ satisfies $\dig{|\B| - n} \leq \max(\log(|\B|) + \log(n))$.
After computing $\sigs_{r}$ and $\ser{\Delta}_{r+1}$,
checking the condition in line~4, as well as computing the parameter $d$ in line~$5$, can be done in $O(\dig{|\B| - n})$ time. The two type II elementary foldings in lines~5 and~8 take $O(|\B| + n)$ time (Section~\ref{sec:typeIITime}), and the total time for right-folding applications in the loop in lines~6-7 is $O(|\B|)$ (Section~\ref{sec:rightFoldTime}). Thus, the total running time of the procedure is $O(|\B| + n)$.

\subsubsection{Overall Running Time}
Let $\vec{n} = [n_1, n_2, \ldots, n_k]$ be the input vector for the algorithm. Denote $N = \displaystyle{\sum_{1 \leq l \leq k} n_l}$, and note that $N$ is the length of the output string $\alpha$ in case the algorithm does not return ``FAILED''.
It is simple to assert that besides operations conducted within the FOLD procedure, Algorithm SEARCH-BFB performs $O(N)$ operations. For every $1 \leq l \leq k$, FOLD is called once by the BFB algorithm over the collection $\B^{l+1}$ of size $n_{l+1}$ and the integer $n_l$, and runs in $O(n_{l+1} + n_l)$ time (Section~\ref{sec:FOLDTime}). Summing the running time of FOLD for $l = k$ down to 1, its overall running time, as well as the overall running time of Algorithm SEARCH-BFB, is $O(N)$.

\section{The Decision Variant}
\label{sec:decision}

In this section, we describe a simplification of the SEARCH-BFB algorithm which solves the decision variant of the BFB count vector problem. Essentially, this algorithm applies similar steps to those of the search algorithm, yet instead of explicitly maintaining collections $\B$, the algorithm only maintains the signature $\sigs$ of $\B$. We assume that the algorithm maintains explicitly only the prefix $\sigs_r$ of $\sigs$ (for $r = r(\B)$) as a linked list, where for $i > r$ the algorithm takes the value 0 whenever it needs using the value $s_i$.

\begin{algorithm}
\normalsize
	\DontPrintSemicolon
	\footnotesize
	\SetKwInput{Algorithm}{Algorithm}
	\SetKwInput{Procedure}{Procedure}
	\SetKwInput{Input}{Input}
	\SetKwInput{Output}{Output}
	\SetKwIF{If}{ElseIf}{Else}{If}{then}{Else if}{Else}{endif}
	\SetKwFor{For}{For}{do}{endfor}
	\SetKwFor{While}{While}{do}{endw}
	\Algorithm{DECISION-BFB$(\vec{n})$}
	\BlankLine
	\Input{A count vector $\vec{n} = \left[n_1, n_2, \ldots, n_k\right]$.}
	\Output{``TRUE'' if $\vec{n}$ is a BFB count vector, and ``FAILED'' if otherwise.}
	\BlankLine
	Set $n_{k+1} \leftarrow 0$ and $\sera{s}{k+1} \leftarrow \ser{0}$.\;
  \For{$l \leftarrow k$ \textbf{down to} 1}{
  	Apply SIGNATURE-FOLD$(\sera{s}{l+1}, n_{l+1}, n_l)$. If this operation has failed, \textbf{return} ``FALSE''.\;
  	Otherwise, set $\sera{s}{l}$ to be the returned value from the call to SIGNATURE-FOLD$(\sera{s}{l+1}, n_{l+1}, n_l)$.\;
  }
 	Apply SIGNATURE-FOLD$(\sera{s}{1}, n_1, 1)$. If this operation has failed, \textbf{return} ``FALSE'', and otherwise \textbf{return} ``TRUE''.\;
	\label{alg:DECISION-BFB}
\end{algorithm}

\begin{algorithm}[H]
\normalsize
	\DontPrintSemicolon
	\footnotesize
	\SetKwInput{Algorithm}{Algorithm}
	\SetKwInput{Procedure}{Procedure}
	\SetKwInput{Input}{Input}
	\SetKwInput{Output}{Output}
	\SetKwIF{If}{ElseIf}{Else}{If}{then}{Else if}{Else}{endif}
	\SetKwFor{For}{For}{do}{endfor}
	\SetKwFor{While}{While}{do}{endw}

	\Procedure{SIGNATURE-FOLD$(\sigs, n_\B, n)$}
	\BlankLine
	\Input{The signature $\sigs$ and size $n_\B$ of an $l$-block collection $\B$ and an integer $n \geq 0$.}
	\Output{The minimum signature $\sert{s}$ of a folding $\B'$ of $\B$ such that $|\B'| = n$, or the string ``FAILD'' if there is no such $\B'$.}
	\BlankLine
	\If{$n_\B \leq n$}{
		\Return{ADD-EMPTY$(\sigs, n_\B, n - n_\B)$}.\;
	}
	\Else{
	Compute the prefix $\ser{\Delta}_{\dig{n_\B - n}}$ of $\ser{\Delta}(\B)$.\;
	\If{there exists $0 \leq d \leq \dig{n_\B - n}$ such that $n \geq 2^d\max(s_d+1, 0) + \Delta_d$}{
		Let $d$ be the maximum integer sustaining the condition above.\;
		Set $\sert{s} \leftarrow \text{ADD-EMPTY}(\sigs, n_\B, 2^d)$, and $n_{\B'} \leftarrow n_\B + 2^d$. \;
		\If{$n \geq n_{\B'} + 2^{d+1}s'_{d+1}$}{
			Set $s'_{d+1} \leftarrow s'_{d+1} + \left\lceil \frac{n_{\B'} - n}{2^{d+1}}\right\rceil$.\;
			Set $n_{\B'} \leftarrow \Delta_d + 2^d\abs(s'_d) + 2^{d+1}\abs(s'_{d+1})$.\;
			Set $\sert{s} \leftarrow \text{ADD-EMPTY}(\sert{s}, n_{\B'}, n - n_{\B'})$.\;
		}
		\Else{
			Set $s'_{d} \leftarrow s'_{d} + \frac{n_{\B'} - n}{2^{d}} + 2 s'_{d+1}$.\;
			Set $s'_{d+1} \leftarrow 0$.\;
		}
		\Return{$\sert{s}$}.\;
	}
	\lElse{\Return{``FAILED''}\;}
	}
	\label{alg:DECISION-FOLD}
\end{algorithm}

\begin{algorithm}[H]
\normalsize
	\DontPrintSemicolon
	\footnotesize
	\SetKwInput{Algorithm}{Algorithm}
	\SetKwInput{Procedure}{Procedure}
	\SetKwInput{Input}{Input}
	\SetKwInput{Output}{Output}
	\SetKwIF{If}{ElseIf}{Else}{If}{then}{Else if}{Else}{endif}
	\SetKwFor{For}{For}{do}{endfor}
	\SetKwFor{While}{While}{do}{endw}

	\Procedure{ADD-EMPTY$(\sigs, n_\B, m)$}
	\BlankLine
	\Input{The signature $\sigs$ and size $n_\B$ of an $l$-BFB palindrome collection $\B$ containing no $\varepsilon$ elements, and an integer $m \geq 0$.}
	\Output{The signature $\sert{s}$ of the folding $\B' = \B + m \varepsilon$ of $\B$.}
	\BlankLine
		\If{$n_\B = m$}{\Return{$\ser{s}$}}
		\Else{
			Let $d = \dig{n - n_\B}$, and set the prefix $\sert{s}_{d-1}$ to be the copy of the prefix $\sert{s}_{d-1}$ of $\sigs$.\;
			Set $s'_{d} \leftarrow s_d + 1$.\;
			Compute $\sert{\Delta}_{d} = \displaystyle{\sum_{0 \leq i < d} 2^i \abs(s_i)}$.\;
			Set $s'_{d+1} \leftarrow \frac{\ser{\Delta}_{d} + 2^d \abs(s'_d) - n}{2^{d+1}}$.
			\tcp*[h]{All values $s'_i$ for $i > d+1$ are implicitly set to 0.}\;
			\Return{$\sert{s}$}.\;
		}
	\label{alg:ADD-EMPTY}
\end{algorithm}

The fact that the signature modifications applied by Procedure SIGNATURE-FOLD yield identical signatures to those of the collections computed by Procedure FOLD can be asserted from Conclusion~\ref{conc:r} and Claims~\ref{clm:addEpsilons} and~\ref{clm:rightFoldingSeq}. It may also be asserted that the total number of operations in all calls to Procedure ADD-EMPTY (lines~2,~7, and~11 in Procedure SIGNATURE-FOLD), as well as the computation of $\ser{\Delta}_{\dig{n_\B - n}}$ in line~4, checking the condition in line~5, and computing $d$ in line~6, is $O(r(\B) + r(\B')) = O(\log n_\B + \log n)$. Besides these operations, Procedure SIGNATURE-FOLD applies additional $O(1)$ operations, hence its total running time is $O(\log n_\B + \log n)$. Therefore, the overall running time of Algorithm DECISION-BFB is $\displaystyle{O\left(\sum_{0 \leq l \leq k} (\log n_{l+1} + \log n_l)\right) = O\left(\sum_{0 \leq l \leq k} \log n_l\right) = O(\tilde{N})}$, where $\tilde{N}$ is the number of bits in the representation of the input vector $\vec{n}$. A more involved amortized analysis, omitted from this text, may show that the algorithm performs $O(\tilde{N})$ bit operations, hence being strictly linear with respect to its input length.

\section{The Distance Variant}

This section gives Algorithm DISTANCE-BFB for solving the distance variant of the BFB count vector problem. As a matter of fact, the presented algorithm solves the problem for every suffix $\sera{n}{l} = [n_l, n_{l+1}, \ldots, n_k]$ of the input vector $\vec{n} = [n_1, n_2, \ldots, n_k]$.

For a vector $\vec{n} = [n_1, n_2, \ldots, n_k]$ of length $k$ and an integer $m$, denote by $[m, \vec{n}]$ the $(k+1)$-length vector $[m, n_1, n_2, \ldots, n_k]$.
The algorithm is generic and may work with any vector distance measure $\delta$, provided that for
any equal-length three vectors $\vec{n}$, $\sert{n}$, and $\sert{n}'$ such that $\delta(\ser{n}, \sert{n}) \leq \delta(\ser{n}, \sert{n}')$, (1) $\delta(\sert{n}, \sert{n}) = \delta(\sert{n}', \sert{n}') = 0 \leq \delta(\ser{n}, \sert{n}) \leq \delta(\sert{n}, \sert{n}') \leq 1$, and (2) for any pair of numbers $m$ and $m'$, $\delta([m, \ser{n}], [m', \sert{n}]) \leq \delta([m, \ser{n}], [m', \sert{n}'])$.
For some precision parameter $0 \leq \eta < 1$, the algorithm finds the exact solution for the distance variant of the BFB count vector problem for every suffix of the input vector for which the solution is at most $\eta$, and returns the approximated solution 1 to suffixes for which the solution is greater than $\eta$.

Similarly to Algorithms SEARCH-BFB and DESCISION-BFB, Algorithm DISTANCE-BFB runs $k$ iterations on an input vector $\vec{n} = [n_1, n_2, \ldots, n_k]$, indexed from $k$ down to 1. At the end of iteration $l$, the algorithm computes a collection $S^l$ containing elements of the form $(\sera{n}{i} = [n^i_l, n^i_{l+1}, \ldots, n^i_k], \sera{s}{i})$, where $\sera{s}{i}$ is the minimum signature of an $l$-block collection $\B^i$ admitting the count vector $\sera{n}{i}$, and $\delta(\sera{n}{l}, \sera{n}{i}) \leq \eta$. It is guaranteed that for every BFB count vector $\sera{n}{j} = [n^j_l, n^j_{l+1}, \ldots, n^j_k]$ such that $\delta(\sera{n}{l}, \sera{n}{j}) \leq \eta$ and every $l$-block collection $\B^j$ admitting $\sera{n}{j}$, $S^l$ contains a pair $(\sera{n}{i}, \sera{s}{i})$ such that $\delta(\sera{n}{l}, \sera{n}{i}) \leq \delta(\sera{n}{l}, \sera{n}{j})$ and $\sera{s}{i} \leq \sera{s}{j}$.

\begin{algorithm}
\normalsize
	\DontPrintSemicolon
	\footnotesize
	\SetKwInput{Algorithm}{Algorithm}
	\SetKwInput{Procedure}{Procedure}
	\SetKwInput{Input}{Input}
	\SetKwInput{Output}{Output}
	\SetKwIF{If}{ElseIf}{Else}{If}{then}{Else if}{Else}{endif}
	\SetKwFor{For}{For}{do}{endfor}
	\SetKwFor{ForEach}{For each}{do}{endfor}
	\SetKwFor{While}{While}{do}{endw}
	\Algorithm{DISTANCE-BFB$(\vec{n}, \eta)$}
	\BlankLine
	\Input{A count vector $\vec{n} = [n_1, n_2, \ldots, n_k]$, and a precision parameter $0 \leq \eta < 1$.}
	\Output{For every $1 \leq l \leq k$, the algorithm reports the minimum distance $\delta_l$ of the suffix $\sera{n}{l} = [n_l, n_{l+1}, \ldots, n_k]$ of $\vec{n}$ from a BFB count vector, in case this distance is at most $\eta$.}
	\BlankLine
	\tcp {Collections of the form $S^l$ contain pairs $(\sert{n}^i = [n'_l, n'_{l+1}, \ldots, n'_k], \sera{s}{i})$, where $\sera{s}{i}$ is the minimum signature of an $l$-block collection $\B^i$ admitting the count vector $\sert{n}^i$, and $\delta(\sera{n}{l}, \sert{n}^i) \leq \eta$.}
	Set $S^{k+1}$ be a collection containing only the pair $(\vec{0}, 0)$.\;
  \For{$l \leftarrow k$ \textbf{down to} 1}{
		Set $\delta_l \leftarrow 1$.\;
		Set $S^{l} \leftarrow \emptyset$.\;
	  \ForEach{$(\sert{n}^i = [n'_{l+1}, \ldots, n'_k], \sera{s}{i}) \in S^{l+1}$}{
  		\ForEach{$n \geq 1$ such that $\delta(\sera{n}{l}, [n ,\sert{n}^i]) \leq \eta$}{
  			\If{SIGNATURE-FOLD$(\sert{n}^i, n'_{l+1}, n) = \sigs$, IS-PALINDROMIC$(\sigs)$, and for all $(\sert{n}^j, \sera{s}{j}) \in S^{l}$ such that $\delta(\sera{n}{l}, \sert{n}^j) \leq \delta(\sera{n}{l}, [n ,\sert{n}^i])$ it is true that $\sigs < \sert{n}^j$}{
	  			Add $([n, \sert{n}^i], \sigs)$ to $S^{l}$.\;
					Set $\delta_l \leftarrow \min(\delta_l, \delta(\sera{n}{l}, [n ,\sert{n}^i]))$.\;
  			}
  		}
 		}
	\textbf{Report} $\delta_l$.\;
  }
	\BlankLine
	\hrule
	\BlankLine
	\BlankLine
	\BlankLine
	\setcounter{AlgoLine}{0}
	\Procedure{IS-PALINDROMIC$(\sigs)$}
	\BlankLine
	\Input{The signature $\sigs$ of an $l$-BFB palindrome collection $\B$.}
	\Output{``TRUE'' if it is possible to concatenate all elements in $\B$ into a single $l$-BFB palindrome, and ``FALSE'' otherwise.}
	\BlankLine

	Compute the prefix $\ser{\Delta}_{r+1}$ of $\ser{\Delta}(\B)$ for $r = r(\B)$. \tcp*[h]{Note that $|\B| = \Delta_{r+1}$}\;
	\lIf{there exists $0 \leq d \leq \dig{\Delta_{r+1}-1}$ such that $1 \geq 2^d\max(s_d+1, 0) + \Delta_d$}{
		\textbf{return} ``TRUE''.\;
	}
	\lElse{\textbf{return} ``False''.}
	\label{alg:BFB_error}
\end{algorithm}

Consider the signature $\sigs$ of a collection $\B$ of size $n$. It is simple to show that $r(\B) \leq \log n + 1$, and that $-n < s_i \leq n$ for every $0 \leq i \leq r$. Therefore, $\sigs$ can be represented by $O(\log^2 n)$ bits, and so the number of different signatures of collections of size $n$ is upper bounded by $2^{O(\log^2 n)}$. In addition, under realistic assumptions, we may assume that the number of different values $n$ examined in line~6 of Algorithm DISTANCE-BFB bounded by $2^{O(\log^2 n_l)}$, since this number should approximate the count $n_l$ (for example, using the Poisson $\delta$ function described by the main manuscript, it is possible to show that for every value of $n_l$ and $\sert{n}^i$ and for $n \geq 20 n_l$, $\delta(\sera{n}{l}, [n , \sert{n}^i]) > 1-10^{-6}$, thus choosing $\eta = 1 - 10^{-6}$ guarantees that the loop in lines~6-9 is being executed less than $20n_l$ times for every $(\sert{n}^i, \sera{s}{i}) \in S^{l+1}$). Due to the condition in line~7, every possible signature $\sigs$ appears at most once in some pair in $S^l$, thus the size of $S^l$ is bounded by $2^{O(\log^2 n_l)}$. It is straightforward to observe that the total number of operations in the loop in lines~7-9 is also $2^{O(\log^2 n_l)}$, and so the total running time of the algorithm is bounded by $\displaystyle{\sum_{1 \leq l \leq k} 2^{O(\log^2 n_l)}} \leq 2^{O(\log^2 N)} = N^{O(\log N)}$.

\section{Chromosome simulation details}
Each chromosome pair was modeled as two sequences of 100,000,000 ordered bases.
Then fifty rearrangement were introduced to each chromosome independently. Each
rearrangement type was chosen randomly from deletion, inversion, and
duplication, according to a distribution. Thus, both balanced and unbalanced
rearrangements were used to simulate the chromosomes. If the chosen rearrangement was a
duplication, then it was decided whether or not the duplication would be tandem
and whether or not it would be inverted. Tandem duplications would be inserted
adjacent to the original chromosome segment, and inverted duplications would
have the new duplicated segment reversed with respect to the original segment.

Two rearrangement type regimes were used. In the first regime, referred to as
``evendup'' in the supplemental data, each rearrangement was
a duplication, inversion, or deletion with probability .5, .25, and .25
respectively. Duplications had a 50\% chance of being tandem and,
independently, a 50\% chance of being inverted. In the second regime, called
``highdup'' in the supplemental data, the
probability of duplication, inversion, and deletion were $\frac{7}{11}$,
$\frac{2}{11}$, and $\frac{1}{11}$. The probability of a duplication being
tandem or inverted was .9 and .9. This second regime was created because in the
first, fold-back inversions occur infrequently. The second regime allowed us
to examine tests for BFB when an alternative mechanism is creating many
fold-back inversions.

The size of each non-BFB rearrangement was chosen from a normal distribution
bounded at zero with mean 10,000 and a variance of 10,000,000. Rearrangements
were introduced sequentially in each chromosome. For chromsomes in which BFB
was simulated, consecutive rounds of BFB were introduced after one of the fifty
non-BFB rearrangments. The number of BFB rounds varied from two to ten. Each
BFB round consisted of a prefix of the chromosome undergoing a tandem inverted
duplication. The size of the prefix was selected from a normal distribution
with a mean of zero and a standard deviation of one tenth of the length of the
chromosome.

After each chromosome in the pair was rearranged, the copy numbers and
breakpoints were combined as one would expect from experimental evidence.

\section{Cancer cell line results}
We identified count vectors on three chromosomes from the 746 cancer cell lines that
were long and nearly consistent with BFB. The observed count vectors along with
the nearest count vector consistent with BFB are shown below.

\vspace{8pt}
\noindent
Cell line: AU565 \hspace{8em} Tissue: bone\\
Chromosome 8 between 72.5 MB and 80.0 MB\\
\texttt{Observed 4,8,14,10,8,14,9,13,7,12,9,7}\\
\texttt{Fit\hspace{2.5em} 4,8,14,10,8,14,9,13,7,\textcolor{red}{13},9,7}

\vspace{8pt}
\noindent
Cell line: PC-3 \hspace{8em} Tissue: prostate\\
Chromosome 10 between 60 MB and 82 MB\\
\texttt{Observed 6,10,14, 9,6,9,13,9,5,9,3,14}\\
\texttt{Fit\hspace{2.5em} 6,10,14,\textcolor{red}{10},6,9,13,9,5,9,3,\textcolor{red}{15}}

\vspace{8pt}
\noindent
Cell line: MG-63 \hspace{8em} Tissue: bone\\
Chromosome 8 between 112 MB and 121 MB\\
\texttt{Observed 10,6,8,14,11,14,9,8,13,9,13,9,7}\\
\texttt{Fit\hspace{2.5em} 10,6,8,14,11,\textcolor{red}{15},9,\textcolor{red}{9},13,9,13,9,7}\\

\newpage
\section{ROC curves for varying simulation parameters}
Below are the ROC curves, similar to those in Figure 4 of the main paper,
for many different simulation and test parameters.

\begin{figure}[h!]
    \centering
    \includegraphics[width=.9\columnwidth]{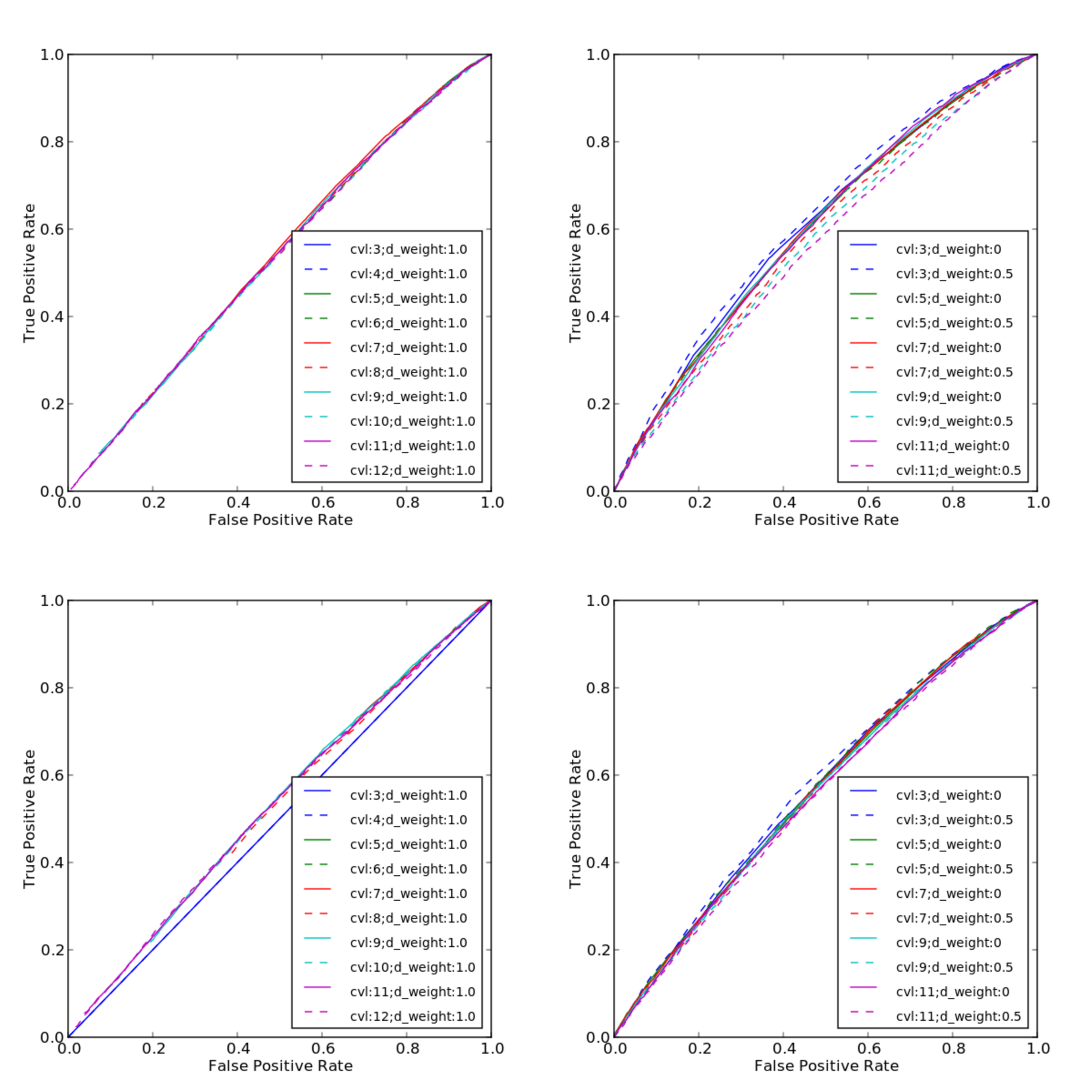}
    \caption{{\normalsize ROC curves for simulations with 2 rounds of BFB.
    Clockwise from the upper left, evendup background with no use of fold-back
    fraction, evendup background using fold-backs, highdup background using
    fold-backs, highdup background with no use of fold-back fraction.}}
\end{figure}
\afterpage{\clearpage}
\begin{figure}[h!]
    \centering
    \includegraphics[width=.9\columnwidth]{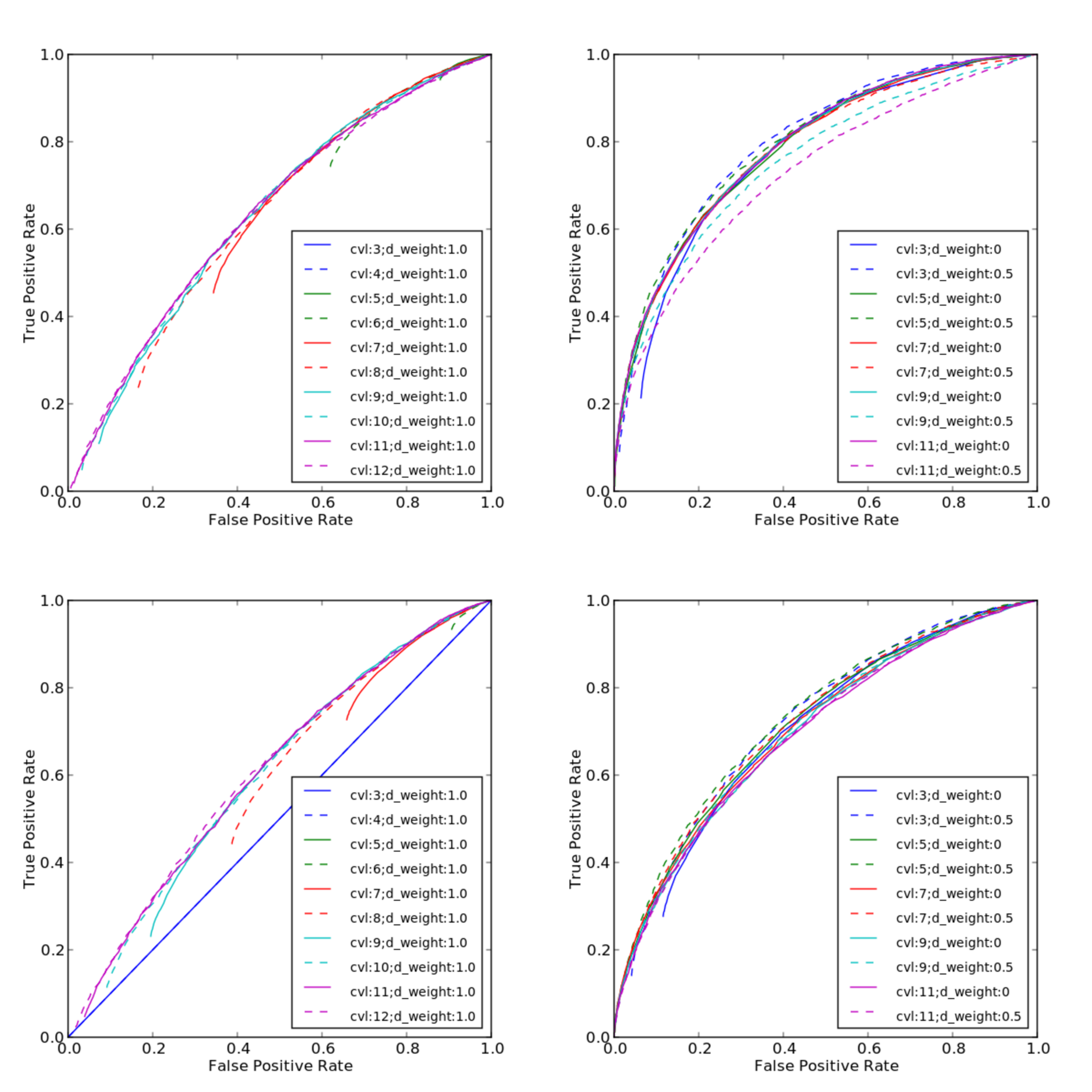}
    \caption{{\normalsize ROC curves for simulations with 4 rounds of BFB.
    Clockwise from the upper left, evendup background with no use of fold-back
    fraction, evendup background using fold-backs, highdup background using
    fold-backs, highdup background with no use of fold-back fraction.}}
\end{figure}
\afterpage{\clearpage}
\begin{figure}[h!]
    \centering
    \includegraphics[width=.9\columnwidth]{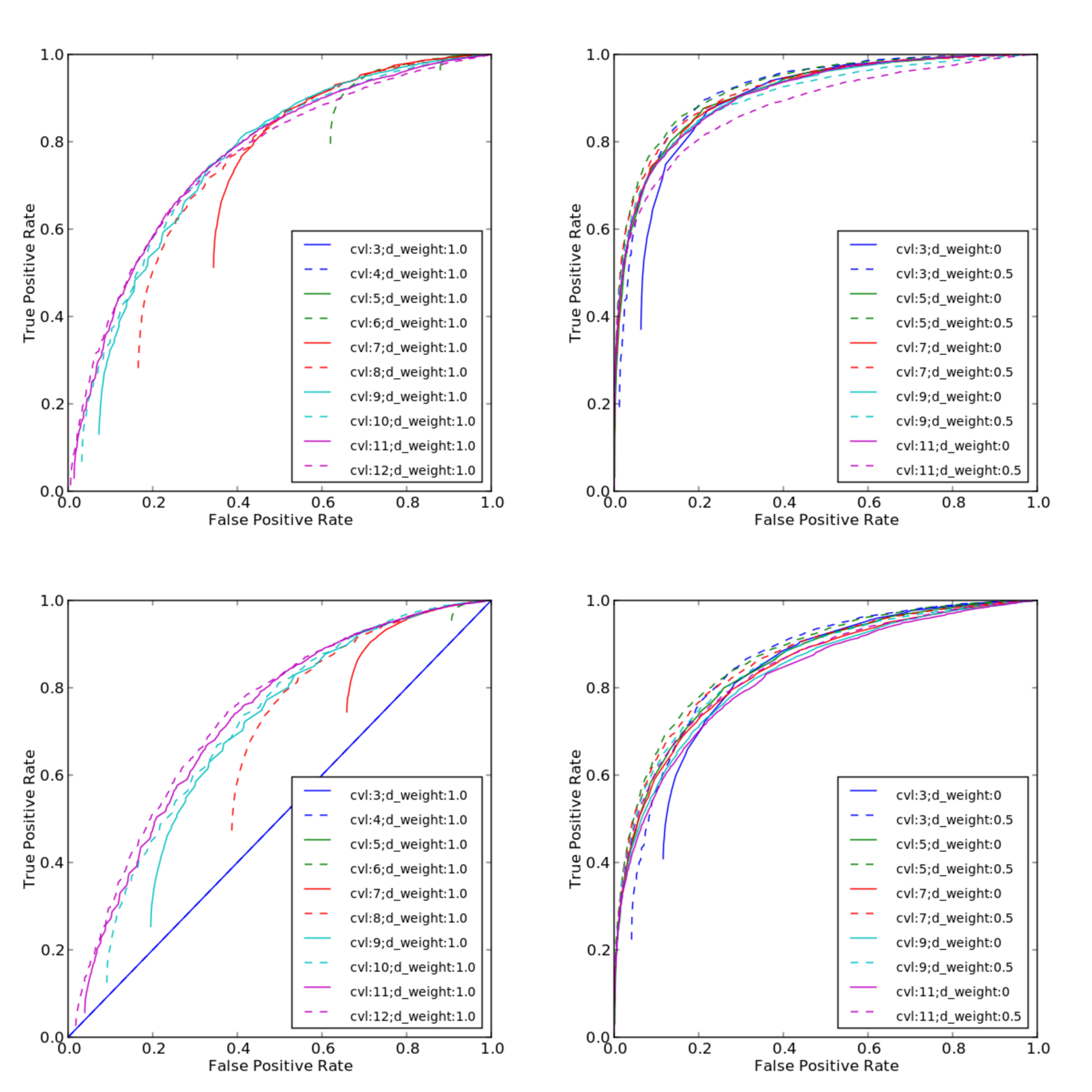}
    \caption{{\normalsize ROC curves for simulations with 6 rounds of BFB.
    Clockwise from the upper left, evendup background with no use of fold-back
    fraction, evendup background using fold-backs, highdup background using
    fold-backs, highdup background with no use of fold-back fraction.}}
\end{figure}
\afterpage{\clearpage}
\begin{figure}[h!]
    \centering
    \includegraphics[width=.9\columnwidth]{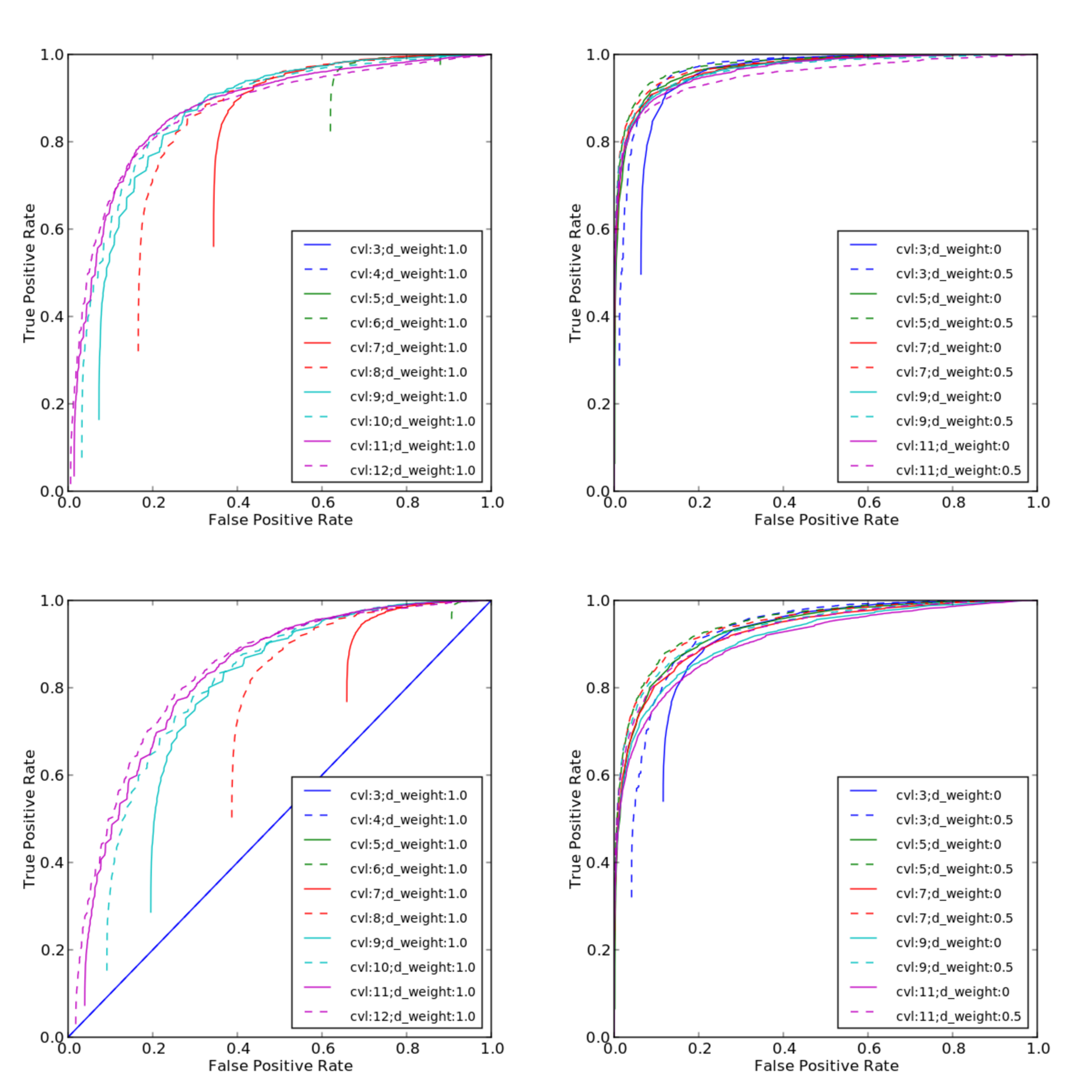}
    \caption{{\normalsize ROC curves for simulations with 8 rounds of BFB.
    Clockwise from the upper left, evendup background with no use of fold-back
    fraction, evendup background using fold-backs, highdup background using
    fold-backs, highdup background with no use of fold-back fraction.}}
\end{figure}
\afterpage{\clearpage}
\begin{figure}[h!]
    \centering
    \includegraphics[width=.9\columnwidth]{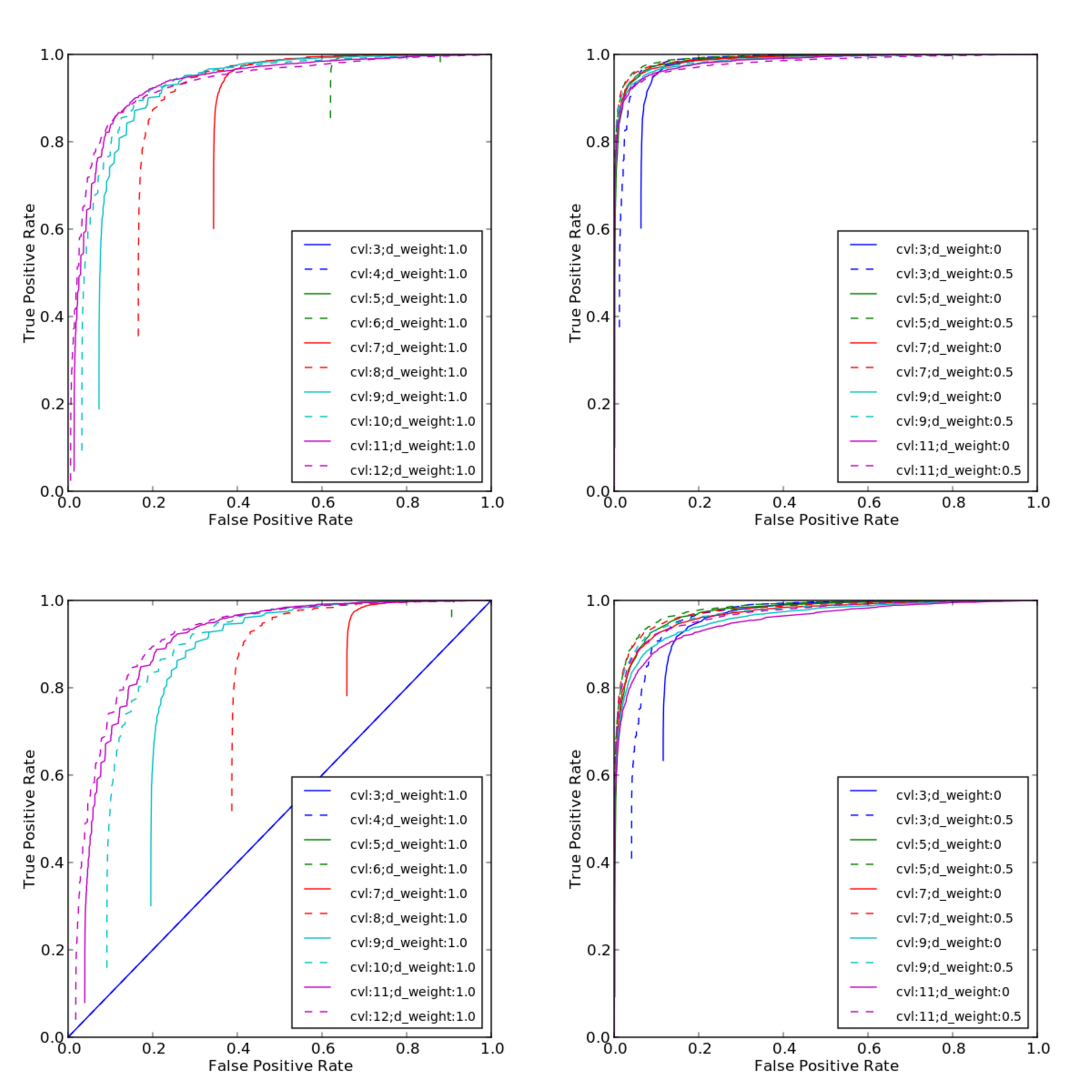}
    \caption{{\normalsize ROC curves for simulations with 10 rounds of BFB.
    Clockwise from the upper left, evendup background with no use of fold-back
    fraction, evendup background using fold-backs, highdup background using
    fold-backs, highdup background with no use of fold-back fraction.}}
\end{figure}
\afterpage{\clearpage}

\newpage
\section{Pancreatic cancer data analysis pipeline}
Figure~\ref{fig:pipeline} shows a graphical layout of the analysis.
\begin{figure}[b!]
    \centering
    \includegraphics[trim =0mm 0mm 0mm 0mm, clip, width=.5\columnwidth]{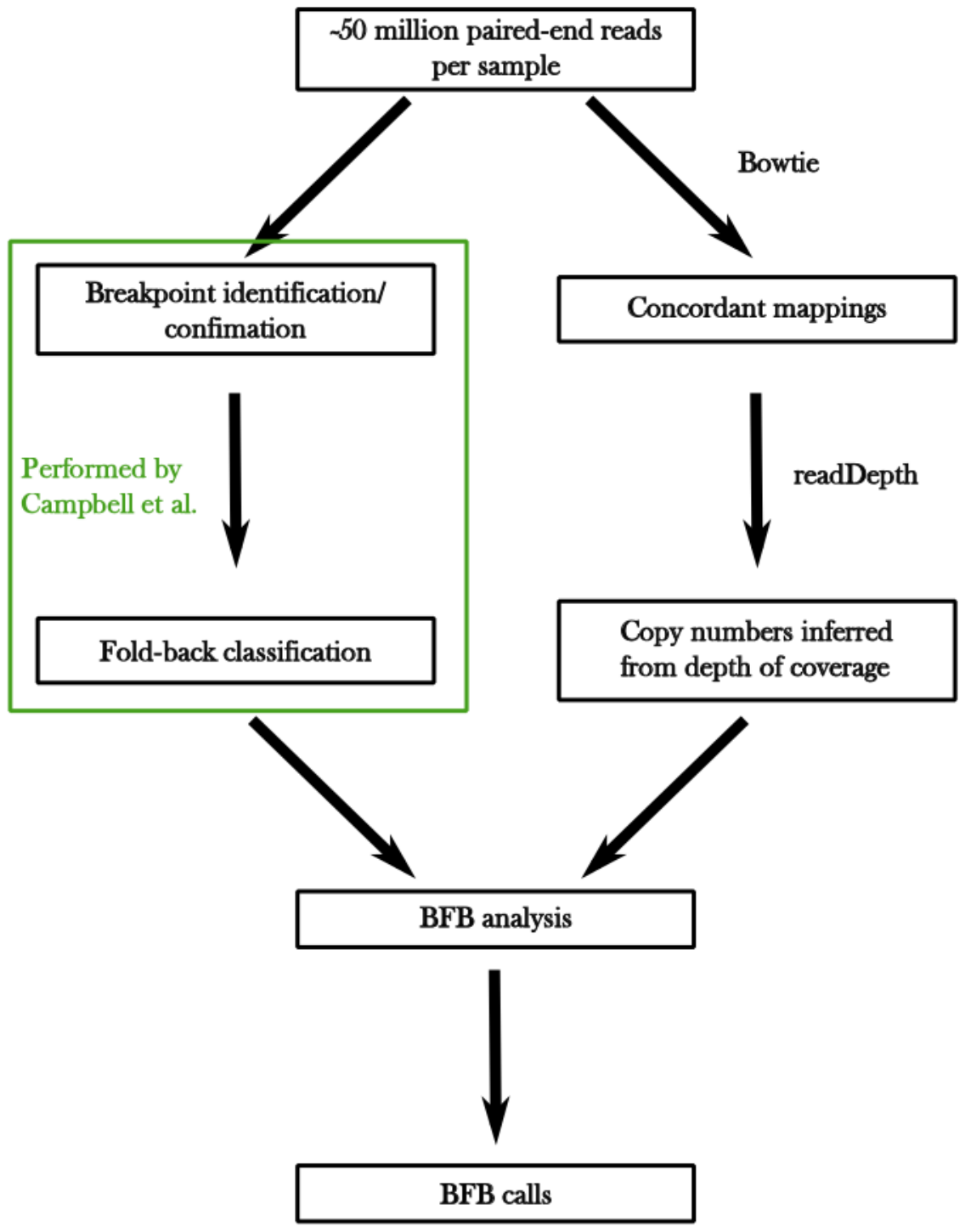}
    \caption{Graphical representation of the analysis performed
    with the pancreatic cancer paired-end sequencing data.\label{fig:pipeline}}
\end{figure}
\clearpage

\section{Possible arrangement of segments on BFB-rearranged chromosome 12}

\begin{figure}[b!]
    \centering
    \includegraphics[width=.8\columnwidth]{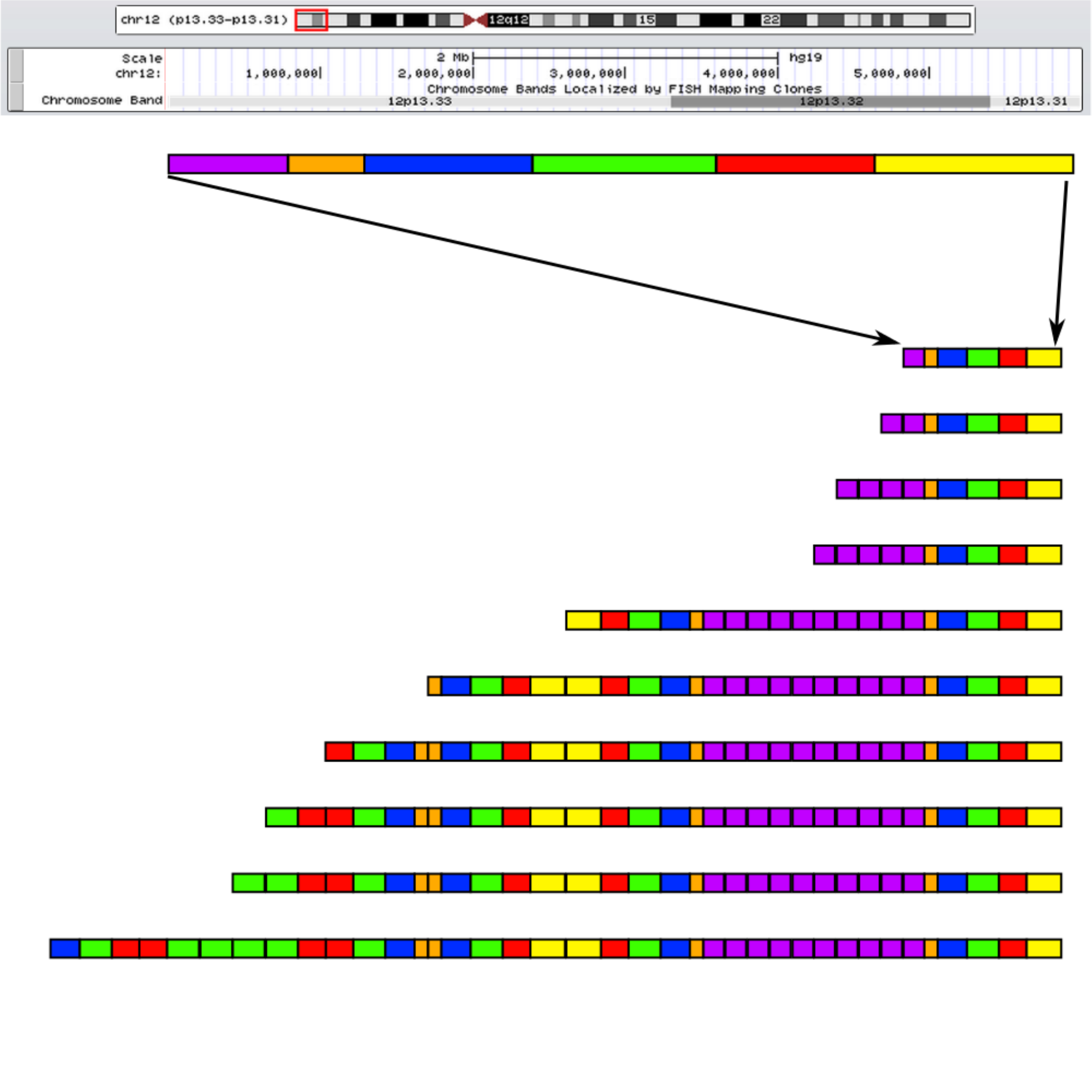}
    \caption{Plausible BFB cycles that could lead to the copy counts observed
    in chromosome 12 of pancreatic cancer sample PD3641.}
\end{figure}
\afterpage{\clearpage}

\bibliographystyle{pnas2009}